\documentclass[a4paper,UKenglish,cleveref,autoref,thm-restate
,authorcolumns]{lipics-v2021}

\graphicspath{{figures/}}
\usepackage{tikz}
\newcommand{\tikzcircle}[2][red,fill=red]{\tikz[baseline=-0.5ex]\draw[#1,radius=#2] (0,0) circle ;}%
\definecolor{defred}{RGB}{139,0,0}
\definecolor{defgray}{RGB}{169,169,169}
\definecolor{defblue}{RGB}{0,0,139}
\definecolor{defpurple}{RGB}{139,0,139}
\definecolor{deforange}{RGB}{255,140,0}
\definecolor{defgreen}{RGB}{0,100,0}
\definecolor{defcyan}{RGB}{0,139,139}
\definecolor{defpink}{RGB}{255,105,180}

\usepackage{booktabs,tabularx,multirow,makecell}
\usepackage{bm}

\title{Stack and Queue Layouts with Defects}

\acknowledgements{Research started at the Bertinoro Workshop on Graph Drawing 2024. Michael A. Bekos and Maria Eleni Pavlidi were supported by HFRI Grant No: 26320.}

\author{Michael A. Bekos}{University of Ioannina, Greece \and \url{https://myweb.uoi.gr/bekos/} }{bekos@uoi.gr}{https://orcid.org/0000-0002-3414-7444}{}

\author{Carla Binucci}{University of Perugia, Italy \and \url{https://mozart.diei.unipg.it/binucci/} }{carla.binucci@unipg.it}{https://orcid.org./0000-0002-5320-9110}{}

\author{Emilio {Di Giacomo}}{University of Perugia, Italy \and \url{https://mozart.diei.unipg.it/digiacomo/} }{emilio.digiacomo@unipg.it}{https://orcid.org/0000-0002-9794-1928}{}

\author{Walter Didimo}{University of Perugia, Italy \and \url{https://mozart.diei.unipg.it/didimo/} }{walter.didimo@unipg.it}{https://orcid.org/0000-0002-4379-6059}{}

\author{Luca Grilli}{University of Perugia, Italy \and \url{https://mozart.diei.unipg.it/grilli/} }{luca.grilli@unipg.it}{https://orcid.org/0000-0002-2463-3772}{}

\author{Maria Eleni Pavlidi}{University of Ioannina, Greece \and \url{https://algo.math.uoi.gr/marialena/index.html}}{m.e.pavlidi@uoi.gr}{https://orcid.org/0009-0009-4500-0112}{}

\author{Alessandra Tappini}{University of Perugia, Italy \and \url{https://mozart.diei.unipg.it/tappini/} }{alessandra.tappini@unipg.it}{https://orcid.org/0000-0001-9192-2067}{}

\author{Alexandra Weinberger}{FernUniversit\"at in Hagen, Germany \and \url{https://www.fernuni-hagen.de/ti/team/alexandra.weinberger.shtml} }{alexandra.weinberger@fernuni-hagen.de}{https://orcid.org/0000-0001-8553-6661}{}

\authorrunning{Bekos, Binucci, Di Giacomo, Didimo, Grilli, Pavlidi, Tappini, Weinberger} 

\Copyright{Michael A. Bekos, Carla Binucci, Emilio Di Giacomo, Walter Didimo, Luca Grilli, Maria Eleni Pavlidi, Alessandra Tappini, Alexandra Weinberger} 


\relatedversion{} 


\usepackage{xspace}
\usepackage{todonotes}
\usepackage{thmtools} 
\usepackage{thm-restate}
\usepackage{apptools}
\usepackage{hyperref}
\usepackage{cleveref}
\usepackage{complexity}
\usepackage[inline]{enumitem}

\newcommand{\restateref}[1]{\IfAppendix{\hyperref[#1]{$\star$}}{\hyperref[#1*]{$\star$}}}

\newtheorem{property}{Property}

\nolinenumbers
\hideLIPIcs

\begin{document}

\maketitle

\begin{abstract}
Linear layouts of graphs -- particularly \emph{stack} and \emph{queue} layouts -- are well-established types of representations in graph drawing, thanks to their connection with numerous theoretical and practical problems. In such layouts, all vertices are linearly ordered and the edges are partitioned into sets that avoid specific forbidden configurations: in a stack layout no two independent edges within the same set cross, whereas in a queue layout no two independent edges within the same set are nested. 
A central problem in this context is to determine, for a given graph $G$, its \emph{stack number} or \emph{queue number}, that is, the minimum number of sets into which the edges can be partitioned so that a corresponding stack or queue layout of $G$ exists.   
In this work, we introduce a relaxation of stack and queue layouts, which allows some forbidden patterns for the edges in the same set. Namely, for a given integer $k > 0$, a \emph{$k$-defective stack layout} (resp. a \emph{$k$-defective queue layout}) allows an edge to be in a crossing (resp. nesting) relationship with at most~$k$ edges within the same set. Our motivation is to extend the classes of graphs that admit linear layouts using a limited number of edge-partition sets, at the cost of allowing some defects.
We study defective linear layouts both from a combinatorial and from an algorithmic perspective, providing an array of results across different graph classes and parameters.
\end{abstract}

\section{Introduction}\label{se:intro}

Over the last four decades, linear layouts of graphs have been a central research topic in graph drawing and topological graph theory. As a result, the corresponding literature is rich in terms of both algorithmic and combinatorial results for several variants that have been proposed over the years. In this work, we focus on stack~\cite{DBLP:journals/jct/BernhartK79} and queue~\cite{DBLP:journals/siamcomp/HeathR92} layouts, where the task is to find a linear order of the vertices of a given graph and a partition of its edges into sets, called \emph{stacks} and \emph{queues}, respectively, such that edges in the same set avoid certain \emph{forbidden patterns} in the underlying linear order. In particular, no two independent edges of a stack~are allowed to \emph{cross} (that is, to have alternating endpoints in the underlying linear order), while two independent edges of a queue are not allowed to \emph{nest} (that is, the endpoints of one edge nest the endpoints of the other). A linear layout of a graph consisting exclusively of stacks is called a \emph{stack layout}, while a \emph{queue layout} is one consisting exclusively of queues. 

Given a graph, a central algorithmic question in this context is to determine its \emph{stack} or \emph{queue number}, i.e., the minimum number of stacks or queues, respectively, that are needed for a corresponding layout to exist. Since, in general, establishing these numbers is computationally hard~\cite{DBLP:journals/siamcomp/HeathR92,Wig82}, a big part of the literature has focused on the combinatorial problem of finding corresponding lower and upper bounds, assuming that the graph to be laid out belongs to a certain graph class (e.g., it is planar); refer to \cref{tbl:bounds} for an overview.

\begin{table}[tb]
	\centering
	\caption{%
	Lower and upper bounds on the 
	queue and stack numbers for various classes of graphs.}
	\label{tbl:bounds}
\begin{tabular}{p{4.5cm}|cccc|cccc}
	\toprule
	& \multicolumn{4}{c|}{\textbf{queue number}} & \multicolumn{4}{c}{\textbf{stack number}} \\
	\cmidrule(l{1ex}r{1ex}){2-5}
	\cmidrule(l{1ex}r{1ex}){6-9}
	\textbf{Graph class} & 
	\multicolumn{1}{c}{lower} & \multicolumn{1}{c}{ref.} & 
	\multicolumn{1}{c}{upper} & \multicolumn{1}{c|}{ref.} & 
	\multicolumn{1}{c}{lower} & \multicolumn{1}{c}{ref.} & 
	\multicolumn{1}{c}{upper} & \multicolumn{1}{c}{ref.}\\
	\midrule
	Trees &  
	1 & \textcolor{blue}{\cite{DBLP:journals/siamcomp/HeathR92}} & 
	1 & \textcolor{blue}{\cite{DBLP:journals/siamcomp/HeathR92}} & 
	1 & \textcolor{blue}{\cite{DBLP:journals/jct/BernhartK79}} & 
	1 & \textcolor{blue}{\cite{DBLP:journals/jct/BernhartK79}} \\
	
	Outerplanar & 
	2 & \textcolor{blue}{\cite{DBLP:journals/siamcomp/HeathR92}} & 
	2 & \textcolor{blue}{\cite{DBLP:journals/siamcomp/HeathR92}} & 
	1 & \textcolor{blue}{\cite{DBLP:journals/jct/BernhartK79}} & 
	1 & \textcolor{blue}{\cite{DBLP:journals/jct/BernhartK79}} \\

    Outer 1-planar & 
	3 & \textcolor{blue}{\cite{DBLP:journals/algorithmica/AuerBBGHNR16}} & 
	3 & \textcolor{blue}{\cite{DBLP:journals/algorithmica/AuerBBGHNR16}} & 
	2 & \textcolor{blue}{\cite{DBLP:journals/algorithmica/AuerBBGHNR16}} & 
	2 & \textcolor{blue}{\cite{DBLP:journals/algorithmica/AuerBBGHNR16}} \\
	
	Series-parallel & 
	3 & \textcolor{blue}{\cite{DBLP:journals/combinatorics/Wiechert17}} & 
	3 & \textcolor{blue}{\cite{DBLP:conf/cocoon/RengarajanM95}} & 
	2 & \textcolor{blue}{\cite{DBLP:conf/cocoon/RengarajanM95}} & 
	2 & \textcolor{blue}{\cite{DBLP:conf/cocoon/RengarajanM95}} \\
		
	Planar 3-trees & 
	4 & \textcolor{blue}{\cite{DBLP:journals/algorithmica/AlamBGKP20}} & 
	5 & \textcolor{blue}{\cite{DBLP:journals/algorithmica/AlamBGKP20}} & 
	3 & \textcolor{blue}{\cite{Wig82}} & 
	3 & \textcolor{blue}{\cite{DBLP:conf/focs/Heath84}} \\

	Planar bipartite & 
	3 & \textcolor{blue}{\cite{DBLP:conf/wads/ForsterKMPR23}} & 
	28 & \textcolor{blue}{\cite{DBLP:conf/wads/ForsterKMPR23}} & 
	2 & \textcolor{blue}{\cite{DBLP:journals/dcg/FraysseixMP95}} & 
	2 & \textcolor{blue}{\cite{DBLP:journals/dcg/FraysseixMP95}} \\
	
	Planar  & 
	4 & \textcolor{blue}{\cite{DBLP:journals/algorithmica/AlamBGKP20}} & 
	42 & \textcolor{blue}{\cite{DBLP:journals/algorithmica/BekosGR23}} & 
	4 & \textcolor{blue}{\cite{DBLP:journals/jocg/KaufmannBKPRU20}} & 
	4 & \textcolor{blue}{\cite{DBLP:journals/jcss/Yannakakis89}} \\
	
	Treewidth at most $t$ & 
	$t+1$ & \textcolor{blue}{\cite{DBLP:journals/combinatorics/Wiechert17}} & 
	$2^t+1$ & \textcolor{blue}{\cite{DBLP:journals/combinatorics/Wiechert17}} & 
	$t+1$ & \textcolor{blue}{\cite{DBLP:journals/siamdm/VandenbusscheWY09}} & 
	$t+1$ & \textcolor{blue}{\cite{DBLP:journals/dam/GanleyH01}} \\
	
	$K_n$ & 
	$\left\lfloor \frac{n}{2} \right\rfloor$ & \textcolor{blue}{\cite{DBLP:journals/siamcomp/HeathR92}} & 
	$\left\lfloor \frac{n}{2} \right\rfloor$ & \textcolor{blue}{\cite{DBLP:journals/siamcomp/HeathR92}} & 
	$\left\lceil \frac{n}{2} \right\rceil$ & \textcolor{blue}{\cite{DBLP:journals/jct/BernhartK79}} & 
	$\left\lceil \frac{n}{2} \right\rceil$ & \textcolor{blue}{\cite{DBLP:journals/jct/BernhartK79}} \\

    $K_{n,n}$ (non-sep.) & 
	$\left\lceil \frac{n}{2} \right\rceil$ & \textcolor{blue}{\cite{DBLP:journals/siamcomp/HeathR92}} & 
	$\left\lceil \frac{n}{2}\right\rceil$ & \textcolor{blue}{\cite{DBLP:journals/siamcomp/HeathR92}} & 
	$\left\lceil \frac{n}{2} \right\rceil$ & \textcolor{blue}{\cite{DBLP:journals/jct/BernhartK79}} &
	$\left\lfloor \frac{2n}{3} \right\rfloor$ +1 & \textcolor{blue}{\cite{DBLP:journals/jct/EnomotoNO97}} \\
	\bottomrule
\end{tabular}
\end{table}

\medskip
\noindent\textbf{Defective linear layouts.} We introduce a relaxation of stack and queue layouts, called \emph{defective}, which allows some forbidden patterns for the edges assigned to the same stack/queue, but to a ``limited extent''. 
As in graph drawing beyond planarity~\cite{DBLP:journals/csur/DidimoLM19,DBLP:books/sp/20/HT2020}, our motivation is to extend the classes of graphs that admit certain types of layouts (in our case linear layouts on a limited number of edge-partition sets), at the cost of allowing some~defects.

Our idea of relaxation is made more 
precise by means of a \emph{conflict graph} for a given linear order of the vertices, namely: each node of the conflict graph corresponds to an edge of the given graph to be laid out, and there is an arc connecting two nodes in the conflict graph if and only if the corresponding edges of the graph define a forbidden pattern in the given linear order when assigned to the same stack/queue. In the research direction of this paper, the goal is to determine a linear order of the vertices of the given graph and a partition of its edges into a minimum number of sets (\emph{defective stacks/queues}), such that
for the edges of each defective stack/queue
the degree of the conflict graph is at most $k$, where $k$ is a natural number defined as part of the input (see \cref{se:basic-definitions} for formal definitions). We refer to such layouts as \emph{$k$-defective}. Note that traditional stack and queue layouts are specializations of the $k$-defective ones to the case $k=0$, that is, when no forbidden patterns are allowed.

To the best of our knowledge, stack and queue layouts with defects have not previously been studied in the literature, although corresponding variants have been extensively studied for several other graph-theoretic problems~\cite{DBLP:journals/combinatorics/AboulkerAH22,DBLP:journals/disopt/AkdemirE15,DBLP:journals/pacmmod/Chang23,DBLP:journals/rsa/GiraoISW24,DBLP:journals/cpc/HendreyW19,DBLP:journals/dm/JingKMX22}. 
For instance, in the classic vertex coloring problem, the variant of the problem in which vertices of the same color are allowed to be adjacent if and only if the corresponding underlying conflict graph has degree at most $k$ has been studied under the term \emph{defective coloring}~\cite{DBLP:journals/jgt/CowenCW86} (where a conflict here corresponds to an edge whose end-vertices are of the same color); refer to~\cite{Wood2018DefectiveClustered} for~a~survey.
 
\medskip\noindent\textbf{Our results.} Inspired by the research on traditional stack and queue layouts, we investigate $k$-defective linear layouts from multiple~perspectives.
Given a graph~$G$ and a positive integer~$k$, the \emph{$k$-defective stack number} (\emph{$k$-defective queue number}, respectively) of~$G$ is the minimum number of defective stacks (defective queues, respectively) required for a $k$-defective layout of $G$ to exist. Our main results are bounds on the edge density and on the defective stack and queue numbers for various types of graphs, which offer a broad view of defective linear layouts across graph classes and parameters.
Moreover, we discuss the connection of our model with defective coloring and provide insights on the algorithmic problem of recognizing graphs that admit linear layouts with defects. 
More precisely, our results are the following:

\begin{itemize}
\item In \cref{se:coloringvslayouts} we study the relation between $k$-defective linear layouts and $k$-defective colorings. We also show that graphs having bounded $k$-defective stack (queue) number also have bounded stack (queue) number.

\item In \cref{se:density} we provide bounds on the edge density of graphs admitting $k$-defective layouts with $h$ defective queues (\Cref{sse:queue-density}) or $h$ defective stacks (\Cref{sse:stack-density}), for different values of $k$ and $h$. See \Cref{tb:edge-density} for a summary of these bounds.

\item In \cref{se:defective-page-number} we give bounds on the $k$-defective queue number (\Cref{sse:queue-defective-page-number}) and on the $k$-defective stack number (\Cref{sse:stack-defective-page-number}) for various graph families, including outer 1-planar graphs, planar graphs, complete graphs, and complete bipartite graphs. See \Cref{tb:queue-stack-number-1,tb:queue-stack-number-k} for a summary of these~bounds.

\item In \cref{se:final} we show that recognizing graphs that admit $k$-defective $1$-stack layouts can be done in quasi-polynomial time and we give some preliminary insights about recognizing graphs that admit $k$-defective $1$-queue layouts. We also provide a summary of our main findings and an array of interesting open research directions that arise from our work.
\item \cref{se:final} concludes our study by providing additional remarks and highlighting some future research directions stemming from our work.
\end{itemize}

\renewcommand{\arraystretch}{1.2}
\setlength{\tabcolsep}{2pt}

\begin{table}[htb]
\caption{Edge density results. The table also reports which bounds are proved to be tight.}
\label{tb:edge-density}
\centering
\begin{tabular}{@{} c|c|cc|cc @{}}
\toprule
  & & \multicolumn{2}{c|}{\textbf{$k$-defective $h$-queue layouts}} & \multicolumn{2}{c}{\textbf{$k$-defective $h$-stack layouts}} \\
\cmidrule(lr){3-4} \cmidrule(l){5-6}
\textbf{$k$} & \textbf{$h$} & \textbf{Upper bound} & \textbf{Tight} & \textbf{Upper bound} & \textbf{Tight} \\
\midrule
$> 1$ & $\geq 1$ & $h(k+2)\left(n -  \frac{h(k+2)+1}{2}\right)$ \textcolor{blue}{[Th.~\ref{th:density-general}]} &  & $O(\sqrt{k}hn)$ \textcolor{blue}{[Co.~\ref{th:1-defective-stack-density}]} & \\
$1$ & $> 1$ & $3h(n - \frac{3h+1}{2})$ \textcolor{blue}{[Th.~\ref{th:density-1page-1defect}]} & $\checkmark$ & $\left(\frac{3}{2}h + 1\right)n - 4h$ \textcolor{blue}{[Co.~\ref{th:1-defective-stack-density}]} & \\
$2$ & $1$ & $\frac{10}{3}n-\frac{21+r}{3}$, $r = n\mod 3$ \textcolor{blue}{[Th.~\ref{th:density-1page-2defect}]} & $\checkmark$ & $3n - 5$ \textcolor{blue}{[Th.~\ref{th:k-defective-1-stack-density}]} & $\checkmark$ \\
$1$ & $1$ & $3n - 6$ \textcolor{blue}{[Th.~\ref{th:density-1page-1defect}]} & $\checkmark$ & $\frac{5}{2}n - 4$ \textcolor{blue}{\cite{DBLP:journals/algorithmica/AuerBBGHNR16,DBLP:journals/ipl/Didimo13}} & $\checkmark$ \\
\bottomrule
\end{tabular}
\end{table}

\begin{table}[htb]
\caption{Queue/stack number results for $k=1$.}
\label{tb:queue-stack-number-1}
\centering
\begin{tabular}{@{} l|c|c|c|c @{}}
\toprule
 & \multicolumn{2}{c}{\textbf{$1$-defective queue number}}  & \multicolumn{2}{c}{\textbf{$1$-defective stack number}} \\
\cmidrule(lr){2-3} \cmidrule(l){4-5}
 \textbf{Graph class} & \textbf{Lower bound} & \textbf{Upper bound} & \textbf{Lower bound} & \textbf{Upper bound} \\
 \midrule
Outerplanar & $2$ \textcolor{blue}{[Le.~\ref{le:outerplanar-defectiveness-1-page}]} & $2$ \textcolor{blue}{\cite{DBLP:journals/siamcomp/HeathR92}} & 1 \textcolor{blue}{\cite{DBLP:journals/jct/BernhartK79}} & 1 \textcolor{blue}{\cite{DBLP:journals/jct/BernhartK79}} \\

Outer 1-planar & $2$ \textcolor{blue}{[Th.~\ref{th:outer-1-planar}]} & $2$ \textcolor{blue}{[Th.~\ref{th:outer-1-planar}]} & 1 \textcolor{blue}{[Th.~\ref{th:charact-k-defective-stack}]} & 1 \textcolor{blue}{[Th.~\ref{th:charact-k-defective-stack}]} \\
Planar & $2$ \textcolor{blue}{[Th.~\ref{th:outer-1-planar}]} & $33$ \textcolor{blue}{[Th.~\ref{th:planar-1-def-queue}]} & $2$ \textcolor{blue}{[Th.~\ref{th:charact-k-defective-stack}]} & $4$ \textcolor{blue}{\cite{DBLP:journals/jctb/Yannakakis20}} \\
$K_n$ & $\left\lceil \frac{n-1}{3} \right\rceil$ \textcolor{blue}{[Co.~\ref{cor:kn-1defect}]} & $\left\lceil \frac{n-1}{3} \right\rceil$ \textcolor{blue}{[Co.~\ref{cor:kn-1defect}]} & $\left\lfloor \frac{n}{3} \right\rfloor - 1$ \textcolor{blue}{[Th.~\ref{th:1-defective-sn-kn}]} & $\left\lceil \frac{n}{3} \right\rceil$ \textcolor{blue}{[Th.~\ref{th:1-defective-sn-kn}]} \\
$K_{n,n}$ (sep.) & $\left\lceil \frac{2n-1}{3} \right\rceil$ \textcolor{blue}{[Co.~\ref{lem:knn-sep-1defect}]} & $\left\lceil \frac{2n-1}{3} \right\rceil$ \textcolor{blue}{[Co.~\ref{lem:knn-sep-1defect}]} & $\left\lceil \frac{n}{2} \right\rceil$ \textcolor{blue}{[Co.~\ref{cor:1-defective-stack-knn-sep}]} & $\left\lceil \frac{2n}{3} \right\rceil$ \textcolor{blue}{[Co.~\ref{cor:1-defective-stack-knn-sep}]} \\
$K_{n,n}$ (non-sep.) & $\left\lceil \frac{n-1}{3} \right\rceil$ \textcolor{blue}{[Co.~\ref{lem:knn-sep-1defect}]} & $\left\lceil \frac{n}{2} \right\rceil$  \textcolor{blue}{\cite{DBLP:journals/siamcomp/HeathR92}} & $\left\lceil \frac{n}{4} \right\rceil$ \textcolor{blue}{[Th.~\ref{thm:1-defective-stack-knn}]} & $\left\lceil \frac{n}{2} \right\rceil$ \textcolor{blue}{[Th.~\ref{thm:1-defective-stack-knn}]} \\
\bottomrule
\end{tabular}
\end{table}

\begin{table}[htb]
\caption{Queue/stack numbers for $k > 1$, where $l_1 = \left\lfloor \frac{3+\sqrt{8k+1}}{2} \right\rfloor$, $l_2 = \left\lfloor \frac{-1 + \sqrt{8k+1}}{2} \right\rfloor$, $l_3 = \sqrt{k} + 1$.}
\label{tb:queue-stack-number-k}
\centering
\begin{tabular}{@{} l|c|c|c|c @{}}
\toprule
& \multicolumn{2}{c}{\textbf{$k$-defective queue number}}  & \multicolumn{2}{c}{\textbf{$k$-defective stack number}} \\
\cmidrule(lr){2-3} \cmidrule(l){4-5}
 \textbf{Graph class} & \textbf{Lower bound} & \textbf{Upper bound} & \textbf{Lower bound} & \textbf{Upper bound} \\
\midrule
Outerplanar & $2$ \textcolor{blue}{[Le.~\ref{le:outerplanar-defectiveness-1-page}]} & $2$ \textcolor{blue}{\cite{DBLP:journals/siamcomp/HeathR92}} & 1 \textcolor{blue}{\cite{DBLP:journals/jct/BernhartK79}} & 1 \textcolor{blue}{\cite{DBLP:journals/jct/BernhartK79}} \\

Outer 1-planar & $2$ \textcolor{blue}{[Th.~\ref{th:outer-1-planar}]}& $2$ \textcolor{blue}{[Th.~\ref{th:outer-1-planar}]} & 1 \textcolor{blue}{[Th.~\ref{th:charact-k-defective-stack}]} & 1 \textcolor{blue}{[Th.~\ref{th:charact-k-defective-stack}]} \\ 
Outer $k$-planar & $2$ \textcolor{blue}{[Th.~\ref{th:outer-1-planar}]} &  & 1 \textcolor{blue}{[Th.~\ref{th:charact-k-defective-stack}]} & 1 \textcolor{blue}{[Th.~\ref{th:charact-k-defective-stack}]} \\
$K_n$ & $\left\lceil \frac{n-1}{k+2} \right\rceil$ \textcolor{blue}{[Th.~\ref{th:kn-Kdefect-lower}]} & $\left\lceil \frac{n-1}{l_1} \right\rceil$ \textcolor{blue}{[Th.~\ref{th:kn-kdefect}]} & $\left\lceil \frac{n}{2k+2} \right\rceil$ \textcolor{blue}{[Th.~\ref{th:k-defective-stack-kn}]} & $\left\lceil \frac{n}{l_2+2} \right\rceil$ \textcolor{blue}{[Th.~\ref{th:k-defective-stack-kn}]} \\
$K_{n,n}$ (sep.) & $\left\lceil \frac{2n-1}{k+2} \right\rceil$ \textcolor{blue}{[Co.~\ref{lem:knn-sep-1defect}]} & $\left\lceil \frac{2n-1}{l_1} \right\rceil$ \textcolor{blue}{[Co.~\ref{lem:knn-sep-1defect}]} & $\left\lceil \frac{n}{k+1} \right\rceil$ \textcolor{blue}{[Co.~\ref{cor:1-defective-stack-knn-nonsep}]} & $\left\lceil \frac{2n}{l_2+2} \right\rceil$ \textcolor{blue}{[Co.~\ref{cor:1-defective-stack-knn-nonsep}]} \\
$K_{n,n}$ (non-sep.) & $\left\lceil \frac{n-1}{k+2} \right\rceil$ \textcolor{blue}{[Co.~\ref{lem:knn-sep-1defect}]} & $\left\lceil\frac{2n-1}{l_1}\right\rceil$ \textcolor{blue}{[Co.~\ref{lem:knn-sep-1defect}]} & $\left\lceil\frac{n}{2k+2}\right\rceil$ \textcolor{blue}{[Th.~\ref{th:k-defective-stack-knn}]} & $\left\lceil \frac{n}{l_3} \right\rceil$ \textcolor{blue}{[Th.~\ref{th:k-defective-stack-knn}]} \\
\bottomrule
\end{tabular}
\end{table}

\section{Basic Definitions}\label{se:basic-definitions}

Let $G$ be an $n$-vertex graph and let  $v_0 \prec \dots \prec v_{n-1}$ be a linear order of its vertices. Two independent edges $(v_i,v_j)$ and $(v_g,v_l)$, with $v_i \prec v_g$, $v_i \prec v_j$ and $v_g \prec v_l$
are in a \emph{crossing relationship} (or \emph{cross}) if 
$v_i \prec v_g \prec v_j \prec v_l$; they are in a \emph{nesting relationship} (or they \emph{nest}) if $v_i \prec v_g \prec v_l \prec v_j$.
When referring to the linear order $\prec$, we say that $v_i$ is \emph{to the left} of $v_j$ (or that $v_j$ is \emph{to the right} of $v_i$) if $v_i\prec v_j$. In other words, we consider~$\prec$ as a left-to-right order. Moreover, we say that an edge $(v_i,v_j)$ is an \emph{$r$-hop edge} or it has \emph{hop-size} $r$, if $r=|i-j|$.

Let $h \geq 1$  and $k \geq 0$ be two integers.
A \emph{$k$-defective $h$-stack ($h$-queue) layout} $\mathcal{L}$ of a graph $G$ consists of a linear order of the vertices of $G$ and a partition of the edges of $G$ into $h$ non-empty sets, called \emph{defective stacks} (\emph{defective queues}), such that every edge is in a crossing (nesting) relationship with at most $k$ other edges of the same defective stack (defective queue); see  \cref{fi:examples-queue-stack-conflict-a,fi:examples-queue-stack-conflict-b,fi:examples-queue-stack-conflict-c,fi:examples-queue-stack-conflict-d}. An equivalent definition in terms of a suitable conflict graph is given in \cref{se:coloringvslayouts}.
As typically done for classic linear layouts, we visually represent defective linear layouts by drawing all vertices as points on a horizontal line, ordered from left to right according to $\prec$, and by drawing the edges as half-circles. 

A set of $k+1$ edges in the same defective stack (queue) such that one of them is in a crossing (nesting) relationship with the other $k$ edges is called a \emph{$k$-defect}.
If $\mathcal{L}$ is a $k$-defective $h$-stack ($h$-queue) layout with $k \ge 1$ and $\mathcal{L}$ is not a $(k-1)$-defective $h$-stack ($h$-queue) layout, then $k$ is the {\em defectiveness} of~$\mathcal{L}$.
For a graph $G$ and for an integer $k \geq 0$, let $h$ be the minimum non-negative integer for which $G$ admits a $k$-defective $h$-stack ($h$-queue) layout. We say that $h$ is the \emph{$k$-defective stack (queue) number} of $G$. The stack (queue) number of $G$ coincides with the $0$-defective stack (queue) number of $G$.
In the following, when talking about a $k$-defective $h$-stack ($h$-queue) layout, we omit the values of $h$ and/or of $k$ if they are not relevant for the context.
Further, since we are mainly interested in $k$-defective stack (queue) layouts for $k \geq 1$, we assume that $k$ is a positive number if not specified otherwise.

\begin{figure}[tb]
    \centering
    \begin{subfigure}{.32\textwidth}
        \centering
        {\includegraphics[page=1,width=\textwidth]{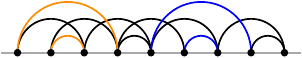}
        \subcaption{\label{fi:examples-queue-stack-conflict-a}}}
    \end{subfigure}
    \hfil
    \begin{subfigure}{.32\textwidth}
        \centering
        {\includegraphics[page=2,width=\textwidth]{figures/examples-queue-stack-conflict.pdf}\subcaption{\label{fi:examples-queue-stack-conflict-c}}}
    \end{subfigure}
    \hfil
    \begin{subfigure}{.32\textwidth}
        \centering
        {\includegraphics[page=3,width=\textwidth]{figures/examples-queue-stack-conflict.pdf}
        \subcaption{\label{fi:examples-queue-stack-conflict-e}}}
    \end{subfigure}
    \hfil
    \begin{subfigure}{.32\textwidth}
        \centering
        {\includegraphics[page=4,width=\textwidth]{figures/examples-queue-stack-conflict.pdf}
        \subcaption{\label{fi:examples-queue-stack-conflict-b}}}
    \end{subfigure}
    \hfil
    \begin{subfigure}{.32\textwidth}
        \centering
        {\includegraphics[page=5,width=\textwidth]{figures/examples-queue-stack-conflict.pdf}
        \subcaption{\label{fi:examples-queue-stack-conflict-d}}}
    \end{subfigure}
    \hfil
    \begin{subfigure}{.32\textwidth}
        \centering
        {\includegraphics[page=6,width=\textwidth]{figures/examples-queue-stack-conflict.pdf}
        \subcaption{\label{fi:examples-queue-stack-conflict-f}}}
    \end{subfigure}
    \hfil
    \caption{
    (a)~A $1$-defective $1$-queue layout; (b)~a $2$-defective $1$-queue layout; (d)~a $1$-defective $1$-stack layout; (e)~a $2$-defective $1$-stack layout.  (c) and (f) The conflict graphs of the layouts in \cref{fi:examples-queue-stack-conflict-c,fi:examples-queue-stack-conflict-d}, respectively. A $k$-defect is represented by edges with the same (non-black) color.}
    \label{fi:configuration}
\end{figure}

\section{Defective Linear Layouts and Defective Coloring}\label{se:coloringvslayouts}

One may observe an interesting relationship between defective stack and queue layouts and defective coloring~\cite{DBLP:journals/jgaa/AngeliniBLD0KMR17}. \emph{Defective coloring} is a variant of the traditional vertex coloring in which adjacent vertices can have the same color, as long as each monochromatic component has a certain structure. Let $\prec$ be a linear order of the vertices of a graph $G$. We define a \emph{conflict graph} $C_{\prec}$ as follows. The nodes of $C_\prec$ correspond to the edges of $G$, and an edge $(u,v)$ exists in $C_\prec$ if and only if the two edges corresponding to $u$ and $v$ form a $1$-defect in $\prec$ (i.e., they are in a crossing/nesting relationship in a defective stack/queue layout);
see~\cref{fi:examples-queue-stack-conflict-e,fi:examples-queue-stack-conflict-f}. By this definition, we immediately obtain the following property.

\begin{property}\label{pr:ch-crossing-graph}
A graph admits a $k$-defective $h$-stack ($h$-queue) layout if and only if there exists a linear order $\prec$ of its vertices whose conflict graph $C_\prec$ admits a defective coloring with $h$ colors such that every monochromatic induced component has vertex-degree at most $k$. 
\end{property}

\noindent As a degree-$k$ graph is $(k+1)$-colorable, the next observation follows from \cref{pr:ch-crossing-graph}.

\begin{observation}\label{obs:bounded-defective}
Every graph that has $k$-defective stack (queue) number $h$ has stack (queue) number at most $(k+1)\cdot h$.
\end{observation}

\Cref{obs:bounded-defective} implies that if a graph has $k$-defective stack (queue) number~$1$, then its stack (queue) number is at most $k+1$. The reverse does not hold in general. For instance, $K_n$ has stack (queue) number $\left\lceil \frac{n}{2} \right\rceil$~\cite{DBLP:journals/jct/BernhartK79} ($\left\lfloor \frac{n}{2} \right\rfloor$~\cite{DBLP:journals/siamcomp/HeathR92}), whereas every defective $1$-stack ($1$-queue) layout of $K_n$ has defectiveness $\frac{(n-2)^2}{4}$ ($\frac{(n-2)(n-3)}{2}$), which is more than $\left\lceil \frac{n}{2} \right\rceil - 1$~($\left\lfloor \frac{n}{2} \right\rfloor - 1$).     
\cref{pr:ch-crossing-graph} has also implications to layouts with a fixed linear order, e.g., if in a fixed linear order $\prec$ each edge is involved in a crossing (nesting) relationship with at most five edges (i.e., the degree of the conflict graph $C_\prec$ is at most $5$), then there is a $1$-defective $3$-stack ($3$-queue) layout respecting $\prec$~\cite{Lo66}. More in general, given a graph $G$ with a fixed vertex order $\prec$, testing whether $G$ admits a $k$-defective $h$-stack ($h$-queue) layout respecting $\prec$ is equivalent to testing whether the corresponding conflict graph $C_\prec$ admits a defective coloring with $h$ colors such that every monochromatic induced component has vertex-degree at most~$k$. The latter testing problem is NP-complete for general graphs even for $k=1$ and $h=2$~\cite{DBLP:journals/jgt/CowenGJ97}. However, this result does not imply NP-completeness for the problem of testing if a graph has a $k$-defective $h$-stack ($h$-queue) layout with a fixed linear order $\prec$ (e.g., for $k=1$ and $h=2$), as not every graph corresponds to the conflict graph of a linear order. This is proven in \cref{lem:crossing-graph-observation}, where the operation of \emph{stacking} a degree-$2$ vertex on an edge $(u,v)$ introduces a new vertex $w$ into a graph along with the edges $(u,w)$ and $(w,v)$.

\begin{lemma}\label{lem:crossing-graph-observation}
For $n \geq 5$, the graph $G_n$ obtained from the complete graph $K_n$ by stacking a degree-$2$ vertex on each of its edges does not express the crossing relationships of a set of $\frac{n(n+1)}{2}$ edges.    
\end{lemma}
\begin{proof}
Assume for a contradiction that graph $G_n$ expresses the crossing relationships of a set of $\frac{n(n+1)}{2}$ edges in a linear layout $\mathcal{L}$. It follows that the $n$ edges $e_0,\ldots,e_{n-1}$ that correspond to $K_n$ in $G_n$ form an $n$-twist in $\mathcal{L}$, that is, they pairwise cross. Without loss of generality, assume that the left endpoint of $e_i$ precedes the left endpoint of $e_j$ when $i < j$. Each degree-$2$ vertex $v_{ij}$ in $G_n \setminus K_n$ corresponds to an edge $e_{ij}$ in $\mathcal{L}$ that is involved in a crossing situation with exactly two edges of $\mathcal{L}$, say $e_i$ and $e_j$, where $i,j\in\{0,\ldots,n-1\}$. Since the edges $e_0,\ldots,e_{n-1}$ form an $n$-twist, it follows that the endpoints of $e_i$ and $e_j$ must be consecutive in $\mathcal{L}$, as otherwise $e_{ij}$ is inevitably involved in a crossing situation with more than two edges; a contradiction to the degree of $v_{ij}$ in $G_n$. Since there exist at most $2n$ consecutive endpoints for the edges $e_0,\ldots,e_{n-1}$, while the degree-$2$ vertices of the conflict graph $G_n$ are $\frac{n(n-1)}{2}$, this yields a contradiction when $n \geq 5$. Hence, $G_n$ does not express the crossing relationships of any set of $\frac{n(n+1)}{2}$ edges, as desired.
\end{proof}

\section{Edge Density}\label{se:density}
In this section we give bounds on the density of graphs admitting $k$-defective $h$-queue layouts and $k$-defective $h$-stack layouts, for different values of $k$ and $h$.

\subsection{Defective Queue Layouts}\label{sse:queue-density}

Let $v_0 \prec \dots \prec v_{n-1}$ be a linear order of the vertices of a graph $G$ and let $K_n$ be the complete graph on the same number of vertices. Partition the set of edges of $K_n$ into $n-1$ \emph{classes} $C'_1, \dots, C'_{n-1}$; two edges $(v_{i_1},v_{j_1})$ and $(v_{i_2},v_{j_2})$, with $0 \leq i_1,j_1,i_2,j_2 \leq n-1$, belong~to~the same class if $\left\lfloor\frac{i_1+j_1}{2} \right\rfloor=\left\lfloor\frac{i_2+j_2}{2}\right\rfloor$. By definition, each class contains a different number of edges ranging from $1$ to $n-1$; the class $C'_i$ is that containing $i$ edges (see \cref{fi:alternating-a}). Now define classes $C_1, \dots, C_{n-1}$, each containing a subset (possibly empty) of distinct edges of $G$. Namely, $C_i$ is the restriction of $C'_i$ to the edges of $G$ (see \cref{fi:alternating-b}). Clearly, $C_i$ contains at most $i$ edges. Also, if $C_i$ contains exactly $i$ edges, these edges form a path.

\begin{figure}[h]
    \centering
    \begin{subfigure}{.49\textwidth}
        \centering
        {\includegraphics[page=6,width=.9\textwidth]{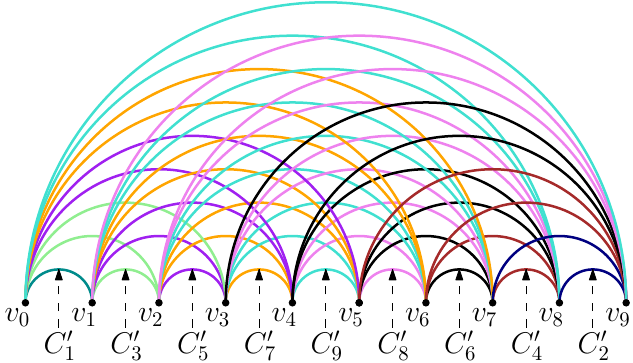}
        \subcaption{\label{fi:alternating-a}}}
    \end{subfigure}
    \hfill
    \begin{subfigure}{.49\textwidth}
        \centering
        {\includegraphics[page=7,width=.9\textwidth]{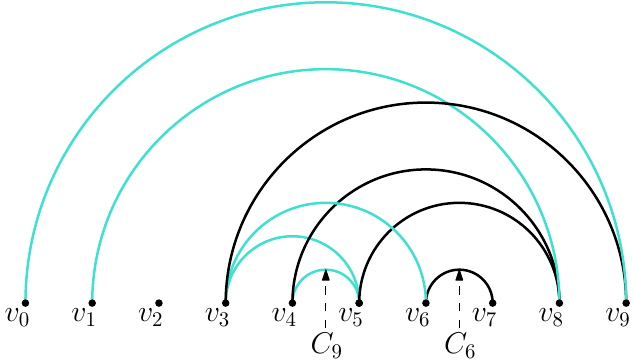}
        \subcaption{\label{fi:alternating-b}}}
    \end{subfigure}
    \caption{(a) Partition of the edges of $K_{10}$ into $9$ classes.
    (b) Partition of the edges of a subgraph~$G$ of $K_{10}$ into two non-empty classes. The turquoise edges belong to $C_9$; the black edges belong to~$C_6$.}
    \label{fi:alternating}
\end{figure}

\begin{lemma}\label{le:alternating}
    Let $\mathcal{L}$ be a $k$-defective $h$-queue layout  of a graph $G$. Every defective queue of $\mathcal{L}$ contains at most $k+2$ edges of each class.  
\end{lemma}
\begin{proof}
    If a defective queue contained a set $E^*$ of more than $k+2$ edges of a class $C_i$, for $1 \leq i \leq n-1$, then the longest edge would be in a nesting relationship with all the edges in $E^*$ except at most one, and the defectiveness of $\mathcal{L}$ would be at least $k+1$.  
\end{proof}

\noindent We obtain the following upper bound on the edge density of $k$-defective $h$-queue layouts by using a technique similar to~\cite{DBLP:journals/jgaa/BoseCMW04,DBLP:journals/dmtcs/DujmovicW04}.

\begin{theorem}\label{th:density-general}
Any $n$-vertex graph that admits a $k$-defective $h$-queue layout has at most $h(k+2)\left(n -  \frac{h(k+2)+1}{2}\right)$ edges, for $n \geq h(k+2) + 1$.
\end{theorem}
\begin{proof}
Let $\mathcal{L}$ be a $k$-defective $h$-queue layout of a graph $G$ with $n$ vertices and $m$ edges. A class $C_i$ of edges of $G$ is \emph{small} if $i \leq h(k+2)$, it is said to be \emph{large} otherwise.
Observe that, since $C_i$ has at most $i$ edges, a small class has at most $h(k+2)$ edges.
Let $n_l$ be the number of large classes and let $m_s$ be the total number of edges of $G$ that belong to a small class. We claim that $m \leq n_lh(k+2)+m_s$. Assume for a contradiction that $m > n_lh(k+2)+m_s$. Then $m - m_s> n_lh(k+2)$; that is, the number of edges from large classes is more than $n_lh(k+2)$. This implies that there is at least one class with more than $h(k+2)$ edges, which in turn implies that there is one defective queue of $\mathcal{L}$ with more than $k+2$ edges from this class, thus contradicting \Cref{le:alternating}.

Since each class $C_i$ has at most $i$ edges, we have $n_l=n-1-h(k+2)$.
Note that since $n \ge h(k+2) + 1$, we have $n_l \ge 0$.
Further, $m_s \leq \sum_{i=1}^{h(k+2)}i=\frac{h(k+2)(h(k+2)+1)}{2}$. 
It follows that $m \leq (n-1-h(k+2))h(k+2)+\frac{h(k+2)(h(k+2)+1)}{2}=h(k+2)\left(n -\frac{h(k+2)+1}{2}\right)$.
Observe that $h(k+2)\left(n -\frac{h(k+2)+1}{2}\right) \leq \frac{n(n-1)}{2}$ for $n \geq h(k+2) + 1$.
\end{proof}

For $k=0$ and any $h\geq 1$, the bound of \cref{th:density-general} is tight~\cite{DBLP:journals/dmtcs/DujmovicW04,DBLP:journals/siamcomp/HeathR92}.
The next theorem shows that it is also tight for $k=1$ and any $h\ge 1$.

\begin{theorem}\label{th:density-1page-1defect}
    Any $n$-vertex graph that admits a $1$-defective $h$-queue layout has at most $3h(n-\frac{3h+1}{2})$ edges, for $n \geq 3h+1$, and this bound is tight.
\end{theorem}

\begin{proof}
    The upper bound follows from \Cref{th:density-general}.
    To prove that the bound is tight, for every $n \geq 3h+1$ we construct~an~$n$-vertex graph $G$ with $3h\left(n-\frac{3h+1}{2}\right)$ edges admitting a $1$-defective $h$-queue layout $\mathcal{L}$. We describe $G$ by means of its layouts $\cal L$; see \cref{fi:density-1-page-h-defective}. 

\begin{figure}[h]
    \centering
    {\includegraphics[page=3, width=.35\textwidth]{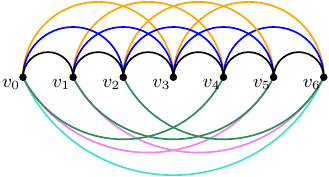}}
    \caption{A $1$-defective $2$-queue layout of a graph
    with 7 vertices and 21 edges.
    }\label{fi:density-1-page-h-defective}
\end{figure}
    
    Let $v_0 \prec v_1 \prec \dots \prec v_{n-1}$ be the linear order of $\mathcal{L}$. For $i = 1,2,\dots,h$, add to the $i$-th defective queue of $\mathcal{L}$, all the $(3(i-1)+1)$-hop edges, the $(3(i-1)+2)$-hop edges, and the $(3(i-1)+3)$-hop edges. 
    Observe that since $n \geq 3h+1$, each defective queue is non-empty. Indeed, the number of edges in the last defective queue is $n-(3(h-1)+1)+n-(3(h-1)+2)+n-(3(h-1)+3) = 3n-9(h-1)-6 \geq 3(3h+1)-9(h-1)-6 = 6$.
    
    The constructed layout is a $1$-defective $h$-queue layout. Namely, each $(3(i-1)+1)$-hop edge is in a nesting relationship with at most one $(3(i-1)+3)$-hop edge, and vice versa. Since the number of $j$-hop edges is $n-j$, with $j\in\{1,2,\dots,n-1\}$, the total number of edges of~$G$ is 
    $\sum_{i=1}^{h}(3n-9(i-1)-6)=3h\left(n-\frac{3h+1}{2}\right)$. 
\end{proof}

For $k \geq 2$, the bound of \cref{th:density-general} is not tight. In particular, for $k=2$ and $h=1$, in the rest of this section, we seek to prove the following.

\begin{theorem}
\label{th:density-1page-2defect}
Any $n$-vertex graph that admits a $2$-defective $1$-queue layout has at most $\frac{10}{3}n-\frac{21+r}{3}$ edges, where $n \geq 3$ and $r=n\mod 3$. This bound is tight for every $n=3s+r$ and $s\in \mathbb{N}^+$.
\end{theorem}

To prove \cref{th:density-1page-2defect}, we need additional definitions and lemmas.
Let $\Gamma$ be a linear layout with linear order $v_0 \prec \dots \prec v_{n-1}$; $\Gamma$ can be described by means of its \emph{$M$-diagram $M(\Gamma)$}, which is a (half) matrix of points such that: (i) there is a point with coordinates $(i,j)$ for every $i \in \{0,1\dots,n-2\}$ and every $j \in \{1,2,\dots,n-1\}$; (ii) point $(i,j)$ is colored if and only if the edge $(v_i,v_j)$ is an edge of $\Gamma$. A point of $M(\Gamma)$ is \emph{used} if the corresponding edge exists in $\Gamma$, it is \emph{unused} otherwise; see \cref{fi:M-diagram-a}. 
The \emph{lower-right wedge} of a point $(i,j)$ of $M(\Gamma)$ is the set of points $(i',j')$ with $i' > i$ and $j' < j$. The \emph{closed lower-right wedge} of $(i,j)$ is the set of points $(i',j')$ with $i'\geq i$ and $j' \leq j$.
The \emph{upper-left wedge} of a point $(i,j)$ of $M(\Gamma)$ is the set of points $(i',j')$ with $i'< i$ and $j' > j$. The \emph{closed upper-left wedge} of $(i,j)$ is the set of points $(i',j')$ with $i'\leq i$ and $j' \geq j$.
An edge $(v_i,v_j)$ has a conflict with all edges $(v_{i'},v_{j'})$ such that $(i',j')$ lies in its upper-left wedge or in its lower-right wedge of $(i,j)$ in $M(\Gamma)$.

\begin{figure}[tb]
    \centering
    {\includegraphics[page=1]{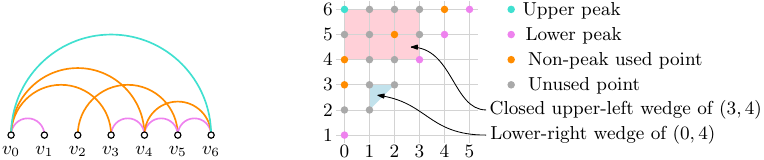}}
    \caption{A linear layout and its M-diagram.}
    \label{fi:M-diagram-a}
\end{figure} 

An \emph{upper peak} of $M(\Gamma)$ is a used point $(i,j)$ such that every other point in its closed upper-left wedge is unused. An upper peak corresponds to an edge that is not nested inside any other edge in $\Gamma$. A \emph{lower peak} of $M$ is a used point $(i,j)$ such that every other point in its closed lower-right wedge is unused. A lower peak corresponds to an edge that does not nest any other edge in $\Gamma$.  Notice that a peak can be a lower peak and an upper peak at the same time; see \cref{fi:M-diagram-b} (left). 

An \emph{orthogonal polyline} of $M(\Gamma)$ is a polyline consisting of alternating horizontal and vertical segments, each connecting two points of $M(\Gamma)$. A point $(i,j)$ of an orthogonal polyline is a \emph{left (right) turn} if walking along the polyline from bottom to top, one makes a left (right) turn at $(i,j)$. 
Refer to \cref{fi:M-diagram-b}.
The \emph{upper frontier} of $M(\Gamma)$ is an orthogonal polyline starting at the bottomost point of $M(\Gamma)$, ending at the rightmost point of $M(\Gamma)$, and whose right turns are all and only the upper peaks. The \emph{lower frontier} of $M(\Gamma)$ is an orthogonal polyline starting at the bottommost point of $M(\Gamma)$, ending at the rightmost point of $M(\Gamma)$, and whose left turns are all and only the lower peaks. By definition, the upper frontier and the lower frontier are both $x$- and $y$-monotone. Further, the upper frontier and the lower frontier can share vertices and edges, and they can also cross. The upper and lower frontiers \emph{cross} if walking along one of them, one sees portions of the other one both on the left-hand side and on the right-hand side; see \cref{fi:M-diagram-b} (left).

\begin{figure}
    \centering
    {\includegraphics[page=2]{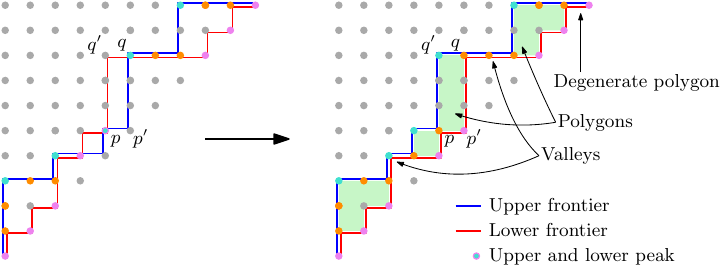}}
    \caption{Transformation of an M-diagram whose upper and lower frontiers cross into an M-diagram whose upper and lower frontiers do not cross.}
    \label{fi:M-diagram-b}
\end{figure}  

The next four lemmas allow us to concentrate on 2-defective 1-queue layouts whose M-diagram has suitable properties that will help us prove~\cref{th:density-1page-2defect}. In particular, the next lemma proves that a layout $\Gamma$ such that the upper and the lower frontiers of $M(\Gamma)$ cross can be modified by removing such an undesired property, while not decreasing the number of edges. We remark that if $M(\Gamma)$ has a point that is simultaneously a lower and an upper peak, then the upper and lower frontiers necessarily cross. Thus, the removal of crossings between frontiers ensures that no point is simultaneously an upper and a lower peak.

\begin{lemma}\label{le:no-crossing-frontiers}
Let $\Gamma$ be a 2-defective 1-queue layout with $m$ edges. There exists a 2-defective 1-queue layout $\Gamma'$ with $m'$ edges such that $m \leq m'$ and such that the upper frontier and the lower frontier of $M(\Gamma')$ do not cross.
\end{lemma}
\begin{proof}
    If the upper and lower frontiers of $M(\Gamma)$ do not cross, then $\Gamma'=\Gamma$. Thus, assume that the lower frontier $\phi_L$ and the upper frontier $\phi_U$ cross; see \cref{fi:M-diagram-b}. This means that there exists at least one portion of $\phi_L$ that is to the left of $\phi_U$ when walking along $\phi_U$ from bottom to top. Let $p$ and $q$ be two points shared by $\phi_U$ and $\phi_L$ such that:  (i) $\phi_L$ is to the left of $\phi_U$ when walking along $\phi_U$ from bottom to top; and (ii) when walking from $p$ to $q$ along $\phi_U$ there is no other vertex shared by $\phi_U$ and $\phi_L$. Let $\phi'_L$ (resp. $\phi'_U$) be the portion of $\phi_L$  (resp. $\phi_U$) from $p$ to $q$. We claim that $\phi'_U$ and $\phi'_L$ do not contain a used point other than $p$ and $q$. Namely, suppose that $\phi'_L$ contains a used point $r$ between $p$ and $q$; since $r$ does not belong to $\phi'_U$, it is not an upper peak and therefore its closed upper-left wedge contains at least one used point. Let $r'$ be a used point in the closed upper-left wedge of $r$ such that it has no used point in its closed upper-left wedge. Point $r'$ is an upper peak and therefore it belongs to $\phi_U$. In this case however, $\phi_U$ would not be $x$- and $y$-monotone. Thus, point $r$ cannot exist. Symmetrically $\phi'_U$ cannot contain a used point between $p$ and $q$. As a consequence, the polygon defined by $\phi'_L$ and $\phi'_U$ is a rectangle $R$ whose used points are only $p$ and $q$. Let $p'$ and $q'$ be the two corners of $R$ distinct from $p$ and $q$, with $p'$ belonging to $\phi'_U$ and $q'$ belonging to $\phi'_L$. By making $p'$ and $q'$ used (i.e., by adding the corresponding edges to $\Gamma$), $p'$ becomes a lower peak and $q'$ becomes an upper peak. Also, no conflict is created. As a consequence, $\phi'_L$ becomes part of $\phi_U$ and $\phi'_U$ becomes part of $\phi_L$ (in other words, the two frontiers exchange the two portions $\phi'_U$ and $\phi'_L$). Finally, the number of used points (and therefore the number of edges in the layout) has increased by two units. If after modifying $\phi'_U$ and $\phi'_L$, the (new) upper and lower frontiers still cross (in some other portion), we apply again the procedure above. 
\end{proof}

\begin{lemma}\label{le:shared-used}
Let $\Gamma$ be a 2-defective 1-queue layout with $m$ edges. There exists a 2-defective 1-queue layout $\Gamma'$ with $m'$ edges such that $m \leq m'$ and such that in $M(\Gamma')$ all points shared by the upper and lower frontiers are used.
\end{lemma}
\begin{proof}
    Let $(i,j)$ be a point shared by the upper frontier $\phi_U$ and the lower frontier $\phi_L$. The upper-left wedge of $(i,j)$ does not contain any used point. Namely, if a point $(i',j')$ with in the upper-left wedge of $(i,j)$ was used, then $(i',j')$ would be a point of $\phi_U$ and, by the monotonicity of $\phi_U$, $(i,j)$ would not be a point of $\phi_U$, contradicting the fact that $(i,j)$ is shared by $\phi_U$ and $\phi_L$. Symmetrically, the lower-right wedge of $(i,j)$ does not contain any used point. It follows that using $(i,j)$, i.e., adding the corresponding edge to $\Gamma$, does not create any conflict.
\end{proof}

An $M$-diagram with the properties of \cref{le:no-crossing-frontiers,le:shared-used} is called a \emph{standard} $M$-diagram; see \cref{fi:M-diagram-b} (right). In a standard $M$-diagram the upper and lower frontier define a set of, possibly degenerate, orthogonal polygons, each sharing a vertex with the next one (except the last one). A degenerate orthogonal polygon consists of a sequence of edges shared by the upper and lower frontiers. 
We refer to this set of polygons as the \emph{polygons of $M(\Gamma)$}. Notice that the first and the last polygons of $M(\Gamma)$ are always degenerate. Each point shared by the upper and lower frontiers, including possibly shared upper and lower peaks, is called a \emph{valley}.

\begin{lemma} \label{le:at-most-two-points}
Let $\Gamma$ be a 2-defective 1-queue layout with $m$ edges with a standard $M$-diagram. There exists a 2-defective 1-queue layout $\Gamma'$ with $m'$ edges such that $m \leq m'$, $M(\Gamma')$ has the same set of polygons of $M(\Gamma)$, and each non-degenerate polygon of $M(\Gamma')$ has at most two used points that are neither peaks nor valleys. 
\end{lemma}
\begin{proof}
    Let $P$ be a non-degenerate polygon of $M(\Gamma)$. The peaks of $P$ are all used points by definition; the two valleys that $P$ shares with the polygon preceding and with the polygon following it are also used because $M(\Gamma)$ is standard. Denote by $p_1, p_2, \dots,p_a$ the upper peaks of $P$ and by $q_1,q_2,\dots q_b$ the lower peaks of $P$. We claim that the edge corresponding to a peak $p_l$ with $l \in \{2,3,\dots,a-1\}$ or to a peak $q_g$ with $g \in \{2,3,\dots,b-1\}$ has exactly two conflicts in $\Gamma$. 

    \begin{figure}
        \centering
        {\includegraphics[page=3,width=\textwidth]{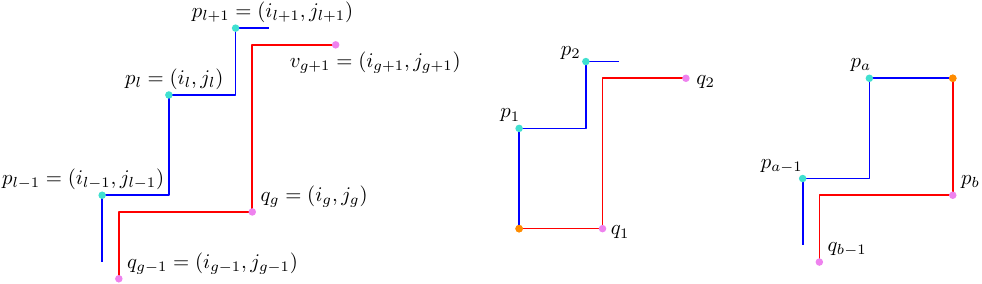}}
        \caption{Illustration for the proof of \cref{le:at-most-two-points}.}
        \label{fi:M-diagram-c}
    \end{figure}  
    
    Suppose that there exists one upper peak $p_l=(i_l,j_l)$ whose corresponding edge $e_l$ has at most one conflict. This means that the lower-right wedge of $p_l$ contains at most one used point. The polygon can exist only if in the lower-right wedge of $p_l$ there is a left turn (i.e., an angle of $\pi/2$ inside the polygon); such a left turn can only be a lower peak $q_g=(i_g,j_g)$. So, $e_l$ has at least one conflict. Assume hence that $e_l$ has exactly one conflict, i.e., that the lower-right wedge of $p_l$ contains no used point other than $q_g$. This means that the two lower peaks $q_{g-1}$ and $q_{g+1}$ must be outside the lower-right wedge of $p_l$ and therefore they have coordinates $(i_{g-1},j_{g-1})$ and $(i_{g+1},j_{g+1})$ such that $i_{g-1} \leq i_l$ and $j_{g+1} \geq j_{l}$; see \cref{fi:M-diagram-c}~(left). Since the upper and lower frontiers do not cross, this implies that the two upper peaks  $p_{l-1}$ and $p_{l+1}$ have coordinates $(i_{l-1},j_{l-1})$ and $(i_{l+1},j_{l+1})$ such that $j_{l-1} > j_{g}$ and $i_{l+1}<i_{g}$. That is $p_{l-1}$, $p_l$, and $p_{l+1}$ are all in the upper-left wedge of $q_g$ and therefore the edge $e_g$ corresponding to $q_g$ has three conflicts. But this contradicts the fact that $\Gamma$ is 2-defective.  

    With a similar argument it is possible to prove that at least one between the edges that correspond to $p_{1}$ and to $q_1$ has exactly two conflicts in $\Gamma$ while the other one has exactly one conflict; see \cref{fi:M-diagram-c}~(middle). Symmetrically, we can prove that at least one between the edges that correspond to $p_a$ and $q_b$ has exactly two conflicts in $\Gamma$ while the other one has exactly one conflict; see \cref{fi:M-diagram-c}~(right).

    So every peak in $P$, except at most two, correspond to edges that already have the maximum number of conflicts. Since any point of $P$ is in the upper-left wedge of some lower peak or in the lower-right wedge of some upper peak, it follows that at most two points can be used that are neither peaks nor valleys. 
\end{proof}

\begin{lemma}\label{le:all-rectangles}
Let $\Gamma$ be a 2-defective 1-queue layout with $m$ edges. There exists a 2-defective 1-queue layout $\Gamma'$ with $m'$ edges such that $m \leq m'$, $M(\Gamma')$ is standard, and the non-degenerate polygons of $M(\Gamma')$ are all rectangles.
\end{lemma}
\begin{proof}
    By \cref{le:no-crossing-frontiers,le:shared-used}, we can assume that $M(\Gamma)$ is standard and thus we can consider its polygons. If the non-degenerate polygons of $M(\Gamma)$ are all rectangles then $\Gamma'=\Gamma$. Assume then that there is a non-degenerate polygon $P$ that is not a rectangle. Let $p$ be the valley that $P$ shares with the polygon preceding it. Let $q_1$ and $q_2$ be the first and the second corner of $P$ encountered when walking along the upper sub-frontier of $P$ starting from $p$. Notice that $q_1$ is an upper peak that has the same $x$-coordinate as $p$, while $q_2$ is not a peak. Analogously, let $r_1$ and $r_2$ be the first and the second corner of $P$  encountered when walking along the lower sub-frontier of $P$ starting from $p$. In this case $r_1$ is a lower peak that has the same $y$-coordinate as $p$, while $r_2$ is not a peak. We now describe how to modify $\Gamma$ to produce a new layout $\Gamma''$ where the polygon $P$ is replaced by a rectangle and a smaller non-degenerate polygon $P'$, i.e., a non-degenerate polygon whose upper and lower sub-frontiers have less points. There are three cases depending of the $x$-coordinate and the $y$-coordinate of $q_2$ and $r_2$. We denote the $x$- and $y$-coordinate of a point $p$ as $x(p)$ and $y(p)$, respectively.

\begin{figure}[tb]
    \centering
    \begin{subfigure}{.5\textwidth}
        \centering
        {\includegraphics[page=4,width=\textwidth]{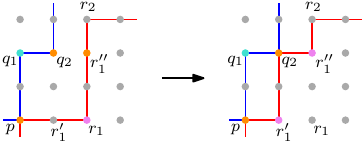}
        \subcaption{Case 1.}{\label{fi:M-diagram-d-a}}}
    \end{subfigure}
    \hfill
    \begin{subfigure}{.5\textwidth}
        \centering
        {\includegraphics[page=5,width=\textwidth]{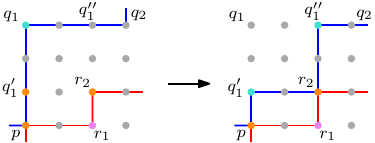}
        \subcaption{Case 2.}{\label{fi:M-diagram-d-b}}}
    \end{subfigure}
    \hfill
    \begin{subfigure}{.5\textwidth}
        \centering
        {\includegraphics[page=6,width=\textwidth]{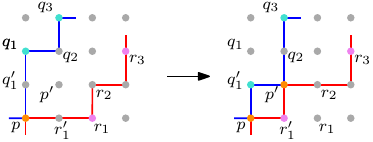}
        \subcaption{Case 3.}{\label{fi:M-diagram-d-c}}}
    \end{subfigure}
    \caption{Illustration for the proof of \cref{le:all-rectangles}.}
    \label{fi:M-diagram-d}
\end{figure}

    \begin{description}
        \item[Case 1: $x(q_2)<x(r_2)$ and $y(r_2) > y(q_2)$.] Refer to \cref{fi:M-diagram-d-a}. In this case we replace the lower peak $r_1$ with two new lower peaks $r'_1$ of coordinates $(x(q_2),y(r_1))$ and $r''_1$ of coordinates $(x(r_1),y(q_2))$. This replacement means that $r_1$ is no longer used (the corresponding edge is removed from $\Gamma$) and $r'_1$ and $r''_1$ are used in $\Gamma''$ (notice that each of them may have already been used in $\Gamma$). After this change the point $q_2$ is shared by the upper and lower frontiers, i.e., it is a valley of $M(\Gamma'')$ and, by \cref{le:shared-used}, it can also be used. The three points $q_2$, $r'_1$, and $r''_1$ are not peaks nor valleys in $M(\Gamma)$ and therefore, by \cref{le:at-most-two-points}, at most two of them may have already been used in $M(\Gamma)$. On the other hand, the point $r_1$ was used in $M(\Gamma)$ and it is no longer used in $M(\Gamma'')$. So at most three among  $q_2$, $r'_1$, $r''_1$, and $r_1$ were used in $M(\Gamma)$, while exactly three of these points are used in $M(\Gamma'')$, namely $q_2$, $r'_1$, and $r''_1$. Hence, the number of used points in $M(\Gamma'')$ is at least the number of used points in $M(\Gamma)$.  
        
        As observed above, the point $q_2$ is a valley in $M(\Gamma'')$; this means that the polygon $P$ is replaced by the rectangle defined by $p$, $q_1$, $q_2$, and $r'_1$, and by a new non-degenerate polygon $P'$ that contains all the remaining upper and lower peaks of $P$ (including the new peak $r''_1$). 
        
        \item[Case 2: $x(q_2) > x(r_2)$ and $y(r_2) < y(q_2)$.]
         Refer to \cref{fi:M-diagram-d-b}. In this case we replace the upper peak $q_1$ with two new upper peaks $q'_1$ of coordinates $(x(q_1),y(r_2))$ and $q''_1$ of coordinates $(x(r_2),y(q_1))$. This replacement means that $q_1$ is no longer used (the corresponding edge is removed from $\Gamma$) and $q'_1$ and $q''_1$ are used in $\Gamma''$ (notice that each of them may have already been used in $\Gamma$). After this change the point $r_2$ is shared by the upper and lower frontiers, i.e., it is a valley of $M(\Gamma'')$ and, by \cref{le:shared-used}, it can also be used. The three points $r_2$, $q'_1$, and $q''_1$ are not peaks nor valleys in $M(\Gamma)$ and therefore, by \cref{le:at-most-two-points}, at most two of them may have already been used in $M(\Gamma)$. On the other hand, the point $q_1$ was used in $M(\Gamma)$ and it is no longer used in $M(\Gamma'')$. So at most three among  $r_2$, $q'_1$, $q''_1$, and $q_1$ were used in $M(\Gamma)$, while exactly three of these points are used in $M(\Gamma'')$, namely $r_2$, $q'_1$, and $q''_1$. Hence, the number of used points in $M(\Gamma'')$ is at least the number of used points in $M(\Gamma)$.  
        
        As observed above, the point $r_2$ is a valley in $M(\Gamma'')$; this means that the non-degenerate polygon $P$ is replaced by the rectangle defined by $p$, $r_1$, $r_2$, and $q'_1$, and by a new non-degenerate polygon $P'$ that contains all the remaining upper and lower peaks of $P$ (including the new peak $q''_1$). 

        \item[Case 3: $x(q_2) < x(r_2)$ and $y(r_2) < y(q_2)$.] Notice that, since $P$ is not a rectangle it cannot be $x(q_2)=x(r_2)$ and $y(r_2)=y(q_2)$ at the same time. Refer to \cref{fi:M-diagram-d-c}. In this case we replace the lower peak $r_1$ with a new lower peak $r'_1$ of coordinates $(x(q_2),y(r_1))$ and the upper peak $q_1$ with a new upper peak $q'_1$ of coordinates $(x(q_1),y(r_2))$. This replacement means that $r_1$ and $q_1$ are no longer used (the corresponding edges are removed from $\Gamma$) and $r'_1$ and $q'_1$ are used in $\Gamma''$ (notice that each of them may have already been used in $\Gamma$). After this change the point $p'=(x(q_2),y(r_2))$ is shared by the upper and lower frontiers, i.e., it is a valley of $M(\Gamma'')$ and, by \cref{le:shared-used}, it can also be used. 
        
        The two points $q'_1$ and $q_2$  are not used in $M(\Gamma)$ as otherwise the upper-left wedge of $r_1$ would contain more than two points: the upper peak $q_1$, another upper peak $q_3$ (having the same $x$-coordinate as $q_2$) and $q'_1$ or $q_2$. Analogously, the two points $r'_1$ and $r_2$ are not used in $M(\Gamma)$ as otherwise the lower-right wedge of $q_1$ would contain more than two points: the lower peak $r_1$, another lower peak $r_3$ (having the same $y$-coordinate as $r_2$) and $r'_1$ or $r_2$. Thus, among the seven points $q_1$, $q'_1$, $q_2$, $r_1$, $r'_1$, $r_2$, and $p'$, exactly two are used in $M(\Gamma)$, while exactly three are used in $M(\Gamma'')$, namely $q'_1$, $r'_1$, and $p'$. Hence, the number of used points in $M(\Gamma'')$ is larger than the number of used points in $M(\Gamma)$.
        
        As observed above, the point $p'$ is a valley in $M(\Gamma'')$; this means that the non-degenerate polygon $P$ is replaced by the rectangle defined by $p$, $q'_1$, $r'_1$, and $p'$, and by a new non-degenerate polygon $P'$ that contains all the remaining upper and lower peaks of $P$.  

        \item[Case 4: $x(q_2) = x(r_2)$ or $y(r_2) = y(q_2)$.] Notice that, since $P$ is not a rectangle it cannot be $x(q_2)=x(r_2)$ and $y(r_2)=y(q_2)$ at the same time. We proceed as in the previous case. The only difference is that: $(i)$ if $x(q_2) = x(r_2)$, we have  $r'_1 \equiv r_1$ and $p' \equiv r_2$; $(ii)$ if $y(r_2) = y(q_2)$, we have $q'_1 \equiv q_1$ and $p' \equiv q_2$.
    \end{description}

    In each of the cases above, the non-degenerate polygon $P$ is replaced by a rectangle and a non-degenerate polygon~$P'$ that is smaller than $P$ (that is, whose upper and lower frontiers have less points). By repeatedly applying the above modification, we replace $P'$ with a set of rectangles. Repeating the same procedure for all non-degenerate polygons that are not rectangles, one obtains the desired layout $\Gamma'$.
\end{proof}

A \emph{canonical $M$-diagram} is a standard $M$-diagram with the additional property that every non-degenerate polygon is a rectangle. A \emph{maximal canonical $M$-diagram} is a canonical $M$-diagram such that  no unused point can become used while keeping the layout $2$-defective.

Let $\pi$ be a path. An \emph{Italian coloring} of $\pi$ is a coloring of the vertices of $\pi$ with colors \emph{green}, \emph{white} and \emph{red} such that: (i) between any pair of consecutive red vertices there are at most two green vertices; (ii) at the beginning and at the end of the path there are at least two consecutive red vertices. The \emph{weight} of an Italian coloring of a path is $2n_g+n_r$, where $n_g$ is the number of green vertices and $n_r$ is the number of red vertices. 

\begin{lemma}\label{le:italian-coloring}
Let $\Gamma$ be a $2$-defective $1$-queue layout with $n$ vertices and $m$ edges, and let $M(\Gamma)$ be maximal canonical. Then, there exists a path $\pi$ with $2n-3$ vertices that admits an Italian coloring whose weight is equal to $m$.
\end{lemma}
\begin{proof}
    The upper and lower frontiers of $M(\Gamma)$ have $2n-3$ points. Namely, $M(\Gamma)$ has $n-1$ rows and $n-1$ columns. Both the lower and the upper frontiers have $n-2$ vertical edges and $n-2$ horizontal edges (one edge for every pair of consecutive rows/columns); thus the number of edges is $2n-4$ and the number of vertices is $2n-3$. 

    \begin{figure}
        \centering
        {\includegraphics[page=7,width=.5\textwidth]{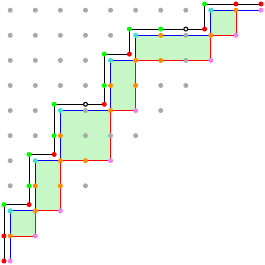}}
        \caption{Bijection between the vertices of the upper frontier of an M-diagram and the vertices of a path $\pi$. The obtained coloring of the vertices of $\pi$ is an Italian coloring.}
        \label{fi:italian-coloring}
    \end{figure} 
    
    Let $\pi$ be a path with $2n-3$ vertices and bijectively map the vertices of $\pi$ to the points of the upper frontier. Color with red all the vertices that correspond to the valleys, with green the vertices that correspond to all the other used points, and with white the vertices that correspond to all the unused points; refer to \cref{fi:italian-coloring}. The corresponding coloring is a valid Italian coloring. In fact, there cannot be three green vertices between two consecutive red vertices. If so, there would be a rectangle whose upper sub-frontier contains three used points, but this means that the lower peak of this rectangle has three conflicts. We now claim that $m=2n_g+n_r$. To this aim, it is sufficient to prove that for every point corresponding to a green vertex there is another point in the lower frontier that is used and not shared with the upper frontier. Let $R$ be a rectangle of $M(\Gamma)$. If $R$ is a $1 \times 1$ rectangle, then it has only one point corresponding to a green vertex: its upper peak. In the lower frontier there is the lower peak that is used and not shared with the upper frontier. If $R$ is not a $1 \times 1$ rectangle, then it has two points each corresponding to a green vertex: one is the upper peak and the other one is a point that is not a peak nor a valley. Both these points belong to the upper sub-frontier of $R$; since $R$ is a rectangle, the number of points of its lower sub-frontier is equal to the number of points of its upper sub-frontier. Thus, it contains at least another point other than the lower peak. Let $\lambda$ be the set of points in the lower sub-frontier of $R$ that are not shared with the upper frontier and are distinct from the peak. It is easy to see that any of the points in $\lambda$ can be a used point (i.e., using it does not cause a used vertex to have three conflicts). Since $M(\Gamma)$ is maximal, one point in $\lambda$ is used. 
\end{proof}

\begin{lemma}\label{le:italian-weight}
    Let $\pi$ be a path with $2n-3$ vertices having an Italian coloring of weight $\omega$. Then $\omega \leq \left(\frac{10}{3}n-\frac{21+r}{3}\right)$, where $r=n \mod 3$, and this bound is tight for every $n=3s+r$ and $s\in \mathbb{N}^+$.
\end{lemma}
\begin{proof}
    Let $\pi$ be a path with $2n-3$ vertices having an Italian coloring $\chi$ of weight $\omega$. We first claim that if there are white vertices in $\pi$ we can find a new coloring with no white vertex and larger weight than $\omega$. Let $\sigma$ be the sequence of red and green vertices along the path; we can re-color $\pi$ so that the first $|\sigma|$ vertices are colored according to $\sigma$ and the remaining vertices are white (in other words, we move the white vertices to the end of $\pi$). Re-coloring the white vertices red, we obtain a valid Italian coloring $\chi'$ whose weight $\omega'$ is larger than $\omega$. Namely, the number of green vertices is unchanged while the number of red vertices has increased. So we can assume that $\chi$ has no white vertex. 
    We now claim that the maximum weight that can be obtained by using only red and green colors is the one in the statement. The first two and the last two vertices of $\pi$ are necessarily colored red in any Italian coloring. Let $n'=2n-7$ be the number of remaining points and let $n'_r=n_r-4$. Since $n_g+n_r=2n-3$, we have $n_g+n'_r=n'$. We consider the following cases:
    \begin{itemize}
        \item Case 1:  $n'=3s'$ for some $s' \in \mathbb{N}^+$. The value of the weight is $2n_g+n_r=2n_g+n'_r+4$. Further, since $n'$ is a multiple of $3$ and at most two green vertices can exist between two red vertices, it must be $n'_r \geq \frac{n_g}{2}$. The maximum value for $2n_g+n'_r+4$ subject to the constraints $n_g+n'_r=n'$ and $n'_r \geq \frac{n_g}{2}$ is achieved when $n_g$ is maximum, which in this case happens when $n_g=\frac{2n'}{3}$ and $n'_r=\frac{n'}{3}$. In this case the weight of $\chi$ is $\frac{4n'}{3}+\frac{n'}{3}+4=\frac{5n'+12}{3}=\frac{5(2n-7)+12}{3}=\frac{10n}{3}-\frac{23}{3}$. Notice that if $n'=3s'$, we have $n=3s+2$, for  $s=\frac{s'+1}{2}$, and thus the statement holds when $n'=3s'$.
        \item Case 2:  $n'=3s'+1$ for some integer $s' \in \mathbb{N}^+$. The argument is similar to the previous one, but in this case the maximum of $2n_g+n'_r+4$ is achieved when $n_g = \frac{2(n'-1)}{3}+1$ and $n'_r=\frac{(n'-1)}{3}$. In this case the weight of $\chi$ is $\frac{4(n'-1)}{3}+\frac{(n'-1)}{3}+6=\frac{5n'+13}{3}=\frac{5(2n-7)+13}{3}=\frac{10n}{3}-\frac{22}{3}$. In this case, $n'=3s'+1$ implies $n=3s+1$ for $s=\frac{s'}{2}+1$ (with $s'$ even) and the statement holds.
        \item Case 3: $n'=3s'+2$ for some $s' \in \mathbb{N}^+$. The argument is similar to the previous cases, but in this case the maximum of $2n_g+n'_r+4$ is achieved when $n_g=\frac{2(n'-2)}{3}+2$ and $n'_r=\frac{(n'-2)}{3}$. The weight is now $\frac{4(n'-2)}{3}+4+\frac{(n'-2)}{3}+4=
        \frac{5n'+14}{3}=\frac{5(2n-7)+14}{3}=\frac{10n}{3}-\frac{21}{3}$. In this case $n'=3s'+2$ implies $n=3s$ for $s=\frac{s'+3}{2}$ (with $s'$ odd) and the statement holds.\qedhere
\end{itemize}
\end{proof} 

\begin{lemma}\label{le:density-k2-h1-upper}
Any $n$-vertex graph that admits a $2$-defective $1$-queue layout has at most $\left(\frac{10}{3}n-\frac{21+r}{3}\right)$ edges, where $n \geq 3$ and $r=n \mod 3$.
\end{lemma}
\begin{proof}
    By \cref{le:no-crossing-frontiers,le:shared-used,le:at-most-two-points,le:all-rectangles}, we can assume that $\Gamma$ is such that $M(\Gamma)$ is maximal canonical. Assume that the number of edges of $\Gamma$ is larger than the values in statement.  
    By \cref{le:italian-coloring}, there exists a path $\pi$ with $2n-3$ vertices that has an Italian coloring whose weight is $m$ and therefore larger than the values in the statement; but this contradicts \cref{le:italian-weight}.
\end{proof}

\noindent Now, we are ready to prove \cref{th:density-1page-2defect}.

\begin{figure}[tb]
    \centering
    \begin{subfigure}{.45\textwidth}
        \centering        {\includegraphics[page=1,width=.85\textwidth]{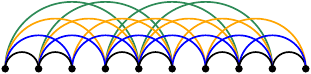}
        \subcaption{\label{fi:density-1-page-2-defective-a}}}
    \end{subfigure}
    \hfil
    \begin{subfigure}{.45\textwidth}
        \centering
        {\includegraphics[page=2,width=.85\textwidth]{figures/density-1-page-2-defective.pdf}
        \subcaption{\label{fi:density-1-page-2-defective-b}}}
    \end{subfigure}
    \caption{Illustration for the proof of \cref{th:density-1page-2defect}. (a) A $2$-defective $1$-queue layout of a graph with $n=10$ vertices. (b) The conflict graph of the layout in \cref{fi:density-1-page-2-defective-a}.}
    \label{fi:density-1-page-2-defective}
\end{figure}

\begin{proof}[Proof of \cref{th:density-1page-2defect}]
The upper bound of the statement is given by \cref{le:density-k2-h1-upper}.
We now show that for every integer $s \geq 1$, there exists an $n$-vertex graph with $n = 3s+r$ vertices and $\frac{10}{3}n - \frac{21+r}{3}$ edges that admits a $2$-defective $1$-queue layout.

Let $s$ be any positive integer and let $n=3s+r$. We show a graph $G$ with $n$ vertices and $\frac{10}{3}n - \frac{21+r}{3}$ edges that admits a $2$-defective $1$-queue layout. We describe a $2$-defective $1$-queue layout of $G$ for each of the three possible values of $r$.
Refer to \cref{fi:density-1-page-2-defective-a} for an illustration of a graph with $10$ vertices (i.e., $s=3$ and $r=1$).
The layout of $G$ consists of:
        \begin{itemize}
        \item $\frac{2}{3}(n-2-r'+1)=\frac{2}{3}n - \frac{1}{3} + \frac{r'}{3}$ 1-hop edges (black in \cref{fi:density-1-page-2-defective-a}), where $r' = (n-2) \mod 3$. Consider the complete sequence of $(n-1)$ 1-hop edges from left to right. We delete the last edge from each triple of consecutive 1-hop edges, except for the very last edge of the sequence;
        \item $(n-2)$ 2-hop edges (blue in \cref{fi:density-1-page-2-defective-a});
        \item $(n-3)$ 3-hop edges (orange in \cref{fi:density-1-page-2-defective-a});
        \item $\frac{2}{3} (n-4-r'') + r'' = \frac{2}{3}n - \frac{8}{3} + \frac{r''}{3}$ 4-hop edges (green in \cref{fi:density-1-page-2-defective-a}), where $r'' = (n-4) \mod 3$. Consider the complete sequence of $(n-4)$ 4-hop edges from left to right. We delete the last edge from each triple of 4-hop edges.
        \end{itemize}
        
Observe that: $(i)$ when $r=0$, we have $r'=1$ and $r''=2$; $(ii)$ when $r=1$, we have $r'=2$ and $r''=0$; $(iii)$ when $r=2$, we have $r'=0$ and $r''=1$.
The graph has in total $\frac{2}{3}n - \frac{1}{3} +\frac{r'}{3} + n-2 + n-3 + \frac{2}{3}n - \frac{8}{3} +\frac{r''}{3} = \frac{10}{3}n - \frac{24}{3} + \frac{r'}{3} + \frac{r''}{3}$ = $\frac{10}{3}n - \frac{21+r}{3}$ edges. 
        
By construction, there is no nesting relationship between $p$-hop edges and $(p+1)$-hop edges, for any $p \in \{1,2,3\}$. Furthermore: $(i)$ each 1-hop edge is in a nesting relationship with at most two edges, one being a 3-hop edge and the other being a 4-hop edge; $(ii)$ each 2-hop edge is in a nesting relationship with at most one edge, which is a 4-hop edge; $(iii)$ each 3-hop edge is in a nesting relationship with at most one edge, which is a 1-hop edge; $(vi)$ each 4-hop edge is in a nesting relationship with at most two edges, one being a 2-hop edge and the other being a 1-hop edge. \cref{fi:density-1-page-2-defective-b} depicts the conflict graph of the layout in \cref{fi:density-1-page-2-defective-a}, which shows the nesting relationships between the different types of edges (the conflict graph has vertex-degree at most two).
Hence, the described layout for $G$ is a $2$-defective $1$-queue layout. 
\end{proof}

\subsection{Defective Stack Layouts}\label{sse:stack-density}

As we will argue in~\cref{sse:stack-defective-page-number}, a graph has $k$-defective stack number $1$ if and only if it is outer $k$-planar (see \cref{th:charact-k-defective-stack}).
This characterization naturally extends the corresponding one by Bernhart and Kainen~\cite{DBLP:journals/jct/BernhartK79} for $k=0$. 
Recall that a graph is outer $k$-planar if it admits a drawing in which all vertices are on the outer face and each edge has at most $k$ crossings.
For $k \geq 1$, it also implies the next theorem using known density results for outer~\cite{DBLP:journals/algorithmica/AuerBBGHNR16,abrego2024book} and general~\cite{DBLP:journals/combinatorica/PachT97} $k$-planar graphs.

\begin{theorem}[\cite{DBLP:journals/algorithmica/AuerBBGHNR16},\cite{DBLP:journals/combinatorica/PachT97},\cite{abrego2024book}]\label{th:k-defective-1-stack-density}
Any $n$-vertex graph that admits a $k$-defective $1$-stack layout has at most: 
\begin{enumerate*}[label=(\roman*)]
    \item $\frac{5}{2}n - 4$ edges, if $k=1$,
    \item $3n-5$ edges, if $k=2$,
    \item $\frac{13}{4}n-\frac{11}{2}$ edges, if $k=3$,
    \item $O(\sqrt{k}n)$, otherwise.
\end{enumerate*}
Also, the first three bounds are~tight.
\end{theorem}

\begin{corollary}\label{th:1-defective-stack-density}
Any $n$-vertex graph that admits a $k$-defective $h$-stack layout has at most 
\begin{enumerate*}[label=(\roman*)]
    \item $(\frac{3}{2}h+1)n - 4h$ edges, if $k=1$,
    \item $(2h+1)n-5h$ edges, if $k=2$,
    \item $(\frac{9}{4}h+1)n-\frac{11}{2}h$, if $k=3$,
    \item $O(\sqrt{k}hn)$, otherwise.
\end{enumerate*}
\end{corollary}

\begin{proof}
We give the proof for the case $k=1$; the cases $k=2$ and $k=3$ are analogous, while the case $k>3$ trivially follows from \cref{th:charact-k-defective-stack}.
Let $G$ be an $n$-vertex graph admitting a $1$-defective $h$-stack layout $\mathcal{L}$. Let $v_0,\ldots,v_{n-1}$ be the vertices of $G$ as they appear in $\mathcal{L}$. We say that an edge $e$ of $G$ \emph{connects consecutive vertices} in $\mathcal{L}$ if there exists an $i$ in $\{0,\ldots,n-1\}$ such that $e=(v_i,v_{i+1})$; indices are taken$\mod n$. By \cref{th:charact-k-defective-stack}, it follows that each defective stack of $\mathcal{L}$ is an outer 1-planar graph on $n$ vertices. By neglecting the (at most $n$) edges connecting consecutive vertices in $\mathcal{L}$, it follows that each defective stack of $\mathcal{L}$ has at most $\frac{3}{2}n - 4$ edges by \cref{th:k-defective-1-stack-density}. Since $\mathcal{L}$ has $h$ defective stacks in total, this gives at most $\frac{3}{2}hn-4h$ edges in $G$. Thus, $G$ has at most $(\frac{3}{2}h+1)n-4h$ edges as claimed, since the edges connecting consecutive vertices in $\mathcal{L}$ are at most $n$.
\end{proof}

We remark that a tight bound on the maximum number of edges of $n$-vertex graphs that admit classic stack layouts (i.e., when $k=0$) with $h \geq 1$ stacks is $(h+1)n - 3h$~\cite{DBLP:journals/jct/BernhartK79,DBLP:journals/dmtcs/DujmovicW04}.

\section{Defective Queue and Stack Number for Some Graph Classes}\label{se:defective-page-number}

In this section we prove upper and lower bounds on the defective queue and stack number of some notable graph classes.

\subsection{Defective Queue Number}\label{sse:queue-defective-page-number}

It is known that outerplanar graphs have queue number $2$~\cite{DBLP:conf/cocoon/RengarajanM95}, i.e., there exist outerplanar graphs that do not admit a $1$-queue layout and $2$ queues always suffice. It is natural to ask whether a constant defectiveness is enough to lay out every outerplanar graph on one defective queue, i.e., whether every outerplanar graph has $k$-defective queue number $1$ for $k=O(1)$. Unfortunately, this is not always the case. The next lemma shows that there exist outerplanar graphs requiring linearly many defects in any such $k$-defective $1$-queue layout. 

\begin{lemma}\label{le:outerplanar-defectiveness-1-page}
For every integer $j \geq 3$, there exists an outerplanar graph $G$ with $n = 3j+1$ vertices such that every defective $1$-queue layout of $G$ has defectiveness at least $\frac{n-1}{3} - 2$. 
\end{lemma}

\begin{proof}
Given any integer $j \geq 3$, consider the outerplanar graph $G$ with $n = 3j+1$ vertices consisting of a path of $n-1$ vertices, called the \emph{backbone} of $G$, plus a vertex, called the \emph{apex} of $G$, connected to all the vertices of the backbone. The edges of $G$ that connect the apex to the backbone are called the \emph{fan edges} while the edges of the backbone are called the \emph{backbone edges}. \cref{fi:fan-a} shows an example of $G$ for $j=4$, and hence $n = 13$. The apex is the red vertex with label 0; the backbone edges are bold and the fan edges are gray. Note that $G$ consists in total of $m = 2n-3$ edges, namely $n-1$ fan edges and $n-2$ backbone edges.

Let $\mathcal{L}$ be any linear layout of $G$ on one defective queue. We show that $\mathcal{L}$ necessarily contains a $k$-defect with $k \geq \frac{(n-1)}{3} - 1$. Denote by $a$ the apex of $G$, and consider the position of $a$ in the linear order $\prec$ of $\mathcal{L}$. Denote by $E_m$ the subset of backbone edges that connect a vertex to the left of $a$ with a vertex to the right of $a$ (i.e., each of these edges cuts the vertical line passing through $a$). Also, let $E_l$ be the subset of backbone edges that lie entirely to the left of $a$, and let $E_r$ be the subset of backbone edges that lie entirely to the right of $a$.  Clearly, $E_m$, $E_l$, and $E_r$ form a partition of the backbone edges, thus $|E_m| + |E_l| + |E_r| = n-2$.
Observe that, if $|E_m| \geq 2$ then $\mathcal{L}$ contains a $k$-defect with $k \geq |E_m| - 2$. Indeed, if $x$ is the vertex that immediately precedes (or that immediately follows) $a$ in the linear order $\prec$, then the fan edge connecting $x$ to $a$ is nested in at least $|E_m| -2$ edges of $E_m$ (at most two of the edges in $E_m$ are incident to $x$). Also, if $|E_l| \geq 2$ then $\mathcal{L}$ contains a $k$-defect with $k \geq |E_l| - 2$. Indeed, let $y$ be the first vertex of the linear order $\prec$ (i.e., the leftmost vertex of $\mathcal{L}$). There are at least $|E_l|-2$ edges of $E_l$ that are nested in the edge $(a,y)$ (at most two of the edges in $E_l$ are incident to $y$). With a symmetric argument, if $|E_r| \geq 2$ then $\mathcal{L}$ contains a $k$-defect with $k \geq |E_r| - 2$. Therefore, $\mathcal{L}$ contains a $k$-defect such that $k \geq M = \max\{|E_l|-2, |E_m|-2, |E_r|-2\}$. Since $n-2$ is not divisible by 3 (whereas $n-1$ is), the value of $M$ is minimized when two among $E_l, E_m$, and $E_r$ have the same size, namely size $\frac{n-1}{3}$, and the other one has one fewer element, namely it has size $\frac{n-1}{3}-1$; in this case $M=\frac{n-1}{3}-2$, which proves the claim.

A construction of a linear layout that matches $k = \frac{n-1}{3}-2$ for the graph in \cref{fi:fan-a} is shown in \cref{fi:fan-b}. In this construction $|E_l|=|E_r|=\frac{n-1}{3}$ and $|E_m|=\frac{n-1}{3}-1$. Also, both the leftmost and the rightmost vertices of the layout are degree-2 vertices of the backbone.
\end{proof}

\begin{figure}[tb]
    \centering
    \begin{subfigure}{.48\textwidth}
        \centering
        {\includegraphics[page=1,width=.8\textwidth]{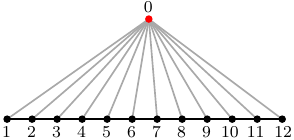}}
        \subcaption{\label{fi:fan-a}}
    \end{subfigure}
    \hfil
    \begin{subfigure}{.48\textwidth}
        \centering
        {\includegraphics[page=2,width=.8\textwidth]{figures/fan.pdf}}
        \subcaption{\label{fi:fan-b}}
    \end{subfigure}   
    \caption{Illustration for the proof of \cref{le:outerplanar-defectiveness-1-page}.}
    \label{fi:fan}
\end{figure}

On the positive side, we prove that outer $1$-planar graphs, which include the outerplanar ones, have $1$-defective queue number $2$.

\begin{theorem}\label{th:outer-1-planar}
Outer $1$-planar graphs have $1$-defective queue number $2$.
\end{theorem}

\begin{proof}
The lower bound follows from \cref{le:outerplanar-defectiveness-1-page}.
To prove the upper bound, we adjust a construction by Felsner, Liotta, and Wismath~\cite{DBLP:journals/jgaa/FelsnerLW03} and independently by Bannister, Devanny, Dujmović, Eppstein, and Wood~\cite{DBLP:journals/algorithmica/BannisterDDEW19} that yields planar straight-line drawings of biconnected outerplanar graphs, so-called \emph{weakly leveled planar drawings}, in which the vertices are drawn on a set of pairwise distinct horizontal lines (called \emph{levels}), the edges of the same level (\emph{span-$0$ edges}) form a forest of paths, while the remaining edges connect vertices on consecutive levels (\emph{span-$1$ edges}); variants have also been studied by Dujmović, P{\'o}r, and Wood~\cite{DBLP:journals/dmtcs/DujmovicPW04} and by Bekos, Gronemann, and Raftopoulou~\cite{DBLP:journals/algorithmica/BekosGR23}. 

Let $G$ be an outer $1$-planar graph. Without loss of generality, we assume that $G$ is maximal, that is, the addition of any edge yields a graph that is not outer $1$-planar. Let~$G_p$ be the \emph{planar skeleton} of $G$, i.e., the subgraph of $G$ obtained by removing all pairs of crossing edges. Since $G$ is maximal, both $G$ and $G_p$ are biconnected~\cite{DBLP:journals/algorithmica/AuerBBGHNR16}. Further, each bounded face of~$G_p$ is either of length~$4$ (if it contains a pair of crossing edges in $G$) or of length~$3$ (otherwise). We present a recursive construction to produce a planar straight-line drawing~$\Gamma(G_p)$ of $G_p$, whose bounded faces are convex, that maintains the following invariants:

\begin{enumerate} [label=I.\arabic*\upshape,ref=I.\arabic*]
\item\label{inv:1} The cycle delimiting the outer face consists of two strictly $x$-monotone paths, referred to as \emph{upper} and \emph{lower envelopes}, respectively.
\item\label{inv:2} The $y$-coordinates of the endvertices of each edge differ either by $0$ (span-$0$ edge) or by $1$ (span-$1$ edge). 
\end{enumerate}

In the base of the recursion, $G_p$ consists either of a solely triangular face or of a solely quadrangular face that can be trivially drawn convex with span-$0$ and span-$1$ edges satisfying Invariants~\ref{inv:1} and~\ref{inv:2}. Assume now that $G_p$ has at least two bounded faces. Let $f$ be a bounded face corresponding to a leaf of the weak dual\footnote{The weak dual of a biconnected outerplanar graph $H$ is a tree having a vertex for each bounded face of the outerplanar embedding of $H$, while an edge between two vertices $u$ and $v$ exists in the weak dual if and only if the faces corresponding to $u$ and $v$ are adjacent in the outerplanar embedding of $H$.} of $G_p$. If $f$ is a face of length~$3$, then it contains a vertex $v$ of degree $2$. Otherwise (namely, $f$ is a face of length~$4$), it contains two adjacent vertices $v_1$ and $v_2$ of degree~$2$. In the former case, we remove $v$ from $G_p$, while in the latter case we remove $v_1$ and $v_2$ from $G_p$. This yields a new outerplanar graph $G_p'$ with fewer vertices, which admits a planar straight-line drawing $\Gamma(G_p')$ satisfying Invariants~\ref{inv:1} and \ref{inv:2}, whose bounded faces are convex. By Invariant~\ref{inv:1}, none of the edges on the outer face is drawn vertically in $\Gamma(G_p')$. To obtain drawing $\Gamma(G_p)$ of $G_p$, we proceed as follows. 

Assume that $G_p$ was obtained by removing a single degree-$2$ vertex $v$ (two adjacent degree-$2$ vertices $v_1$ and $v_2$) from $G$. Let $u$ and $w$ be the neighbors of $v$ (of $v_1$ and $v_2$) in $G_p$. It follows that $(u,w)$ is an edge of $G$ that lies on the outer face of $\Gamma(G_p')$. By Invariant~\ref{inv:2}, $(u,w)$ is either a span-$0$ or a span-$1$ edge. Assume that $(u,w)$ is along the upper envelope of $\Gamma(G')$; the case in which it is along the lower envelope is symmetric. Assume first that $(u,w)$ is a span-$1$ edge in $\Gamma(G_p')$ and, without loss of generality, that $y(u)=y(w)+1$; see \cref{fig:outer-1-planar-queue-1,fig:outer-1-planar-queue-2}. Then, we insert vertex $v$ (vertices $v_1$ and $v_2$) into $\Gamma(G_p')$ by setting $y(v)=y(u)$ and $x(v)=\frac{1}{2}x(u)+\frac{1}{2}x(w)$ ($y(v_1)=y(v_2)=y(u)$, $x(v_1)=\frac{2}{3}x(u)+ \frac{1}{3}x(w)$ and $x(v_2)=\frac{1}{3}x(u)+ \frac{2}{3}x(w)$). Hence, the newly introduced edges are span-$0$ and span-$1$ edges, which implies that Invariant~\ref{inv:2} is maintained. Since (the non-vertical) edge $(u,w)$ of the upper envelope of $\Gamma(G_p')$ is replaced by the $x$-monotone path $u \rightarrow v \rightarrow w$ ($u \rightarrow v_1 \rightarrow v_2 \rightarrow w$) in $\Gamma(G_p)$, Invariant~\ref{inv:1} is also maintained. Assume now that $(u,w)$ has span $0$ in $\Gamma(G_p')$; see \cref{fig:outer-1-planar-queue-2,fig:outer-1-planar-queue-4}. Then, we insert vertex $v$ (vertices $v_1$ and $v_2$) into $\Gamma(G_p')$ such that $y(v)=y(u)+1$ and $x(v)=\frac{1}{2}x(w)+\frac{1}{2}x(u)$ ($y(v_1)=y(v_2)=y(u)+1$, $x(v_1)=\frac{2}{3}x(u)+ \frac{1}{3}x(w)$ and $x(v_2)=\frac{1}{3}x(u)+ \frac{2}{3}x(w)$). Similarly to the previous case, one can argue that Invariants~\ref{inv:1} and~\ref{inv:2} are maintained. Also, it is not difficult to see that the newly introduced bounded face is convex, by construction.

To obtain a $1$-defective $2$-queue layout of $G$, we proceed by inserting all pairs of crossing edges of $G\setminus G_p$ into $\Gamma(G_p)$ to obtain first a 1-planar drawing $\Gamma(G)$ of $G$ as follows. Let $e$ and $e'$ be such a pair and let $f$ be the bounded face of $G_p$ containing $e$ and $e'$. By the construction of face $f$ in $\Gamma(G_p)$, either both edges $e$ and $e'$ are span-$1$ edges, or one, say $e$, is a span-$0$ edge and the other one, that is, $e'$, is a span-$1$ edge. In the former case, since $f$ is convex, both $e$ and $e'$ can be drawn as straight-line segments completely in the interior of $f$. In the latter case, edge $e$ has to be drawn as a curve in the interior of $f$ to avoid edge overlaps, while edge $e'$ is drawn as a straight-line segment in the interior of $f$ (refer to the gray-colored edges of \cref{fig:outer-1-planar-queue-2,fig:outer-1-planar-queue-4}).  We transform drawing $\Gamma(G)$ into a $1$-defective $2$-queue layout of $G$ as follows
:
\begin{enumerate}[label=(\emph{\roman*})]
\item For any two vertices $u$ and $v$ of $G$, $u \prec v$ if and only if either $y(u) > y(v)$, or $y(u)=y(v)$ and $x(u)<x(v)$ in $\Gamma(G)$, and 
\item the edges assigned to the first (second) defective queue in the layout are span-$0$ edges (span-$1$ edges, respectively) in~$\Gamma(G)$.
\end{enumerate}
The edges assigned to the first (defective) queue form a $0$-defective layout, as no edge nests another edge. By \cref{le:charact-arched-level-1-planar}, the edges assigned to the second defective queue form a $1$-defective queue layout, completing the proof.
\end{proof}

\begin{figure}[tb]
    \centering
    \begin{subfigure}{.48\textwidth}
    \centering
    {\includegraphics[page=1,width=.8\textwidth]{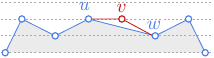}
    \subcaption{\label{fig:outer-1-planar-queue-1}}}    
    \end{subfigure}
    \begin{subfigure}{.48\textwidth}
    \centering
    {\includegraphics[page=3,width=.8\textwidth]{figures/outer-1-planar-queue.pdf}
    \subcaption{\label{fig:outer-1-planar-queue-3}}}    
    \end{subfigure}
    \begin{subfigure}{.48\textwidth}
    \centering
    {\includegraphics[page=2,width=.8\textwidth]{figures/outer-1-planar-queue.pdf}
    \subcaption{\label{fig:outer-1-planar-queue-2}}}    
    \end{subfigure}
    \begin{subfigure}{.48\textwidth}
    \centering
    {\includegraphics[page=4,width=.8\textwidth]{figures/outer-1-planar-queue.pdf}
    \subcaption{\label{fig:outer-1-planar-queue-4}}}
    \end{subfigure}
    \caption{Illustration for the proof of \cref{th:outer-1-planar}.}
    \label{fig:outer-1-planar-queue}
\end{figure}

As the layouts constructed in \cref{th:outer-1-planar} consist of two defective queues, one of which is $0$-defective, \cref{th:outer-1-planar} and  \cref{obs:bounded-defective} imply that the queue number of outer 1-planar graphs is at most $3$. This is an alternative proof of the same upper bound given by~Auer~et~al.~\cite{DBLP:journals/algorithmica/AuerBBGHNR16} based on BFS traversal of the graph. 

For general planar graphs, we prove in the following an upper bound of $33$ on their $1$-defective queue number by adjusting a well-known technique~\cite{DBLP:journals/jacm/DujmovicJMMUW20}. 
To do so, however, we need to introduce a few additional definitions and an auxiliary lemma.
Given a graph $G$, an \emph{$H$-partition} of $G$ is a partition of the vertices of $G$ into sets $A_x$ with $x\in V(H)$, called \emph{bags}, such that for each edge $(u,v)$ of $G$ with $u\in A_x$ and $v \in A_y$ either $x=y$ holds or $(x,y)$ is an edge of $H$. In the former case, $(u,v)$ is called \emph{intra-bag} edge, while in the latter case it is called \emph{inter-bag} edge. A \emph{BFS-layering} of $G$ is a partition $\mathcal{L}=(V_0, V_1,\ldots)$ of its vertices according to their distance from a specific vertex of~$G$. In other words, it is a special type of an $H$-partition, where $H$ is a path and each bag $V_i$ corresponds to a \emph{layer}. In this special type of $H$-partition, an intra-bag edge is called \emph{intra-layer}, while an inter-bag edge is called \emph{inter-layer}. An $H$-partition has \emph{layered-width} $\ell$ with respect to a BFS-layering $\mathcal{B}$ if each bag of $H$ has at most $\ell$ vertices on each layer of~$\mathcal{B}$.

\begin{lemma}\label{le:BFS-layering}
 Let $G$ be a graph admitting an $H$-partition of layered-width $\ell$ with respect to a BFS-layering $\mathcal{B}$. Then, the $1$-defective queue number of $G$ is upper-bounded by:
 \[
 (3 \cdot \text{qn}(H)+1) \cdot \left\lceil \frac{2\ell}{3} \right\rceil + \left\lceil \frac{\ell-1}{3} \right\rceil.
 \]
 \end{lemma}

\begin{proof}
The order of the vertices of $G$ is defined as in the original proof in~\cite[Lemma~9]{DBLP:journals/jacm/DujmovicJMMUW20}. More specifically, let $x_1, \dots, x_h$ be the vertices of $H$ as they appear in a $qn(H)$-queue layout of $H$, that is, $x_1\prec_{H}\; \dots \prec_{H}\; x_h$, and let  $A_{x_1}, \dots, A_{x_h}$ be the corresponding bags of the $H$-partition. Then, the order $\overrightarrow{V_i}$ of the vertices of each layer $V_i$ of the BFS-layering $\mathcal{B}$ with $i\geq 0$ is:
\[
\overrightarrow{V_i}=A_{x_1}\cap V_i, A_{x_2}\cap V_i, \dots, A_{x_h}\cap V_i
\]
where each subset $A_{x_j}\cap V_i$ is ordered arbitrarily, for $j=1,\ldots, h$. To obtain a total order $\prec_{G}$ for the vertices of $G$ we concatenate, for each $i=0,1,\ldots,n$, the order $\overrightarrow{V_i}$ of the vertices of each layer $V_i$. 
For the edge-to-queue assignment, we exploit two properties that are given in the proof of~\cite[Lemma~9]{DBLP:journals/jacm/DujmovicJMMUW20}. First, two intra-bag edges of $G$ nest in $\prec_{G}$, only if they belong to the same bag; see the blue and red edges in \cref{fig:intra_bag_edges}. Second, two inter-layer edges of $G$ nest in $\prec_{G}$, only if their endpoints belong to the~same pair of consecutive layers of $\mathcal{B}$; see the blue, purple and orange edges in \cref{fig:inter_bag_edges}. With these two properties in mind, we are now ready to describe the edge-to-queue assignment: 

\begin{figure}
    \centering
    \begin{subfigure}{.45\textwidth}
    \centering
    {\includegraphics[page=1]{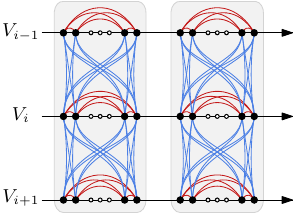}
    \subcaption{\label{fig:intra_bag_edges} Intra-bag edges}} 
    \end{subfigure}
    \begin{subfigure}{.45\textwidth}
    \centering
    {\includegraphics[page=2]{figures/planar_queues}
    \subcaption{\label{fig:inter_bag_edges} Inter-bag edges}} 
    \end{subfigure}
    \caption{Illustration of
    (a)~Intra-bag edges; the intra-layer ones are red, 
    while the inter-layer ones are blue, and 
    (b)~inter-bag edges; the intra-layer ones are green, 
    while the inter-layer ones are purple (forward) and orange (backward).
    In both subfigures, a bag is illustrated as a gray box, in which vertices of different layers (that are part of the bag) are drawn at different horizontal lines; the white vertices indicate that the intersection of a layer with the content of a bag may consist of more than four vertices; the edges towards these vertices have not been illustrated.}
    \label{fig:categories-edges}
\end{figure}

\begin{itemize}
\item The intra-layer intra-bag edges of $G$ are assigned to at most $\left\lceil \frac{\ell-1}{3} \right\rceil$ $1$-defective queues, since, by \cref{cor:kn-1defect}, the $1$-defective queue~number of $K_\ell$ is $\left\lceil \frac{\ell-1}{3} \right\rceil$; refer to the red edges of \cref{fig:intra_bag_edges}.
\item The inter-layer intra-bag edges of $G$ are assigned to at most $\left\lceil \frac{2\ell}{3}\right\rceil$ queues, since, by \cref{lem:knn-sep-1defect}, the $1$-defective queue~number of $K_{\ell, \ell}$  is at most $\left\lceil \frac{2\ell}{3}\right\rceil$, when all vertices of the first bipartition precede those of the second; refer to the blue edges of \cref{fig:intra_bag_edges}.

\item The remaining edges of $G$ are inter-bag edges connecting vertices of different~bags; see \cref{fig:inter_bag_edges}. Their edge-to-queue assignment is determined by further partitioning the inter-layer inter-bag edges into two categories as follows. Let $(u,v)$ be an inter-layer inter-bag edge with $u \in A_x \cap V_i$ and $v \in A_y \cap V_{i+1}$, for some $i\geq0$ and $x\neq y$. Then $(u,v)$ is \emph{forward}, if $x\; \prec_{H}\; y$ holds in $\prec_{H}$, otherwise, $(u,v)$ is \emph{backward}; refer to the purple and the orange edges of \cref{fig:inter_bag_edges}, respectively. In the following, we demonstrate that all inter-bag edges can be assigned to at most $3 \cdot\left\lceil \frac{2\ell}{3}\right\rceil\cdot qn(H)$ $1$-defective queues.
\begin{itemize}
\item By \cref{lem:knn-sep-1defect}, intra-layer inter-bag edges of $G$ can be assigned to at most $qn(H)\cdot\left\lceil \frac{2\ell}{3}\right\rceil$ $1$-defective queues; on each layer, an edge of $H$ corresponds to a subgraph of $K_{\ell, \ell}$, where all vertices of the first bipartition precede those of the second; refer to the green edges of \cref{fig:inter_bag_edges}.
\item By \cref{lem:knn-sep-1defect}, forward inter-layer inter-bag edges of $G$ can be assigned to at most $qn(H)\cdot\left\lceil \frac{2\ell}{3}\right\rceil$ $1$-defective queues; for two consecutive layers, an edge of $H$  corresponds to a subgraph of $K_{\ell, \ell}$, where all vertices of the first bipartition precede those of the second; refer to the purple edges of \cref{fig:inter_bag_edges}.
\item Symmetrically, by \cref{lem:knn-sep-1defect},  all backward inter-layer inter-bag edges of $G$ can be assigned to at most $qn(H)\cdot\left\lceil \frac{2\ell}{3}\right\rceil$ $1$-defective queues; refer to the orange edges of \cref{fig:inter_bag_edges}.
\end{itemize}
\end{itemize}
Since no two edges of $H$ nest in $\prec_{H}$, it follows that no two edges of $G$ can form a $k$-defect in $\prec_G$ with $k>1$, which completes the proof of the lemma. 
\end{proof}

\begin{theorem}\label{th:planar-1-def-queue}
The $1$-defective queue number of planar graphs is at most $33$.
\end{theorem}
\begin{proof}
Since every planar graph is a subgraph of a graph isomorphic to the strong product of a path,  
a clique on three vertices~\cite{DBLP:journals/jacm/DujmovicJMMUW20}, and a planar graph of tree-width at most $3$, it follows that it admits an $H$-partition of layered-width $3$, where $H$ is a planar graph of tree-width at most $3$ (i.e., a planar $3$-tree). 
By \cref{le:BFS-layering}, the $1$-defective queue number of planar graphs cannot be more than $33$, as planar $3$-trees have queue number at most $5$~\cite{DBLP:journals/algorithmica/AlamBGKP20}.   
\end{proof}

We now give a lower bound (\cref{th:kn-Kdefect-lower}) and an upper bound (\cref{th:kn-kdefect}) of the $k$-defective queue number of the complete graph $K_n$. Note that for $k\geq 2$ there is a gap between the bounds, while for $k=1$ the two bounds are the same, which directly implies~\cref{cor:kn-1defect}.

\begin{theorem}\label{th:kn-Kdefect-lower}
  The $k$-defective queue number of $K_n$ is at least $\left\lceil \frac{n-1}{k+2}\right\rceil$.
\end{theorem}

\begin{proof}
Consider an (arbitrary) $k$-defective queue layout $\mathcal{L}$ of $K_n$.
As observed in \Cref{sse:queue-density}, there exists a partition of the edges of $\mathcal{L}$ into classes $C_1, C_2,\dots, C_{n-1}$. Since $K_n$ is complete, the class $C_{n-1}$ has $n-1$ edges and, by \Cref{le:alternating}, at most $k+2$ edges of this class can be in the same defective queue. Hence, all $n-1$ edges of $C_{n-1}$ require at least $\left\lceil \frac{n-1}{k+2}\right\rceil$ $k$-defective queues proving the lower bound.
\looseness=-1
\end{proof}

\begin{theorem}\label{th:kn-kdefect}
The $k$-defective queue number of $K_n$ is at most $\left\lceil \frac{n-1}{l} \right\rceil$, where $l {=} \left\lfloor \frac{3+\sqrt{8k+1}}{2} \right\rfloor$.
\end{theorem}

\begin{proof}
We prove the statement by giving an explicit construction.
For each integer $a \geq 0$, we place on defective queue $a$ all edges whose hop-size belongs to $\{al+1, \ldots, al+l\}$, where $l = \lfloor \frac{3+\sqrt{8k+1}}{2} \rfloor$.

We first show that this construction yields a $k$-defective queue layout.
Consider a defective queue $a$ of our construction.
It contains all edges with hop-size $\delta$ for $al+1 \leq \delta \leq al+l$.
Among the edges in defective queue $a$, the largest number of nested edges is attained by those with hop-size $al + l$.
Let $(v_i,v_{i+al+l})$ be such an edge in the given defective queue.
This edge nests precisely those edges~$(v_{i'},v_{j'})$ in the same defective queue satisfying $i<i'<j'<i+al+l$.
These are exactly the $\delta$-hop edges with $al+1 \leq \delta \leq al+l-2$.
Indeed, there is one nested $(al+l-2)$-hop edge (namely, $(v_{i+1},v_{i+al+l-1})$), two nested $(al+l-3)$-hop edges (namely, $(v_{i+1},v_{i+al+l-2})$ and $(v_{i+2},v_{i+al+l-1})$), and so on, and finally $al+l-1-al-1=l-1$ edges with hop-size $al+1$.
More generally, for each $\delta \in \{al+1, \ldots, al+l-2\}$, there are $al+l-1-\delta$ nested edges with hop-size $\delta$.
Thus, the edge $(v_i, v_{i+al+l})$ has $\sum_{\delta=al+1}^{al+l-2}(al+l-1-\delta)=\frac{1}{2}(l-2)(l-1)$ nested edges.
From $l = \left\lfloor \frac{3+\sqrt{8k+1}}{2} \right\rfloor$, it follows that $\frac{(l-2)(l-1)}{2}\leq k$, and thus the construction gives a $k$-defective queue layout.

It remains to bound the number of defective queues.
The possible hop-sizes in $K_n$ are exactly $1, 2, \ldots, n-1$.
Each hop-size belongs to exactly one interval $\{al+1, \ldots, al+l\}$.
So every edge of $K_n$ is assigned to exactly one defective queue.
Hence, the number of defective queues required is the number of intervals of the form $\{al+1, \ldots, al+l\}$ needed to cover the hop-sizes $1, 2, \ldots, n-1$, namely $\lceil \frac{n-1}{l} \rceil$.
This completes the proof.
\end{proof}

\begin{corollary}\label{cor:kn-1defect}
The $1$-defective queue number of $K_n$ is $\left\lceil \frac{n-1}{3} \right\rceil$.
\end{corollary}

\noindent In the following we present bounds on the $k$-defective queue number of the complete bipartite graph $K_{n,n}$, i.e., we assume that both parts of it are of the same cardinality. In particular, we study the general case as well as the natural case in which the vertices of one part of $K_{n,n}$ precede those of its second part. We refer to the latter as the \emph{separated setting}~\cite{DBLP:journals/tcs/AlamBGKP22}.

\begin{corollary}\label{lem:knn-sep-1defect}
 The $1$-defective queue number of $K_{n,n}$ in the separated setting is $\left\lceil \frac{2n-1}{3}\right\rceil$. 
 For $k>1$, the $k$-defective queue number of $K_{n,n}$ in the separated setting is at least $\left\lceil \frac{2n-1}{k+2}\right\rceil$ and at most $\left\lceil \frac{2n-1}{l} \right\rceil$, while in the non-separated setting it is at least $\left\lceil \frac{n-1}{k+2}\right\rceil$ and at most $\left\lceil \frac{2n-1}{l} \right\rceil$, where $l=\left\lfloor \frac{3+\sqrt{8k+1}}{2} \right\rfloor$. 
\end{corollary}

\begin{proof}
The upper bounds for all settings follow immediately from those of the $k$-defective queue number of $K_{2n}$ (\cref{th:kn-Kdefect-lower}). 
For the lower bounds of the $1$- and $k$-defective queue numbers of  $K_{n,n}$ in the separated setting, we partition the edges as in the proof of Theorem~\ref{th:kn-Kdefect-lower}. The lower bounds of $K_{n,n}$ in the general setting follow directly, as every defective queue layout of $K_{n,n}$ contains a corresponding one of $K_{\frac{n}{2},\frac{n}{2}}$ in the separated setting.   
\end{proof}

\subsection{Defective Stack Number}\label{sse:stack-defective-page-number}

The definition of $k$-defective $h$-stack layout of a graph $G$ implies that the subgraph of $G$ induced by the edges of each defective stack is an outer $k$-planar graph, and vice versa. This immediately yields the following characterization of the graphs admitting $k$-defective $1$-stack layouts, which was known for traditional stack layout, i.e., for $k=0$~\cite{DBLP:journals/jct/BernhartK79}. 

\begin{theorem}\label{th:charact-k-defective-stack}
A graph has $k$-defective stack number $1$ if and only if it is outer $k$-planar.
\end{theorem}

The bounds on the defective stack number of outer $k$-planar graphs reported in \cref{tb:queue-stack-number-1,tb:queue-stack-number-k} immediately derive from \cref{th:charact-k-defective-stack}.
We now study the $k$-defective stack number of $K_n$ and $K_{n,n}$. For $k=1$, our upper bound for $K_n$ is tight (up to an additive constant) as we demonstrate graphs with the same amount of edges minus $O(1)$. 

\begin{theorem}\label{th:1-defective-sn-kn}
The $1$-defective stack-number of $K_n$ is at least $\left\lfloor \frac{n}{3}\right\rfloor-1$ and at most $\left\lceil \frac{n}{3} \right\rceil$.
\end{theorem}

\begin{proof}
 We assume that $n$ is a multiple of $3$ and we prove that the complete graph on $n$ vertices $K_n$ admits a $1$-defective stack layout $\mathcal{L}_n$ with $\frac{n}{3}$ defective stacks. This proves the upper bound. We denote the defective stacks of $\mathcal{L}_n$ by $\mathcal{A}_0,\ldots,\mathcal{A}_{\kappa-1}$ and $\mathcal{B}_0,\ldots,\mathcal{B}_{\lambda-1}$, such that $\kappa+\lambda=\frac{n}{3}$; see \cref{app:fig:1-defective-sn-kn} for an illustration. Let $v_0,\ldots,v_{n-1}$ be the order of the vertices of $K_n$ in $\mathcal{L}_n$, and assume in the following that indices are taken $\mod n$. We first assign the edges $(v_i,v_{i+1})$, with $0 \leq i \leq n-1$, to any defective stack of $\mathcal{L}_n$. Thus, in the following assignment scheme, we omit these $n$ edges. 

We consider two cases in our proof, depending on whether $n$ is even or odd. First, assume that $n$ is even. In this case, $\kappa=\lambda=\frac{n}{6}$. We assign edges to the defective stacks $\mathcal{A}_0,\ldots,\mathcal{A}_{\frac{n}{6}-1}$, such that for each $j$ in $[0,\frac{n}{6}-1]$ the defective stack $\mathcal{A}_j$ contains the following $\frac{3n}{2}-4$ edges:

\noindent
\begin{tabular}{l l l l l }
\tikzcircle[defred, fill=defred]{2.8pt} $(v_{3j+i-1}, v_{n+3j-i-1})$ & $1 \leq i \leq \frac{n}{2}-1$ & & \tikzcircle[defblue, fill=defblue]{2.8pt} $(v_{3j+i-1}, v_{n+3j-i-2})$ & $1 \leq i \leq \frac{n}{2}-2$\\
\tikzcircle[defgreen, fill=defgreen]{2.8pt} $(v_{3j+i-1}, v_{n+3j-i-3})$ & $0 \leq i \leq \frac{n}{2}-2$ & & &
\end{tabular}
\medskip

Next, we assign edges to the defective stacks $\mathcal{B}_0,\ldots,\mathcal{B}_{\frac{n}{6}-1}$, such that for each $j$ in $[0,\frac{n}{6}-1]$ the defective stack $\mathcal{B}_j$ contains the following $\frac{3n}{2}-5$ edges:

\noindent
\begin{tabular}{ l l l l l }
\tikzcircle[defgray, fill=defgray]{2.8pt} $(v_{3j+i+1}, v_{n+3j-i})$ & $1 \leq i \leq \frac{n}{2}-2$ & & \tikzcircle[deforange, fill=deforange]{2.8pt} $(v_{3j+i+1}, v_{n+3j-i-1})$ & $0 \leq i \leq \frac{n}{2}-2$\\
\tikzcircle[defpurple, fill=defpurple]{2.8pt} $(v_{3j+i+1}, v_{n+3j-i-2})$ & $0 \leq i \leq \frac{n}{2}-3$ & & &
\end{tabular}
\bigskip 

In the case where $n$ is odd, our assignment is slightly different. Namely, we set $\kappa=\frac{n+3}{6}$ and $\lambda=\frac{n+3}{6}-1$ and assign edges to the defective stacks $\mathcal{A}_0,\ldots,\mathcal{A}_{\frac{n+3}{6}-1}$, such that for each $j$ in $[0,\frac{n+3}{6}-1]$ the defective stack $\mathcal{A}_j$ contains the following $\frac{3(n-1)}{2} - 3$ edges:

\noindent
\begin{tabular}{ l l l l l }
\tikzcircle[defred, fill=defred]{2.8pt} $(v_{3j+i}, v_{n+3j-i-1})$ & $1 \leq i \leq \frac{n-1}{2}-1$ & & \tikzcircle[defblue, fill=defblue]{2.8pt} $(v_{3j+i}, v_{n+3j-i-2})$ & $0 \leq i \leq \frac{n-1}{2}-2$\\
\tikzcircle[defgreen, fill=defgreen]{2.8pt} $(v_{3j+i}, v_{n+3j-i-3})$ & $0 \leq i \leq \frac{n-1}{2}-2$ & & &
\end{tabular}
\medskip

Finally, we assign edges to the defective stacks $\mathcal{B}_0,\ldots,\mathcal{B}_{\frac{n+3}{6}-2}$, such that for each $i$ in $[0,\frac{n+3}{6}-2]$ the defective stack $\mathcal{B}_i$ contains the following $\frac{3(n-1)}{2} - 3$ edges:

\noindent
\begin{tabular}{ l l l l l }
\tikzcircle[defgray, fill=defgray]{2.8pt}  $(v_{3j+i+1}, v_{n+3j-i+1})$ & $1 \leq i \leq \frac{n-1}{2}-1$ & & \tikzcircle[deforange, fill=deforange]{2.8pt} $(v_{3j+i+1}, v_{n+3j-i})$ & $1 \leq i \leq \frac{n-1}{2}-1$\\
\tikzcircle[defpurple, fill=defpurple]{2.8pt} $(v_{3j+i+1}, v_{n+3j-i-1})$ & $0 \leq i \leq \frac{n-1}{2}-2$ & & &
\end{tabular}
\medskip

\begin{figure}
    \begin{subfigure}{.18\textwidth}
        \centering
        {\includegraphics[page=6]{1-defective-stack-kn.pdf}}
        \subcaption{\label{app:fig:k12-a0}}
    \end{subfigure}
    \begin{subfigure}{.18\textwidth}
        \centering
        {\includegraphics[page=7]{1-defective-stack-kn.pdf}}
        \subcaption{\label{app:fig:k12-a1}}
    \end{subfigure}
    \begin{subfigure}{.18\textwidth}
        \centering
        {\includegraphics[page=8]{1-defective-stack-kn.pdf}}
        \subcaption{\label{app:fig:k12-b0}}
    \end{subfigure}
    \begin{subfigure}{.18\textwidth}
        \centering
        {\includegraphics[page=9]{1-defective-stack-kn.pdf}}
        \subcaption{\label{app:fig:k12-b1}}
    \end{subfigure}
    
    \begin{subfigure}{.18\textwidth}
        \centering
        {\includegraphics[page=1]{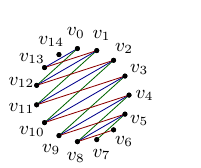}}
        \subcaption{\label{app:fig:k15-a0}}
    \end{subfigure}
    \begin{subfigure}{.18\textwidth}
        \centering
        {\includegraphics[page=2]{1-defective-stack-kn.pdf}}
        \subcaption{\label{app:fig:k15-a1}}
    \end{subfigure}
    \begin{subfigure}{.18\textwidth}
        \centering
        {\includegraphics[page=3]{1-defective-stack-kn.pdf}}
        \subcaption{\label{app:fig:k15-a2}}
    \end{subfigure}
    \begin{subfigure}{.18\textwidth}
        \centering
        {\includegraphics[page=4]{1-defective-stack-kn.pdf}}
        \subcaption{\label{app:fig:k15-b0}}
    \end{subfigure}
    \begin{subfigure}{.18\textwidth}
        \centering
        {\includegraphics[page=5]{1-defective-stack-kn.pdf}}
        \subcaption{\label{app:fig:k15-b1}}
    \end{subfigure}
   
    \caption{Illustration of $1$-defective stack layouts of $K_{12}$ with 4 defective stacks 
    (a)~$\mathcal{A}_0$,
    (b)~$\mathcal{A}_1$, 
    (c)~$\mathcal{B}_0$,  
    (d)~$\mathcal{B}_1$, and of $K_{15}$ with 5 defective stacks 
    (e)~$\mathcal{A}_0$,
    (f)~$\mathcal{A}_1$, 
    (g)~$\mathcal{A}_2$,
    (h)~$\mathcal{B}_0$, and 
    (i)~$\mathcal{B}_1$.}
    \label{app:fig:1-defective-sn-kn}
\end{figure}

By construction, each edge in $\mathcal{L}_n$ is either crossing-free (refer to the edges colored blue and orange above) or is crossed by at most one other edge (refer to the edges colored red and green, and also to the edges colored gray and purple above). Further, it is not difficult to see that no edge is assigned to two distinct defective stacks. To complete the proof, we count the total number of edges in $\mathcal{L}_n$. Summing up~over $j$, we conclude that $\mathcal{L}_n$ contains $(\frac{3n}{2} - 4)\cdot \frac{n}{6} + (\frac{3n}{2} - 5)\cdot \frac{n}{6}$ edges, if $n$ is even; otherwise, $\mathcal{L}_n$ contains $(\frac{3(n-1)}{2} - 3)\cdot \frac{n+3}{6} + (\frac{3(n-1)}{2} - 3)\cdot (\frac{n+3}{6}-1)$ edges. Thus, in total $\mathcal{L}_n$ has $\frac{1}{2}(n^2-3n)$ edges, i.e., the number of edges of $K_n$ neglecting the $n$ edges that we assigned at the beginning to some (defective) stack of $\mathcal{L}_n$.  This completes the proof of the upper bound.

For the lower bound, let $h$ be the number of defective stacks in a $1$-defective stack layout 
of $K_n$. Since $K_n$ has $\frac{n(n-1)}{2}$ edges, by \cref{th:1-defective-stack-density} we get:
$\frac{n(n-1)}{2} \leq (\frac{3h}{2}+1)n - 4h$,
which implies 
$h \geq \frac{n(n-3)}{3n - 8}$.
Since the right-hand side of the last inequality is lower bounded by $\frac{n-3}{3}$, the total number of defective stacks required in a $1$-defective stack layout of $K_n$ is~$\left\lfloor \frac{n}{3}\right\rfloor-1$.
\end{proof}

\noindent \cref{th:k-defective-stack-kn} forms a generalization of \cref{th:1-defective-sn-kn} providing upper and lower bounds on the $k$-defective stack number of $K_n$. Notably, the obtained upper bound is tight for $k\in \{0,1\}$.

\begin{theorem}\label{th:k-defective-stack-kn}
The $k$-defective stack number of $K_n$ is at least $\left\lceil \frac{n}{2k+2}\right\rceil$ and at most $\left\lceil \frac{n}{l+2} \right\rceil$, where $l {=} \left\lfloor \frac{-1 + \sqrt{8k+1}}{2} \right\rfloor$.
\end{theorem}

\begin{proof}
The proof of the lower bound is an adjustment of a corresponding one given in~\cite{DBLP:journals/jct/BernhartK79}. More precisely, in any linear order of the vertices of $K_n$, there exists an $\frac{n}{2}$-twist, that is, $\frac{n}{2}$ edges that pairwise cross. It follows that at most $k+1$ of them can coexist in a single defective stack of a $k$-defective stack layout of $K_n$, yielding the claimed lower bound. 

The proof of the upper bound is more tedious and it forms an extension of the one of \cref{th:1-defective-sn-kn}.
We assume that $n$ is a multiple of $l+2$, and we prove that the complete graph on $n$ vertices $K_n$ admits a $k$-defective stack layout $\mathcal{L}_n$ with $\frac{n}{l+2}$ defective stacks. This proves the upper bound. We denote the defective stacks of $\mathcal{L}_n$ by $\mathcal{A}_0,\mathcal{A}_1,\ldots,\mathcal{A}_{\kappa-1}$ and $\mathcal{B}_0,\mathcal{B}_1,\ldots,\mathcal{B}_{\lambda-1}$, such that $\kappa+\lambda=\frac{n}{l+2}$. Let $v_0,v_1,\ldots,v_{n-1}$ be the order of the vertices of $K_n$ in $\mathcal{L}_n$ and assume in the following that indices are taken$\mod n$. We first assign the edges $(v_i,v_{i+1})$, with $0 \leq i \leq n-1$, to any defective stack of $\mathcal{L}_n$. Thus, in the following assignment scheme, we omit these $n$ edges. We proceed by considering cases in our proof, depending on whether each of $n$ and $l$ is even or odd. 

\paragraph*{Case 1: $l$ is even.}

Since $n$ is a multiple of $l+2$, which is even, it follows that $n$ is even. In this case, we set $\kappa = \frac{n}{l+2}$ and $\lambda = 0$, and we proceed to assign edges to the defective stacks $\mathcal{A}_0,\mathcal{A}_1,\ldots,\mathcal{A}_{\frac{n}{l+2}-1}$, such that for each $j \in [0,\frac{n}{l+2}-1]$ the defective stack $\mathcal{A}_j$ contains the following $\frac{l+2}{2} (n-3)$ edges (see \cref{fig:k-defective-l-even}): 

\medskip 
\noindent
\begin{tabular}{l@{\;\;}l@{\hspace{1ex}}r @{$ \;\; \leq \; i \; \leq \;\; $} l} 
\tikzcircle[defpurple, fill=defpurple]{2.8pt}      & $(v_{\frac{l+2}{2}j+i-1}, v_{n+\frac{l+2}{2}j-i-1})$   & $1$  & $\frac{n}{2}-1$ \\ 
\tikzcircle[deforange, fill=deforange]{2.8pt}     & $(v_{\frac{l+2}{2}j+i-1}, v_{n+\frac{l+2}{2}j-i-2})$     & $1$  & $\frac{n}{2}-2$ \\
\tikzcircle[defgray, fill=defgray]{2.8pt}   & $(v_{\frac{l+2}{2}j+i-1}, v_{n+\frac{l+2}{2}j-i-2s-3})$ & $-s$  & $\frac{n}{2}-s-2$, $s \in [0,\frac{l}{2}-1]$ \\
\tikzcircle[defgreen, fill=defgreen]{2.8pt}  & $(v_{\frac{l+2}{2}j+i-1}, v_{n+\frac{l+2}{2}j-i-2s-4})$ & $-s$  & $\frac{n}{2}-s-3$, $s \in [0,\frac{l}{2}-1]$ 
\end{tabular}

\begin{figure}[h]
    \begin{subfigure}{.48\textwidth}
        \centering
        {\includegraphics[page=8,width=\textwidth]{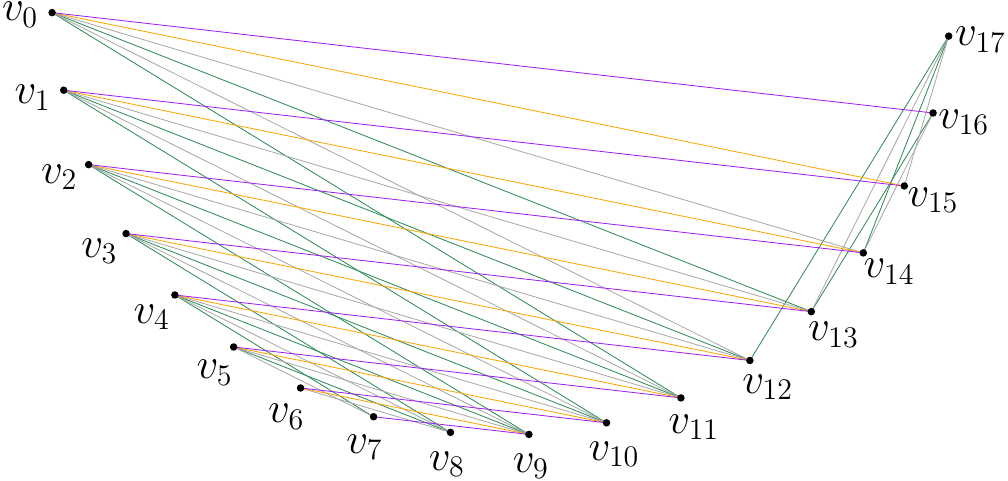}}
        \subcaption{\label{fig:k24-a0}$\mathcal{A}_0$}
    \end{subfigure}    
    \begin{subfigure}{.48\textwidth}
        \centering
        {\includegraphics[page=9,width=\textwidth]{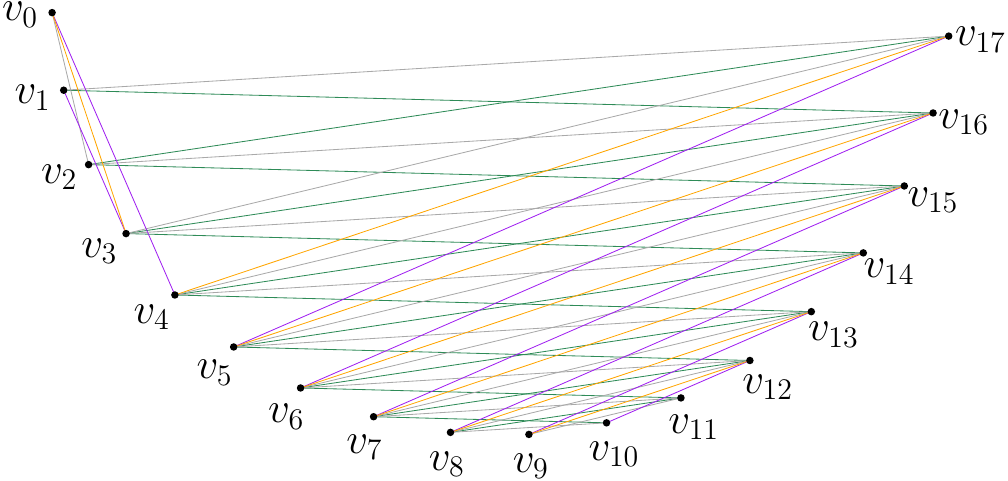}}
        \subcaption{\label{fig:k24-a1}$\mathcal{A}_1$}
    \end{subfigure}
    
    \begin{subfigure}{.48\textwidth}
        \centering
        {\includegraphics[page=10,width=\textwidth]{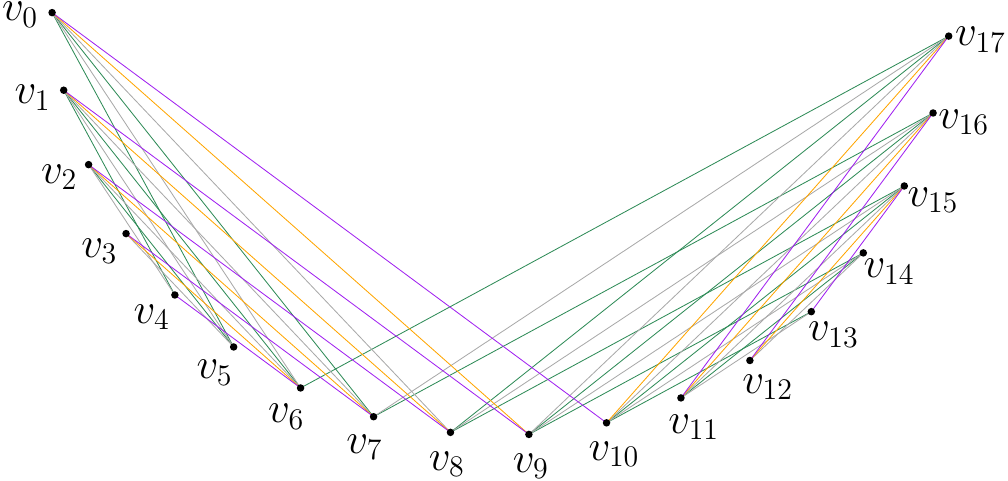}}
        \subcaption{\label{fig:k24-a2}$\mathcal{A}_2$}
    \end{subfigure}
    \caption{A $k$-defective stack layout of $K_{18}$ with $\left\lceil \frac{n}{l+2} \right\rceil$ defective stacks, where $k=10$ and $l=4$.}
    \label{fig:k-defective-l-even}
\end{figure}

\medskip\noindent It follows that the total number of edges that participate in the aforementioned scheme is:
\begin{equation*}
\begin{split}
\sum_{j=0}^{\frac{n}{l+2}-1} \left( \frac{n}{2}-1 +\frac{n}{2}-2 + \sum_{s=0}^{\frac{l}{2}-1} \left( \frac{n}{2} - s -2 -(-s)+1 \right) + \sum_{s=0}^{\frac{l}{2}-1} \left( \frac{n}{2} - s -3 -(-s)+1 \right) \right)= \\
\sum_{j=0}^{\frac{n}{l+2}-1} \left( n-3 + \sum_{s=0}^{\frac{l}{2}-1} \left( n-3 \right) \right) 
=\sum_{j=0}^{\frac{n}{l+2}-1} \left( \frac{l+2}{2}\cdot(n-3) \right) 
= \frac{n}{l+2} \cdot \frac{l+2}{2} (n-3) = \frac{n(n-3)}{2}.
\end{split}
\end{equation*}

Since the indices guarantee that no edge connecting consecutive vertices along the spine appears in the aforementioned scheme, the total number of edges in $\mathcal{L}_n$ equals the number of edges of $K_n$, neglecting the $n$ edges that we assigned at the beginning to some defective stack of $\mathcal{L}_n$, as desired. We next prove that $\mathcal{L}_n$ is $k$-defective and that no edge has been assigned to two distinct defective stacks.
 
On a single defective stack, no two purple edges cross; the same holds for the orange ones. Also, on each defective stack $\mathcal{A}_j$ with $j$ in $[0,\frac{n}{l+2}-1]$, each so-called \emph{root vertex} $v_{\frac{l+2}{2}j+i-1}$ is incident to at most $l+2$ edges (one purple, one orange, at most $\frac{l}{2}$ gray, and at most $\frac{l}{2}$ green edges) that connect it to at most $l+2$ consecutive vertices. Among those, the edges $(v_{\frac{l+2}{2}j+i-1}, v_{n+\frac{l+2}{2}j-i-1})$ and $(v_{\frac{l+2}{2}j+i-1}, v_{n+\frac{l+2}{2}j-i-(l+2))})$ are the ones that receive the most crossings; see \cref{fig:k-defective-l-even-fan}. We argue for the former; the latter is symmetric. The edge $(v_{\frac{l+2}{2}j+i-1}, v_{n+\frac{l+2}{2}j-i-1})$ is crossed by $l$ edges rooted at $v_{\frac{l+2}{2}j+i-2}$, $l-1$ edges rooted at $v_{\frac{l+2}{2}j+i-3}$, and so on. This implies that the edge $(v_{\frac{l+2}{2}j+i-1}, v_{n+\frac{l+2}{2}j-i-1})$ is crossed $l+(l-1)+\ldots+1 = \frac{l(l+1)}{2}$ times. Since for every real number $x$ it holds $\lfloor x \rfloor \leq x$, by setting $l = \left\lfloor \frac{-1 + \sqrt{8k+1}}{2} \right\rfloor$ to $\frac{l(l+1)}{2}$, we obtain that the defectiveness of the layout is $k$, as desired. Further, the edges rooted at  $v_{\frac{l+2}{2}j+i-1}$ in the defective stacks $\mathcal{A}_{j-1}$ and $\mathcal{A}_{j+1}$ immediately precede or follows those at the defective stack $\mathcal{A}_j$, which implies that no edge has been assigned to two defective stacks.

\begin{figure}[thb]
    \centering
    {\includegraphics[page=1,width=\textwidth]{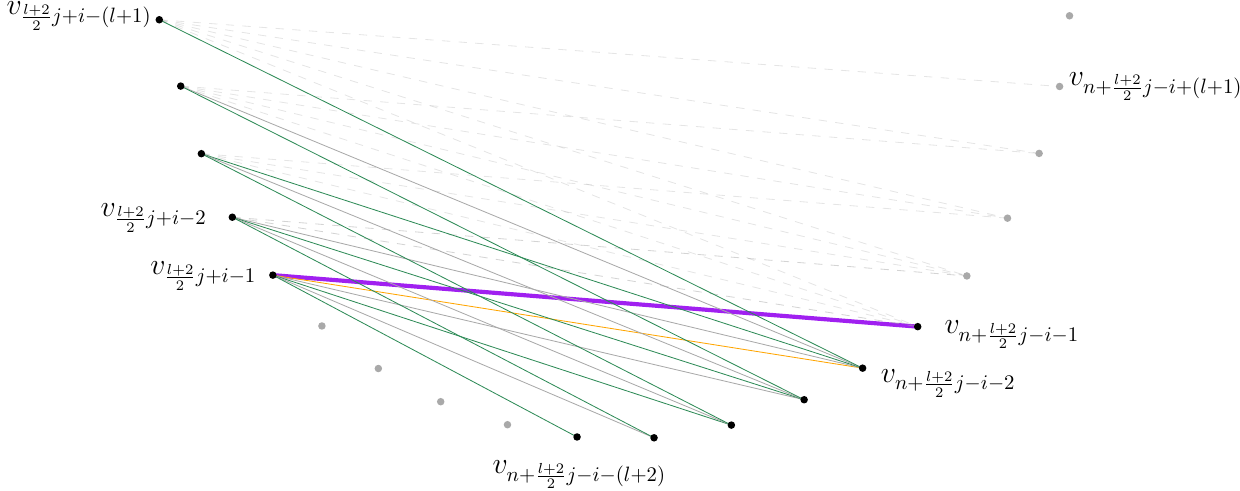}}
   \caption{Illustration for the crossings along the bold-drawn edge $\left(v_{\frac{l+2}{2}j+i-1}, v_{n+\frac{l+2}{2}j-i-1}\right)$.}
    \label{fig:k-defective-l-even-fan}
\end{figure}

\paragraph*{Case 2.a: $l$ is odd and $n$ is even.}

In this case, we set $\kappa=\lambda=\frac{n}{2(l+2)}$. We first assign edges to the defective stacks $\mathcal{A}_0,\mathcal{A}_1,\ldots,\mathcal{A}_{\frac{n}{2(l+2)}-1}$, such that for each $j$ in $[0,\frac{n}{2(l+2)}-1]$ the defective stack $\mathcal{A}_j$ contains the following edges:

\medskip 

\noindent
\begin{tabular}{l@{\;\;}l@{\hspace{1ex}}r @{$ \;\; \leq \; i \; \leq \;\; $} l} 
\tikzcircle[defgray, fill=defgray]{2.8pt}       & $(v_{(l+2)j+i-1}, v_{n+(l+2)j-i-1})$   & $1$  & $\frac{n}{2}-1$ \\ 
\tikzcircle[deforange, fill=deforange]{2.8pt}     & $(v_{(l+2)j+i-1}, v_{n+(l+2)j-i-2s-2})$  & $1-s$  & $\frac{n}{2}-s-2$, $s \in[0,\frac{l-1}{2}]$ \\
\tikzcircle[defpurple, fill=defpurple]{2.8pt}   & $(v_{(l+2)j+i-1}, v_{n+(l+2)j-i-2s-3})$   & $-s$  & $\frac{n}{2}-s-2$, $s \in[0,\frac{l-1}{2}]$\\  
\end{tabular}

\medskip\noindent Next, we assign edges to the defective stacks $\mathcal{B}_0,\mathcal{B}_1,\ldots,\mathcal{B}_{\frac{n}{2(l+2)}-1}$, such that for each $j$ in $[0,\frac{n}{2(l+2)}-1]$ the defective stack $\mathcal{B}_j$ contains the following edges:

\medskip 

\noindent
\begin{tabular}{l@{\;\;}l@{\hspace{1ex}}r @{$ \;\; \leq \; i \; \leq \;\; $} l} 

\tikzcircle[defgray, fill=defgray]{2.8pt} & $(v_{(l+2)j+i+1}, v_{n+(l+2)j-i+1})$ & $1$  & $\frac{n}{2}-2$ \\
\tikzcircle[deforange, fill=deforange]{2.8pt} & $(v_{(l+2)j+i+1}, v_{n+(l+2)j-i-2s})$ & $1-s$  & $\frac{n}{2}-s-2$, $s \in [0,\frac{l-1}{2}]$\\
\tikzcircle[defpurple, fill=defpurple]{2.8pt}   & $(v_{(l+2)j+i+1}, v_{n+(l+2)j-i-2s-1})$ & $-s$  & $\frac{n}{2}-s-2$, $s \in [0,\frac{l-1}{2}]$\\
\end{tabular}

\begin{figure}[h!]
    \begin{subfigure}{.48\textwidth}
        \centering
        {\includegraphics[page=4,width=\textwidth]{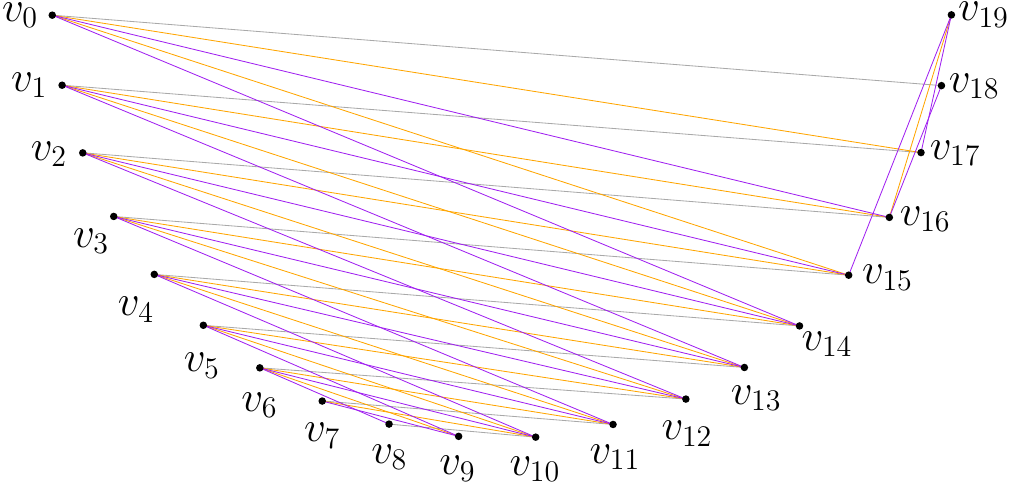}}
        \subcaption{\label{fig:k20-a0}$\mathcal{A}_0$}
    \end{subfigure}    
    \begin{subfigure}{.48\textwidth}
        \centering
        {\includegraphics[page=5,width=\textwidth]{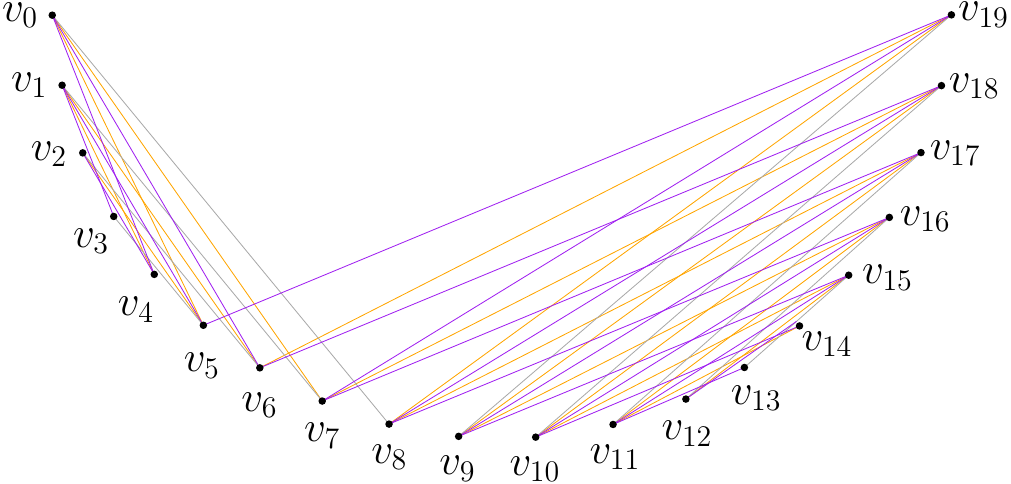}}
        \subcaption{\label{fig:k20-a1}$\mathcal{A}_1$}
    \end{subfigure}
    
    \begin{subfigure}{.48\textwidth}
        \centering
        {\includegraphics[page=6,width=\textwidth]{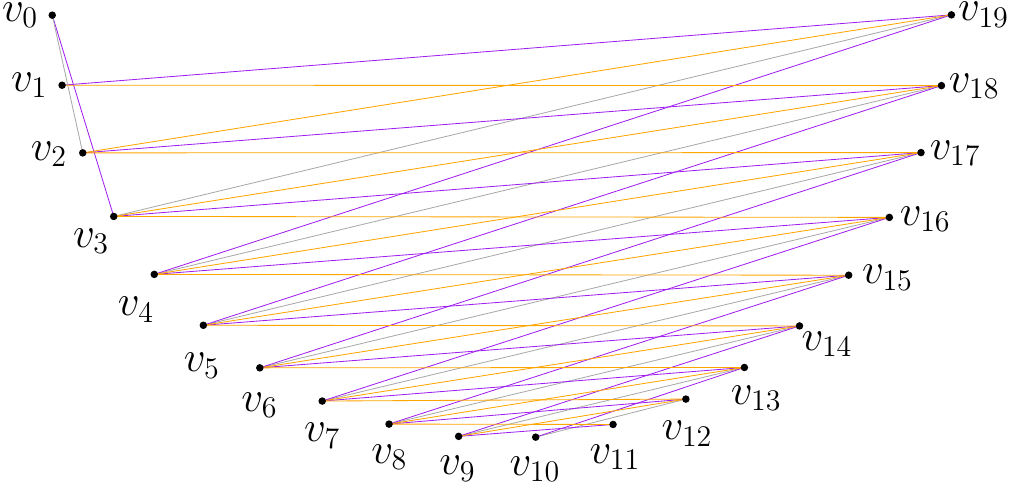}}
        \subcaption{\label{fig:k20-b0}$\mathcal{B}_0$}
    \end{subfigure}
    \begin{subfigure}{.48\textwidth}
        \centering
        {\includegraphics[page=7,width=\textwidth]{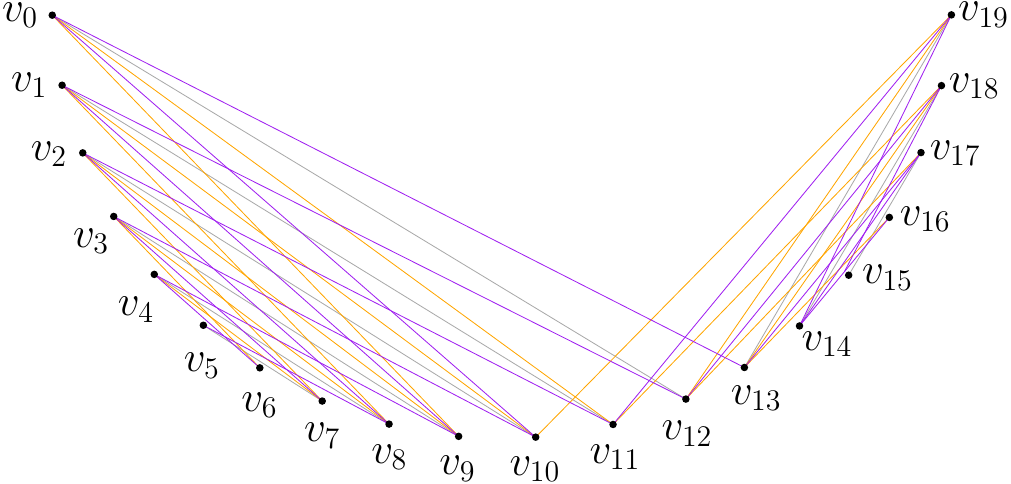}}
        \subcaption{\label{fig:k20-b1}$\mathcal{B}_1$}
    \end{subfigure}
    \caption{A $k$-defective stack layout of $K_{20}$ with $\left\lceil \frac{n}{l+2} \right\rceil$ defective stacks, where $k=6$ and $l=3$.}
    \label{fig:k-defective-l-odd-n-even}
\end{figure}

\medskip\noindent The total number of edges assigned to $\mathcal{A}_0,\mathcal{A}_1,\ldots,\mathcal{A}_{\frac{n}{2(l+2)}-1}$ equals:

\begin{equation*}
\begin{split}
\sum_{j=0}^{\frac{n}{2(l+2)}-1} \left( \frac{n}{2}-1 + \sum_{s=0}^{\frac{l-1}{2}} \left( \frac{n}{2}-s-2-(1-s)+1 \right) + \sum_{s=0}^{\frac{l-1}{2}} \left( \frac{n}{2} - s -2 -(-s)+1 \right) \right)= \\
\sum_{j=0}^{\frac{n}{2(l+2)}-1} \left( \frac{n}{2}-1 + \sum_{s=0}^{\frac{l-1}{2}} \left( n-3 \right) \right) = \frac{n}{2(l+2)} \cdot \left(\frac{n}{2}-1+ \frac{l+1}{2} \cdot (n-3) \right).
\end{split}
\end{equation*}

\medskip\noindent Accordingly, the total number of edges assigned to $\mathcal{B}_0,\mathcal{B}_1,\ldots,\mathcal{B}_{\frac{n}{2(l+2)}-1}$ equals:

\begin{equation*}
\begin{split}
\sum_{j=0}^{\frac{n}{2(l+2)}-1} \left( \frac{n}{2}-2 + \sum_{s=0}^{\frac{l-1}{2}} \left( \frac{n}{2}-s-2-(1-s)+1 \right) + \sum_{s=0}^{\frac{l-1}{2}} \left( \frac{n}{2}- s-2 -(-s)+1 \right) \right)= \\
\sum_{j=0}^{\frac{n}{2(l+2)}-1} \left( \frac{n}{2}-2 + \sum_{s=0}^{\frac{l-1}{2}} \left( n-3 \right) \right) = \frac{n}{2(l+2)} \cdot \left(\frac{n}{2}-2+ \frac{l+1}{2} \cdot (n-3) \right).
\end{split}
\end{equation*}

Summing the number of edges assigned to $\mathcal{A}_0,\mathcal{A}_1,\ldots,\mathcal{A}_{\frac{n}{2(l+2)}-1}$ and $\mathcal{B}_0,\mathcal{B}_1,\ldots,\mathcal{B}_{\frac{n}{2(l+2)}-1}$, we obtain that the total number of edges that participate in the scheme, when $n$ is even and $l$ is odd, is:
\begin{equation*}
\begin{split}
\frac{n}{2(l+2)} \cdot \left(\frac{n}{2}-1+\frac{n}{2}-2+ (l+1) \cdot (n-3) \right) = \frac{1}{2} (n^2 -3n).
\end{split}
\end{equation*}

Thus, the total number of edges in $\mathcal{L}_n$ equals the number of edges of $K_n$, neglecting the $n$ edges that we assigned at the beginning to some defective stack of $\mathcal{L}_n$, as desired. As in Case 1 in which $l$ was even, the indices guarantee that no edge connecting consecutive vertices along the spine appears in the aforementioned scheme. In the following, we prove that $\mathcal{L}_n$ is $k$-defective and that no edge has been assigned to two distinct defective stacks. 

\begin{figure}[thb]
    \centering
    \begin{subfigure}{.78\textwidth}
    \centering
        {\includegraphics[page=2,width=\textwidth]{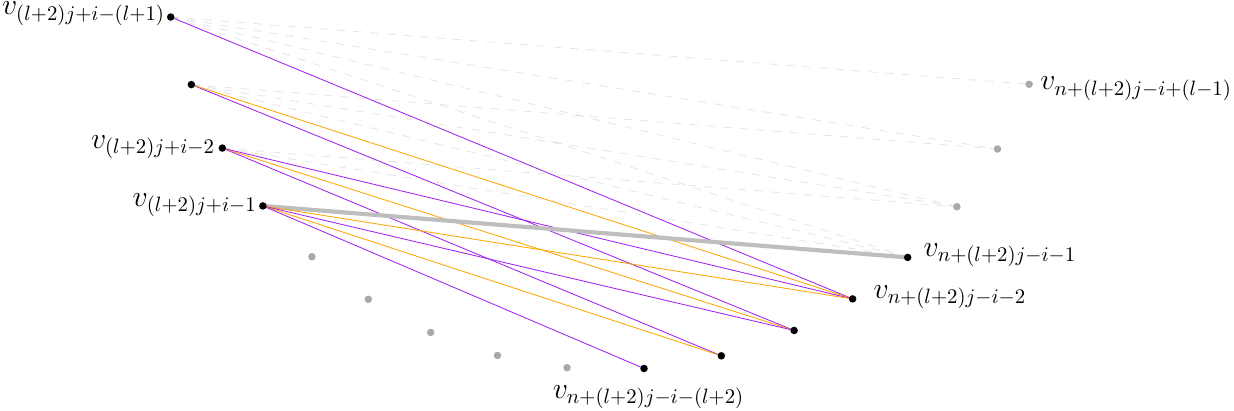}}
        \subcaption{\label{fig:k-defective-l-odd-n-even-A-fan}$\mathcal{A}_j$}
    \end{subfigure}
    \begin{subfigure}{.78\textwidth}
        \centering
        {\includegraphics[page=3,width=\textwidth]{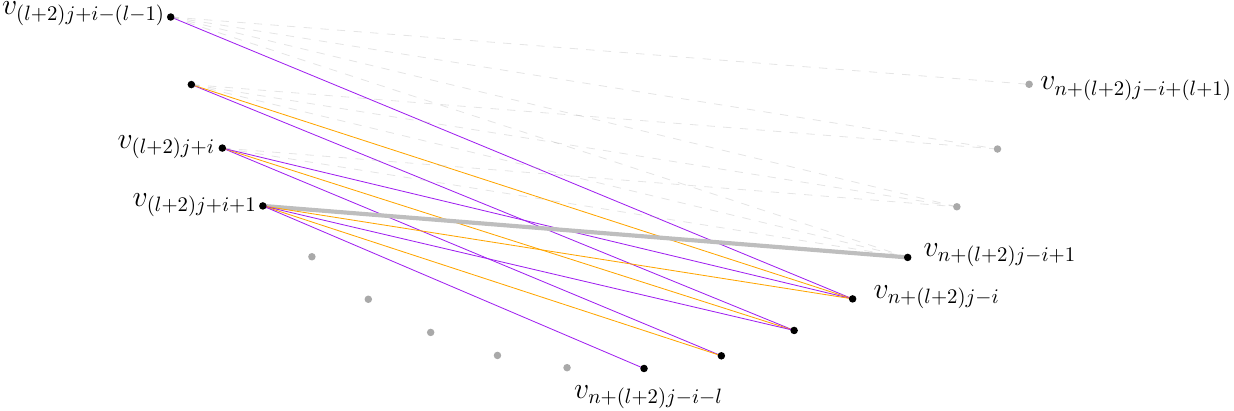}}
        \subcaption{\label{fig:k-defective-l-odd-n-even-B-fan}$\mathcal{B}_j$}
    \end{subfigure}
   \caption{Illustration for the crossings along the bold-drawn edges $(v_{(l+2)j+i-1}, v_{n+(l+2)j-i-1})$ and $(v_{(l+2)j+i+1}, v_{n+(l+2)j-i+1})$, respectively.}
    \label{fig:k-defective-l-odd-n-even-fan}
\end{figure}

For the former, we argue about the edges assigned to each defective stack  $\mathcal{A}_j$ with $j$ in $[0, \frac{n}{2(l+2)} -1]$ (see \cref{fig:k-defective-l-odd-n-even-A-fan}); symmetrically, one can argue for the edges assigned to each defective stack $\mathcal{B}_j$ with $j$ in $[0, \frac{n}{2(l+2)} -1]$ (see \cref{fig:k-defective-l-odd-n-even-B-fan}). In particular, no two gray edges assigned to $\mathcal{A}_j$ cross. Also, on the defective stack $\mathcal{A}_j$ each $\emph{root vertex}$ $v_{(l+2)j+i-1}$ is incident to at most $l+2$ edges (one gray, at most $\frac{l+1}{2}$ orange, and at most $\frac{l+1}{2}$ purple edges) that connect it to at most $l+2$ consecutive vertices. Among those, the edges $(v _{(l+2)j+i-1}, v_{n+(l+2)j-i-1})$ and $(v _{(l+2)j+i-1}, v_{n+(l+2)j-i-(l+2)})$  are the ones that receive the most crossings. We argue about the number of crossings for the former; the latter is symmetric. The edge $(v _{(l+2)j+i-1}, v_{n+(l+2)j-i-1})$ is crossed by $l$ edges rooted at $v _{(l+2)j+i-2}$, $l-1$ edges rooted at $v _{(l+2)j+i-3}$, and so on. This implies that the edge $(v _{(l+2)j+i-1}, v_{n+(l+2)j-i-1})$ is crossed $l+(l-1)+ \ldots + 1=\frac{l(l+1)}{2}$ times. Since for every real number $x$ it holds $\lfloor x \rfloor \leq x$, by setting $l = \left\lfloor \frac{-1 + \sqrt{8k+1}}{2} \right\rfloor$ to $\frac{l(l+1)}{2}$, we obtain that the defectiveness of the layout is $k$, as desired. This completes the proof that $\mathcal{L}_n$ is $k$-defective. 

To complete the proof, we show that no edge has been assigned to two distinct defective stacks. To see this, observe that the edges rooted at  $v_{(l+2)j+i-1}$ in the defective stacks $\mathcal{B}_{j-1}$ and $\mathcal{B}_{j+1}$ immediately precede or follows those at the defective stack $\mathcal{A}_j$, while the edges rooted at $v_{(l+2)j+i-1}$ in the defective stacks $\mathcal{A}_{j-1}$ and $\mathcal{A}_{j+1}$ immediately precede or follow those at the defective stack $\mathcal{B}_j$. This implies that no edge has been assigned to two defective stacks.

\paragraph*{Case 2.b: $l$ is odd and $n$ is odd.} 

In this case, our assignment scheme is slightly different. Namely, we set $\kappa=\frac{1}{2}(\frac{n}{l+2}+1)$ and $\lambda=\frac{1}{2}(\frac{n}{l+2}-1)$. First, we assign edges to the defective stacks $\mathcal{A}_0,\mathcal{A}_1,\ldots,\mathcal{A}_{\frac{1}{2}(\frac{n}{l+2}+1)}$, such that for each $j$ in $[0,\frac{1}{2}(\frac{n}{l+2}+1)]$ the defective stack $\mathcal{A}_j$ contains the following edges:

\noindent
\begin{tabular}{l@{\;\;}l@{\hspace{1ex}}r @{$ \;\; \leq \; i \; \leq \;\; $} l} 
\tikzcircle[defgray, fill=defgray]{2.8pt}       & $(v_{(l+2)j+i}, v_{n+(l+2)j-i-1})$     & $1$  & $\frac{n-1}{2}-1$ \\
\tikzcircle[deforange, fill=deforange]{2.8pt}     & $(v_{(l+2)j+i}, v_{n+(l+2)j-i-2s-2})$    & $-s$  & $\frac{n-1}{2}-s-2$, $s \in [0,\frac{l-1}{2}]$ \\
\tikzcircle[defpurple, fill=defpurple]{2.8pt}   & $(v_{(l+2)j+i}, v_{n+(l+2)j-i-2s-3})$  & $-s$  & $\frac{n-1}{2}-s-2$, $s \in [0,\frac{l-1}{2}]$ 
\end{tabular}

\medskip 
\noindent 
Finally, we assign edges to the defective stacks $\mathcal{B}_0,\mathcal{B}_1,\ldots,\mathcal{B}_{\frac{1}{2}(\frac{n}{l+2}-1)}$, such that for each $j$ in $[0,\frac{1}{2}(\frac{n}{l+2}-1)]$ the defective stack $\mathcal{B}_j$ contains the following edges:

\noindent
\begin{tabular}{l@{\;\;}l@{\hspace{1ex}}r @{$ \;\; \leq \; i \; \leq \;\; $} l} 
\tikzcircle[defgray, fill=defgray]{2.8pt}     & $(v_{(l+2)j+i+1}, v_{n+(l+2)j-i+3})$ & $2$  & $\frac{n-1}{2}$ \\
\tikzcircle[deforange, fill=deforange]{2.8pt} & $(v_{(l+2)j+i+1}, v_{n+(l+2)j-i-2s+1})$ & $1-s$  & $\frac{n-1}{2}-s-1$, $s \in [0,\frac{l-1}{2}]$\\
\tikzcircle[defpurple, fill=defpurple]{2.8pt} & $(v_{(l+2)j+i+1}, v_{n+(l+2)j-i-2s+2})$ & $2-s$  & $\frac{n-1}{2}-s$, $s \in [0,\frac{l-1}{2}]$ \\ 
\end{tabular}

\begin{figure}[h!]
    \begin{subfigure}{.48\textwidth}
        \centering
        {\includegraphics[page=1,width=\textwidth]{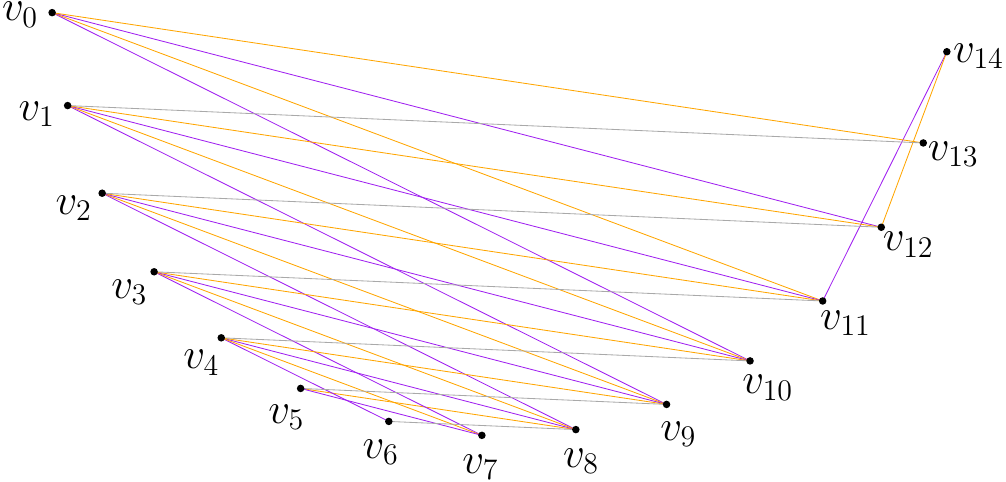}}
        \subcaption{\label{app:fig:k15-a0-k}$\mathcal{A}_0$}
    \end{subfigure}    
    \begin{subfigure}{.48\textwidth}
        \centering
        {\includegraphics[page=2,width=\textwidth]{figures/k-defective.pdf}}
        \subcaption{\label{app:fig:k15-a1-k}$\mathcal{A}_1$}
    \end{subfigure}
    
    \begin{subfigure}{.48\textwidth}
        \centering
        {\includegraphics[page=3,width=\textwidth]{figures/k-defective.pdf}}
        \subcaption{\label{app:fig:k15-b0-k}$\mathcal{B}_0$}
    \end{subfigure}
    \caption{A $k$-defective stack layout of $K_{15}$ with $\left\lceil \frac{n}{l+2} \right\rceil$ defective stacks, where $k=6$ and $l=3$.}
    \label{fig:k-defective-l-odd-n-odd}
\end{figure}

\medskip\noindent It follows that the total number of edges assigned to $\mathcal{A}_0,\mathcal{A}_1,\ldots,\mathcal{A}_{\frac{1}{2}(\frac{n}{l+2}+1)}$ equals:

\begin{equation*}
\begin{split}
\sum_{j=0}^{\frac{1}{2}(\frac{n}{l+2}+1)} \left( \frac{n-1}{2}-1 + \sum_{s=0}^{\frac{l-1}{2}} \left( \frac{n-1}{2}-s-2-(-s)+1 \right) + \sum_{s=0}^{\frac{l-1}{2}} \left( \frac{n-1}{2} - s -2 -(-s)+1 \right) \right)= \\
\sum_{j=0}^{\frac{1}{2}(\frac{n}{l+2}+1)} \left( \frac{n-1}{2}-1 + \sum_{s=0}^{\frac{l-1}{2}} \left( n-3 \right) \right) = \frac{n+l+2}{2(l+2)} \cdot \frac{(n-3)(l+2)}{2}= \frac{(n+l+2)(n-3)}{4}.
\end{split}
\end{equation*}

\medskip\noindent Accordingly, the total number of edges assigned to $\mathcal{B}_0,\mathcal{B}_1,\ldots,\mathcal{B}_{\frac{1}{2}(\frac{n}{l+2}-1)}$ equals:

\begin{equation*}
\begin{split}
\sum_{j=0}^{\frac{1}{2}(\frac{n}{l+2}-1)} \left( \frac{n-1}{2}-2 + \sum_{s=0}^{\frac{l-1}{2}} \left( \frac{n-1}{2}-s-1-(1-s)+1 \right) + \sum_{s=0}^{\frac{l-1}{2}} \left( \frac{n-1}{2}-s-(2-s)+1 \right) \right)= \\
\sum_{j=0}^{\frac{1}{2}(\frac{n}{l+2}-1)} \left( \frac{n-1}{2}-1 + \sum_{s=0}^{\frac{l-1}{2}} \left( n-3 \right) \right) = \frac{n-l-2}{2(l+2)} \cdot \frac{(n-3)(l+2)}{2}= \frac{(n-l-2)(n-3)}{4}.
\end{split}
\end{equation*}

\medskip 
Summing the number of edges assigned to $\mathcal{A}_0,\mathcal{A}_1,\ldots,\mathcal{A}_{\frac{1}{2}(\frac{n}{l+2}+1)}$ and $\mathcal{B}_0,\mathcal{B}_1,\ldots,\mathcal{B}_{\frac{1}{2}(\frac{n}{l+2}-1)}$, we obtain that the total number of edges that participate in the scheme, when both $n$ and $l$ are odd, is:
\begin{equation*}
\begin{split}
\frac{(n+l+2)(n-3)}{4} + \frac{(n-l-2)(n-3)}{4} = \frac{n-3}{4} \cdot (n+l+2+n-l-2) = \frac{n(n-3)}{2}.
\end{split}
\end{equation*}

The fact that $\mathcal{L}_n$ is $k$-defective and that no edge has been assigned to two distinct defective stacks can be proven following reasoning similar to the one of Case 2.a, which concludes the proof of the upper bound.
\end{proof}

\noindent We now turn our attention to $K_{n,n}$. The following is an immediate corollary of \cref{th:1-defective-sn-kn}.

\begin{corollary}\label{cor:1-defective-stack-knn-nonsep}
The $k$-defective stack number of $K_{n,n}$ in the separated setting is at least $\left\lceil \frac{n}{k+1}\right\rceil$ and at most $\left\lceil \frac{2n}{l+2}\right\rceil$, where $l {=} \left\lfloor \frac{-1 + \sqrt{8k+1}}{2} \right\rfloor$.
\end{corollary}

As in the separated setting, $k$-defective queue layouts can be turned into $k$-defective~stack~layouts by reversing the vertex order of one partition set; so, by~\cref{lem:knn-sep-1defect} the following~holds.

\begin{corollary}\label{cor:1-defective-stack-knn-sep}
The $1$-defective stack number of $K_{n,n}$ in the separated setting is at least $\left\lceil \frac{n}{2}\right\rceil$ and at most $\left\lceil \frac{2n}{3}\right\rceil$.    
\end{corollary}

\noindent In the non-separated setting, \Cref{th:1-defective-sn-kn,cor:1-defective-stack-knn-sep} give an upper bound of $\left\lceil \frac{2n}{3} \right\rceil$ on the $1$-defective stack number of $K_{n,n}$. We provide an improvement in the following theorem.
 
\begin{theorem}\label{thm:1-defective-stack-knn}
    The $1$-defective stack number of $K_{n,n}$ is at least $ \left\lceil \frac{n}{4} \right\rceil$ and at most $\left\lceil \frac{n}{2} \right\rceil$.  
\end{theorem}

\begin{proof}
For the upper bound, let $u_0, \ldots, u_{n-1}$ and $v_0,\ldots,v_{n-1}$ be the two parts of vertices of $K_{n,n}$. 
We assume that $n$ is even, and we prove that $K_{n,n}$ admits a $1$-defective stack layout $\mathcal{L}_n$ with $\frac{n}{2}$ $1$-defective stacks, in which the underlying order is $u_0 \prec v_0 \prec u_1 \prec v_1 \prec \ldots \prec u_{n-1} \prec v_{n-1}$ (see \cref{fig:Knn_nonsep}). Each stack of $\mathcal{L}_n$ consists of exactly $2n$ edges that are partitioned into two parts. More precisely, for each $p = 1, \ldots, \frac{n}{2}$, the first part of stack $p$ of $\mathcal{L}_n$  contains the following $2n-4p+3$ edges:

\noindent
\begin{tabular}{ l l l l}
\tikzcircle[defblue, fill=defblue]{2.8pt} $(u_i, v_{n+2p-3-i})$ & $ 2p-2 \leq i \leq \frac{n}{2}+p-2$  &~~~
\tikzcircle[defred, fill=defred]{2.8pt} $(u_i, v_{n+2p-2-i})$ & $ 2p-1 \leq i \leq \frac{n}{2}+p-1$\\
\tikzcircle[deforange, fill=deforange]{2.8pt} $(v_i, u_{n+2p-2-i})$ & $ 2p-1 \leq i \leq \frac{n}{2}+p-2$ &~~~
\tikzcircle[defgreen, fill=defgreen]{2.8pt} $(v_i, u_{n+2p-3-i})$ & $ 2p-2 \leq i \leq \frac{n}{2}+p-2$ 

\end{tabular}
\medskip

\noindent The second part of stack $p$ contains the following $4p-3$ edges:

\noindent
\begin{tabular}{ l l l l}
\tikzcircle[defcyan, fill=defcyan]{2.8pt} $(u_i, v_{2p-2-i})$ & $0 \leq i \leq p-1$ &~~~
\tikzcircle[defgray, fill=defgray]{2.8pt} $(u_i, v_{2p-3-i})$ & $0 \leq i \leq p-2$ \\
\tikzcircle[defpink, fill=defpink]{2.8pt} $(v_i, u_{2p-3-i})$ & $0 \leq i \leq p-2$  &~~~
\tikzcircle[defpurple, fill=defpurple]{2.8pt} $(v_i, u_{2p-2-i})$ & $0 \leq i \leq p-2$ 

\end{tabular}

\begin{figure}
    \centering
    \includegraphics[width=.85\textwidth]{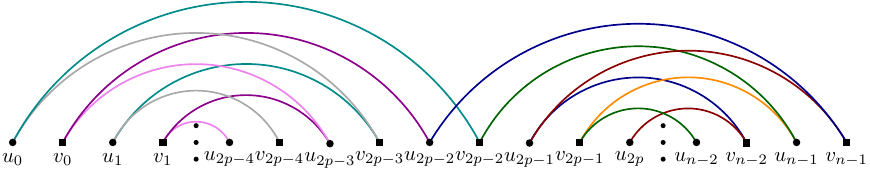}
    \caption{Illustration of the $p$-th stack of a $1$-defective stack layout of $K_{n,n}$ with $\left\lceil \frac{n}{2} \right\rceil$ stacks.}
    \label{fig:Knn_nonsep}
\end{figure}

\medskip \noindent Thus, in total $\mathcal{L}_n$ contains $ \sum_{p=1}^{\frac{n}{2}} (2n - 4p + 3 + 4p -3) = n^2$. By construction, each edge~in~$\mathcal{L}_n$ is crossed by at most one other. More precisely, for the edges of the first part, note that each blue edge crosses exactly one orange edge, each red edge crosses exactly one green edge, while no blue or orange edge crosses a red or green edge. Similarly, we argue about the edges of the second part (see \cref{fig:Knn_nonsep}). The outermost two edges of each part, namely $(u_0,v_{2p-2})$ and $(u_{2p-2},v_{n-1})$, are not involved in such crossings but cross each other. Since no further crossings between the two parts occur, this ensures that the obtained stack layout is $1$-defective. Since no two edges are assigned to the same stack, the proof follows. For the lower bound, we observe that every $1$-defective stack layout of $K_{n,n}$ contains a $1$-defective stack layout of $K_{\left\lceil n/2\right\rceil,\left\lceil n/2 \right\rceil}$ in the separated setting, which by \cref{cor:1-defective-stack-knn-sep} needs $\left\lceil\frac{n}{4}\right\rceil$ stacks.
\end{proof}

\begin{theorem}\label{th:k-defective-stack-knn}
The $k$-defective stack number of $K_{n,n}$ is at least $\left\lceil \frac{n}{2k+2} \right\rceil$ and at most~$\left\lceil \frac{n}{l} \right\rceil$, where $l = \sqrt{k}+1$.
\end{theorem}

\begin{proof}
As every $k$-defective stack layout of $K_{n,n}$ contains a corresponding one of $K_{\left\lceil\frac{n}{2}\right\rceil,\left\lceil\frac{n}{2}\right\rceil}$ in the separated setting, the proof of the lower bound follows from \cref{cor:1-defective-stack-knn-nonsep}.
The proof of the upper bound forms an extension of the one of \cref{thm:1-defective-stack-knn}.
We assume that $n$ is a multiple of $l$, and we prove that the complete bipartite graph $K_{n,n}$ on $2n$ vertices admits a $k$-defective stack layout $\mathcal{L}_n$ with $\frac{n}{l}$ defective stacks. Let $u_0 \prec v_0 \prec u_1 \prec v_1 \prec \ldots \prec u_{n-1} \prec v_{n-1}$ be the order of the vertices of $K_{n,n}$ in $\mathcal{L}_n$ and assume in the following that indices are taken$\mod n$. Each stack of $\mathcal{L}_n$ consists of exactly $nl$ edges that are partitioned into two parts. We proceed by considering cases in our proof, depending on whether each of $n$ and $l$ is even or odd. 

    \paragraph*{Case 1: $l$ is even.}

    Since $l$ is even, it follows that $n$ is even.
    
    \noindent For each $p=1,\ldots, \frac{n}{l}$ the first part of stack $p$ contains the following $nl-l^2p+\frac{l^2}{2}+\frac{l}{2}$ edges:
    \noindent
    \begin{tabular}{ l l l }
    \tikzcircle[defred, fill=defred]{2.8pt} $(u_i, v_{n+lp-l-1-i+2s})$ & $lp-l+2s \leq i \leq \frac{n}{2}+\frac{l}{2}p-\frac{l}{2}+s-1$ & $s\in [0,\frac{l}{2}-1]$ \\
    \tikzcircle[defblue, fill=defblue]{2.8pt} $(u_i, v_{n+lp-l-i+2s})$ & $lp-l+2s+1 \leq i \leq \frac{n}{2}+\frac{l}{2}p-\frac{l}{2}+s$ & $s\in [0,\frac{l}{2}-1]$ \\
    \tikzcircle[defgreen, fill=defgreen]{2.8pt} $(v_i, u_{n+lp-l-i+2s})$ & $lp-l+2s+1 \leq i \leq \frac{n}{2}+\frac{l}{2}p-\frac{l}{2}+s-1$ & $s\in [0,\frac{l}{2}-1] $ \\
     \tikzcircle[deforange, fill=deforange]{2.8pt} $(v_i, u_{n+lp-l-1-i+2s})$ & $lp-l+2s \leq i \leq \frac{n}{2}+\frac{l}{2}p-\frac{l}{2}+s-1$ & $s\in [0,\frac{l}{2}-1]$ \\
    \end{tabular}

    \medskip\noindent The second part of stack $p$ contains the following $l^2p-\frac{l^2}{2}-\frac{l}{2}$ edges:
    
    \noindent
    \begin{tabular}{ l l l }
    \tikzcircle[defcyan, fill=defcyan]{2.8pt} $(u_i, v_{lp-l-1-i+2s})$ & $0 \leq i \leq  \frac{l}{2}p-\frac{l}{2}+s-1$ & $s\in [0,\frac{l}{2}-1] $  \\
    \tikzcircle[defpurple, fill=defpurple]{2.8pt} $(u_i, v_{lp-l-i+2s})$ & $0 \leq i \leq \frac{l}{2}p-\frac{l}{2}+s$ & $s\in [0,\frac{l}{2}-1] $
    \\
    \tikzcircle[defpink, fill=defpink]{2.8pt} $(v_i, u_{lp-l-i+2s})$ & $0 \leq i \leq  \frac{l}{2}p-\frac{l}{2}+s-1$ & $s\in [0,\frac{l}{2}-1] $ \\
    \tikzcircle[defgray, fill=defgray]{2.8pt} $(v_i, u_{lp-l-1-i+2s})$ & $0 \leq i \leq  \frac{l}{2}p-\frac{l}{2}+s-1$ & $s\in [0,\frac{l}{2}-1] $ \\
    \end{tabular}

\paragraph*{Case 2.a: $l$ and $n$ are odd.}
In this case, we consider the cases in which $p$ is even or odd.

\noindent The first part of stack $p$ contains the following $nl-l^2p+\frac{l^2}{2}+\frac{l}{2}$ edges:

\[
\hat{p} =
\begin{cases}
\frac{p-1}{2}, & \text{if } p \text{ is odd} \\
\frac{p}{2}, & \text{if } p \text{ is even } 
\end{cases}
\]

\noindent

\begin{tabular}{ l l l }
    \tikzcircle[defred, fill=defred]{2.8pt} $(u_i, v_{n+lp-l-1-i+2s})$ & 
    $lp-l+2s \leq i \leq \frac{n+1}{2}+\frac{l+1}{2}p-\frac{l+1}{2}+s-1-\hat{p}$ & 
    $s\in [0,\frac{l-1}{2}]$ \\
    \tikzcircle[defblue, fill=defblue]{2.8pt} $(u_i, v_{n+lp-l-i+2s})$ & 
    $lp-l+2s+1 \leq i \leq \frac{n-1}{2}+\frac{l-1}{2}p-\frac{l-1}{2}+s+\hat{p}$ & 
    $s\in [0,\frac{l-1}{2}-1]$ \\
    \tikzcircle[deforange, fill=deforange]{2.8pt} $(v_i, u_{n+lp-l-1-i+2s})$ & 
    $lp-l+2s \leq i \leq \frac{n-1}{2}+\frac{l-1}{2}p-\frac{l-1}{2}+s-1+\hat{p}$ & 
    $s\in [0,\frac{l-1}{2}]$ \\
    \tikzcircle[defgreen, fill=defgreen]{2.8pt} $(v_i, u_{n+lp-l-i+2s})$ & 
    $lp-l+2s+1 \leq i \leq \frac{n+1}{2}+\frac{l+1}{2}p-\frac{l+1}{2}+s-1+\hat{p}$ & 
    $s\in [0,\frac{l-1}{2}-1] $ \\
\end{tabular}

\noindent The second part of stack $p$ contains the following $l^2p-\frac{l^2}{2}-\frac{l}{2}$ edges:

\[
\hat{p} =
\begin{cases}
\frac{p-1}{2}, & \text{if } p \text{ is odd} \\
\frac{p-2}{2}, & \text{if } p \text{ is even } 
\end{cases}
\]

\noindent
\begin{tabular}{ l l l }
\tikzcircle[defcyan, fill=defcyan]{2.8pt} $(u_i, v_{lp-l-1-i+2s})$ & $0 \leq i \leq \frac{l+1}{2}p-\frac{l+1}{2}+s-1-\hat{p}$ & $s\in [0,\frac{l-1}{2}] $  \\
\tikzcircle[defpurple, fill=defpurple]{2.8pt} $(u_i, v_{lp-l-i+2s})$ & $0 \leq i \leq  \frac{l-1}{2}p-\frac{l-1}{2}+s+\hat{p}$ & $s\in [0,\frac{l-1}{2}-1] $ \\
\tikzcircle[defpink, fill=defpink]{2.8pt} $(v_i, u_{lp-l-1-i+2s})$ & $0 \leq i \leq  \frac{l-1}{2}p-\frac{l-1}{2}+s-1+\hat{p}$ & $s\in [0,\frac{l-1}{2}] $ \\
\tikzcircle[defgray, fill=defgray]{2.8pt} $(v_i, u_{lp-l-i+2s})$ & $0 \leq i \leq  \frac{l+1}{2}p-\frac{l+1}{2}+s-1+\hat{p}$ & $s\in [0,\frac{l-1}{2}-1]$ \\
\end{tabular}

\paragraph*{Case 2.b: $l$ is odd and $n$ is even.}

\[
\hat{p} =
\begin{cases}
\frac{p-1}{2}, & \text{if } p \text{ is odd} \\
\frac{p-2}{2}, & \text{if } p \text{ is even } 
\end{cases}
\]

\noindent The first part of stack $p$ contains the following $nl-l^2p+\frac{l^2}{2}+\frac{l}{2}$ edges:

\noindent

\begin{tabular}{ l l l }
    \tikzcircle[defred, fill=defred]{2.8pt} $(u_i, v_{n+lp-l-1-i+2s})$ & 
    $lp-l+2s \leq i \leq \frac{n}{2}+ \frac{l+1}{2}p-\frac{l+1}{2}+s-1-\hat{p}$ & 
    $s\in [0,\frac{l-1}{2}]$ \\
    \tikzcircle[defblue, fill=defblue]{2.8pt} $(u_i, v_{n+lp-l-i+2s})$ & 
    $lp-l+2s+1 \leq i \leq  \frac{n}{2}+ \frac{l-1}{2}p-\frac{l-1}{2}+s+\hat{p}$ & 
    $s\in [0,\frac{l-1}{2}-1]$ \\
    \tikzcircle[deforange, fill=deforange]{2.8pt} $(v_i, u_{n+lp-l-1-i+2s})$ & 
    $lp-l+2s \leq i \leq \frac{n}{2}+\frac{l-1}{2}p-\frac{l-1}{2}+s-1+\hat{p}$ & 
    $s\in [0,\frac{l-1}{2}]$ \\
    \tikzcircle[defgreen, fill=defgreen]{2.8pt} $(v_i, u_{n+lp-l-i+2s})$ & 
    $lp-l+2s+1 \leq i \leq \frac{n}{2}+\frac{l+1}{2}p-\frac{l+1}{2}+s-1-\hat{p}$ & 
    $s\in [0,\frac{l-1}{2}-1] $ \\
\end{tabular}

\noindent The second part of stack $p$ contains the following $l^2p-\frac{l^2}{2}-\frac{l}{2}$ edges:

\noindent
\begin{tabular}{ l l l }
\tikzcircle[defcyan, fill=defcyan]{2.8pt} $(u_i, v_{lp-l-1-i+2s})$ & $0 \leq i \leq \frac{l+1}{2}p-\frac{l+1}{2}+s-1-\hat{p}$ & $s\in [0,\frac{l-1}{2}] $  \\
\tikzcircle[defpurple, fill=defpurple]{2.8pt} $(u_i, v_{lp-l-i+2s})$ & $0 \leq i \leq  \frac{l-1}{2}p-\frac{l-1}{2}+s+\hat{p}$ & $s\in [0,\frac{l-1}{2}-1] $ \\
\tikzcircle[defpink, fill=defpink]{2.8pt} $(v_i, u_{lp-l-1-i+2s})$ & $0 \leq i \leq  \frac{l-1}{2}p-\frac{l-1}{2}+s-1+\hat{p}$ & $s\in [0,\frac{l-1}{2}] $ \\
\tikzcircle[defgray, fill=defgray]{2.8pt} $(v_i, u_{lp-l-i+2s})$ & $0 \leq i \leq  \frac{l+1}{2}p-\frac{l+1}{2}+s-1-\hat{p}$ & $s\in [0,\frac{l-1}{2}-1]$ \\
\end{tabular}

\noindent Thus, in all cases, in total $\mathcal{L}_n$ contains $\sum_{p=1}^{\frac{n}{l}} \left(nl-l^2p+\frac{l^2}{2}+\frac{l}{2} + l^2p-\frac{l^2}{2}-\frac{l}{2} \right) = n^2$. By construction, each edge in $\mathcal{L}_n$ is crossed by at most $k$ other edges. 

Since no additional crossings occur between the two parts, this guarantees that the resulting stack layout is $k$-defective. Moreover, the proof is complete since no two edges share the same stack.
\end{proof}

\section{Additional Remarks and Open Research Directions}\label{se:final}
We have studied stack and queue layouts with defects from a combinatorial point of view. Precisely, we have given edge density results and bounds on the number of required queues/stacks, depending on the value of defectiveness. For this last type of bounds, we now report a summary of our main findings in comparison with the bounds known on the classic queue/stack number~(\cref{tbl:bounds}).

\subparagraph{Key findings for defective queue number.} 
\begin{description}
\item[F1.] Interestingly, the queue number of outerplanar graphs is not reduced if we allow defects, independent of the value $k$ of defectiveness. However, for $k=1$, the queue number of outer 1-planar graphs decreases by one unit, passing from 3 (for layouts without defects)~to~2.
\item[F2.] For planar graphs, 33 queues suffice for 1-defective layouts; recall that the best upper bound known for the queue number (without defects) of planar graphs is 42.
Note however that our lower bound for the queue number of 1-defective layouts is only 2; this large gap between lower and upper bound reflects the one known for the classic queue number and leaves space for improvements.
\item[F3.] About complete  and complete bipartite graphs (in the non-separated setting), 1-defective layouts reduce the number of queues by $33\%$ with respect to layouts without defects (obviously, the reduction asymptotically tends to $100\%$ for increasing values~of~defectiveness).
\end{description}

\subparagraph{Key findings for defective stack number.} 
\begin{description}
\item[F4.] Outer $k$-planar graphs can be realized with one stack if we allow defectiveness $k$. 
\item[F5.] For planar graphs, our lower bound for 1-defective stack number is $2$, while the upper bound is 4, inheriting the value for stack layouts without defects. This gap leaves space for future research. 
\item[F6.] Regarding complete graphs and complete bipartite graphs, the conclusions are similar to those for defective queue layouts (see Finding~F3).
\end{description}

Specific research directions that directly arise from our results include the following: 

\begin{description}
\item[OP1.] Finding tight upper bounds on the edge density of $k$-defective $h$-queue layouts and $k$-defective $h$-stack layouts, when $k,h > 1$ (see \cref{tb:edge-density}).

\item[OP2.] Filling the gap between lower and upper bounds on the $k$-defective stack and queue number for some of the studied classes of graphs, e.g., for planar graphs when $k=1$ (see \cref{tb:queue-stack-number-1}), and for $K_n$ or $K_{n,n}$ when $k>1$ (see \cref{tb:queue-stack-number-k}).

\item[OP3.] Extending the study to other meaningful classes of graphs, such as series-parallel graphs, planar 3-trees, planar bipartite graphs, and bounded-treewidth graphs.
\end{description}

\subsection{The Recognition Problem}\label{sse:recognition} Besides the study of combinatorial questions, establishing the complexity of recognizing graphs that admit $k$-defective $h$-queue/stack layouts (for different values of $k$ and $h$) is an interesting research direction. Below, we give some preliminary results for $h=1$. Note that the edge density bounds in \cref{tb:edge-density} provide a simple criterion for efficiently excluding graphs that do not admit $k$-defective $h$-stack/queue layouts, depending on the values of $k$ and $h$.

\medskip\noindent\textbf{Defective stack layouts.} 
The following result is a corollary of the characterization of \cref{th:charact-k-defective-stack} and a result in~\cite{DBLP:conf/compgeom/KobayashiO025}; it proves that recognizing which graphs admit a $k$-defective $1$-stack layout can be solved in quasi-polynomial time.

\begin{corollary}\label{th:recognition-k-defective-stack}
Deciding whether a graph has $k$-defective stack number~$1$, for $k \geq 1$, can be done in $O(2^{k \log k}n^{3k+O(1)})$ time.
\end{corollary}
In view of \Cref{th:charact-k-defective-stack}, we further note that in~\cite{DBLP:conf/swat/0001C00SU024} the authors study the number of edges one needs to remove from a graph with a given vertex order so that the result is a $k$-defective $h$-stack layout, as each stack in their approach is required to be (outer) $k$-planar.

\medskip\noindent\textbf{Defective queue layouts.}
In contrast to $k$-defective $1$-stack layouts, we do not provide an answer about the complexity of recognizing graphs with $k$-defective queue number one. However, we give a characterization that generalizes the one of graphs with queue number one, for which the recognition problem is NP-complete~\cite{DBLP:journals/siamcomp/HeathR92}. This suggests that the recognition of $k$-defective graphs with queue number one may also be NP-complete for fixed values~of~$k \geq 1$.

An \emph{arched level layout} $\Gamma$ of a graph $G$ is a drawing such that (see \cref{fi:arched-level-1-planar-layout}): 
$(i)$ the vertices of $G$ are points of a set of $\ell \geq 1$ horizontal lines, called \emph{levels}, numbered from $1$ to $\ell$ from top to bottom; $(ii)$ each edge of $G$ connects vertices either on the same level or on consecutive levels;
$(iii)$ an edge connecting vertices on consecutive levels is drawn as a straight-line segment; $(iv)$ an edge connecting two vertices of the same level $i$ ($1 \leq i \leq \ell$) is called \emph{arched edge} and is drawn as a simple Jordan arc $\gamma$ such that:
($a$) $\gamma$ intersects~level~$i$ exactly once (other than at its end-points) to the left of every vertex of level $i$;
($b$) all the vertices of every level $j<i$ are 
inside the closed region bounded by $\gamma$ and the portion of level $i$ that connects the endpoints of $\gamma$.
If $\Gamma$ has no crossings, it is an \emph{arched level planar layout}; if every edge of $\Gamma$ has at most $k$ crossings, it is an \emph{arched level $k$-planar layout}. 

\begin{figure}[tb]
    \centering
    \begin{subfigure}{.4\textwidth}
        \centering
        {\includegraphics[page=7, width=.75\textwidth]{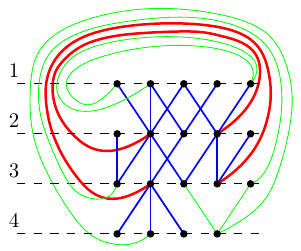}}
        \subcaption{\label{fi:level-1-planar}}
    \end{subfigure}
    \hfil
    \begin{subfigure}{.4\textwidth}
        \centering
        {\includegraphics[page=8, width=.75\textwidth]{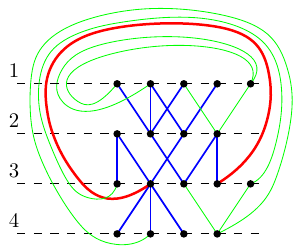}}
        \subcaption{\label{fi:arched-level-1-planar}}
    \end{subfigure}   
    \caption{Two examples of arched level $1$-planar layouts. The layout (a) is not nice, because the red edge in level $2$ is involved in a crossing even after removing all  green edges. The layout~(b)~is~nice.}
    \label{fi:arched-level-1-planar-layout}
\end{figure}

For an arched level $k$-planar layout, we classify the edges of $\Gamma$ as follows (see \cref{fi:arched-level-1-planar-layout}): $(i)$ an edge $(u,v)$ is \emph{green} if $v$ is the rightmost vertex of a level $i$ ($1 \leq i \leq \ell$) and $u$ is either on level $i-1$ or on level $i$; $(ii)$ an edge $(u,v)$ is \emph{blue} if $u$ and $v$ are on different levels and $(u,v)$ is not green; $(iii)$ an edge $(u,v)$ is \emph{red} if $u$ and $v$ are on the same level and $(u,v)$ is not green (that is, none of its end-vertices is the rightmost of this level).

An arched level $k$-planar layout is \emph{nice} if once all green edges are removed, there is no red edge with $k$ crossings (see \cref{fi:arched-level-1-planar}). Our characterization is stated in \cref{th:charact-nice}, whose proof however is based on the following auxiliary lemma.

\begin{lemma}\label{le:charact-arched-level-1-planar}
For every $k \geq 0$, a graph $G$ has $k$-defective queue number $1$ if and only if it is an arched level $k$-planar graph.
\end{lemma}

\begin{proof}
Let $\Gamma$ be an arched level layout of $G$ and let $v_{i,1}, v_{i,2}, \dots, v_{i,r_i}$ be the vertices of level $i$ ($1 \leq i \leq \ell$) in the left-to-right order they appear along level $i$. 
The linear order $v_{1,1},v_{1,2},\dots,v_{1,r_1},v_{2,1},v_{2,2},\dots,v_{2,r_2},\dots, v_{\ell,1},v_{\ell,2},\dots,v_{\ell,r_\ell}$ of the vertices of $V$ is called the \emph{induced order} of $\Gamma$ and is denoted as $\prec_{\Gamma}$.  

We claim that two edges nest in $\prec_{\Gamma}$ if and only if they cross in $\Gamma$. Namely, suppose that two edges $e_1=(a,d)$ and $e_2=(b,c)$ cross in $\Gamma$. We have three cases:
$(i)$ Both edges are arched edges; $(ii)$ both edges are not arched edges; $(iii)$ one edge is arched and the other one is not. In all the three cases, a crossing of $e_1$ and $e_2$ in $\Gamma$ implies that the two edges nest in $\prec_{\Gamma}$, as illustrated in \Cref{fi:configurations}.

Suppose now that  two edges $e_1=(a,d)$ and $e_2=(b,c)$ nest in $\prec_{\Gamma}$, and assume that  $a \prec_{\Gamma} b \prec_{\Gamma} c \prec_{\Gamma} d$. There are four cases: $(i)$ $a,b,c,d$ all appear on the same level $i$; $(ii)$ $a$ is on level $i$ and $b,c,d$ are on level $i+1$;  $(iii)$ $a, b$ are on level $i$ and $c,d$ are on level $i+1$; $(iv)$ $a,b,c$ are on level $i$, and $d$ is on level $i+1$.
In all cases, drawing the edges as required (i.e., straight-line when the end-vertices are on different levels and as curves when they are arched) creates a crossing; see \Cref{fi:configurations}.

Assume that $G$ admits an arched level $k$-planar layout $\Gamma$ and consider the induced order $\prec_{\Gamma}$ of $\Gamma$. By the claim above, two edges nest in $\prec_{\Gamma}$ if and only if they cross in $\Gamma$; since $\Gamma$ is $k$-planar, every edge is in at most $k$ nesting relationships, i.e., $\prec_{\Gamma}$ defines a $k$-defective $1$-queue layout.

\begin{figure}[tb]
    \centering
    \begin{subfigure}{.48\textwidth}
        \centering
        {\includegraphics[page=4,width=.8\textwidth]{configurations}
        \subcaption{\label{fi:configurations-d}}}
    \end{subfigure}
    \hfil
    \begin{subfigure}{.48\textwidth}
        \centering
        {\includegraphics[page=1,width=.8\textwidth]{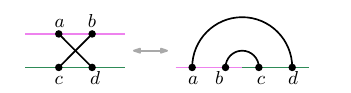}
        \subcaption{\label{fi:configurations-a}}}
    \end{subfigure}
    \hfil
    \begin{subfigure}{.48\textwidth}
        \centering
        {\includegraphics[page=2,width=.8\textwidth]{configurations}\subcaption{\label{fi:configurations-b}}}
    \end{subfigure}
    \hfil
    \begin{subfigure}{.48\textwidth}
        \centering
        {\includegraphics[page=3,width=.8\textwidth]{configurations}
        \subcaption{\label{fi:configurations-c}}}
    \end{subfigure}
    \hfil
    \caption{Crossings in arched level $k$-planar layouts and nestings in $k$-defective $1$-queue layouts.}
    \label{fi:configurations}
\end{figure}

Assume now that $G$ admits a $k$-defective $1$-queue layout $\mathcal{L}$.  Let $u_1 \prec u_2 \prec \dots \prec u_n$ be the linear order of $\mathcal{L}$. One may define several different arched level layouts whose induced order is $\prec$. Since $\mathcal{L}$ has defectiveness $k$, by the claim above every such an arched level layout is $k$-planar. 

We conclude by observing that at least one arched level layout whose induced order is $\prec$ exists. Namely, to obtain such an arched level layout it is sufficient to arbitrarily partition the vertices of $\mathcal{L}$ into $\ell \geq 1$ groups (each defining a different level) so that every group contains vertices that are consecutive according to $\prec$ and any two adjacent vertices are either in the same group or in consecutive groups. Notice that a partition with one or two groups necessarily satisfies the conditions above, and therefore at least one such arched level layout exists\footnote{Another way of partitioning the vertices of $\mathcal{L}$ in levels is described in \cite{DBLP:journals/siamcomp/HeathR92}.}.
\end{proof}

\begin{theorem}\label{th:charact-nice}
For every $k \geq 0$, a graph $G$ has $k$-defective queue number $1$ if and only if it admits a nice arched level $k$-planar layout.
\end{theorem}

\begin{proof}
By \cref{le:charact-arched-level-1-planar}, if $G$ admits a nice arched level $k$-planar layout then it admits a $k$-defective $1$-queue layout.

Conversely, suppose that $G$ admits a $k$-defective $1$-queue layout $\mathcal L$. 
We define a leveling of the vertices of $G$ by partitioning them into sets $L_1, L_2,\dots,L_{\ell}$, with $\ell \geq 1$, as follows (note that this leveling is analogous to the one defined in~\cite{DBLP:journals/siamcomp/HeathR92}). Let $\prec$ be the linear order of $\mathcal L$ and let $v_1,v_2,\dots,v_n$ be the vertices of $G$ ordered according to $\prec$. 
Assume first that $G$ is connected.
Set $L_1$ contains only vertex $v_1$. Assume that $L_i$ contains vertices $v_j,v_{j+1},\dots,v_{j+|L_i|-1}$. Let $N_i$ be the set of vertices that are adjacent to a vertex of $L_i$ and that follow $v_{j+|L_i|-1}$ in $\prec$; let $v_g$ be the vertex with highest index in $N_i$. Set $L_{i+1}$ contains all vertices $v_{j+|L_i|}, v_{j+|L_i|+1},\dots,v_g$.

If $G$ is not connected, let $C_1, C_2,\dots, C_h$ be the connected components of $G$. We apply the procedure above to every connected component defining a collection of sets with different indices for each connected component. More precisely, if $L_j$ is the set with the largest index in the leveling of $C_1,C_2,\dots,C_i$, then the procedure is applied to $C_{i+1}$ starting with set $L_{j+1}$. With this leveling, if the vertices of a set $L_j$ are not connected to any vertex of the set $L_{j+1}$, then $L_{j+1}$ is the first set of a connected component and therefore it contains only one vertex.  

By definition every edge of $G$ connects vertices that are either in the same set or~in~sets~with consecutive indices. 
We construct an arched level layout $\Gamma$ by placing the vertices of~set~$L_i$~on level $i$ (for $i=1,2,\dots,\ell$) and by ordering them from left to right according~to~$\prec$. We color the edges of~$\Gamma$ green, blue, and red, as described above. We claim that $\Gamma$ is $k$-planar and nice.

Since $\mathcal L$ is $k$-defective  and since the induced order of $\Gamma$ is $\prec$, $\Gamma$ is $k$-planar. Indeed, as observed in the proof of~\cref{le:charact-arched-level-1-planar}, every arched level layout whose induced order is the linear order of a $k$-defective $1$-queue layout is $k$-planar. 

We now prove that $\Gamma$ is nice. Suppose it is not. Then, after removing the green edges, at least one red edge $(u,v)$ has $k$ crossings.
Let $i$ be the level of $u$ and $v$ ($1 \leq i \leq \ell$) and let $u \prec v$. Vertex $v$ is not the rightmost vertex of level $i$ as otherwise $(u,v)$ would be green and it would have been removed. Let $v'$ be the rightmost vertex of level $i$. Two cases are possible. 

\begin{figure}[tb]
    \centering
    {\includegraphics[page=1,width=.3\textwidth]{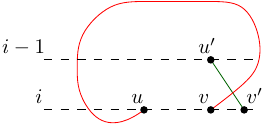}}

    \caption{Illustration for Case 1 of \cref{th:charact-nice}.}
    \label{fi:nice}
\end{figure}

\medskip\noindent{\bf Case 1:} There exists at least one edge connecting a vertex of level $i-1$ to a vertex of level~$i$. By the definition of the set $L_{i}$, there is an edge in $\Gamma$ connecting a vertex $u'$ of level $i-1$ to the rightmost vertex of level $i$, i.e., to $v'$. The (green) edge $(u',v')$ crosses the edge $(u,v)$; indeed, $v'$ is to the right of $v$ and the edge $(u,v)$ crosses level $i-1$ to the right of every vertex by property $(b)$ of the drawing of arched edges. See also~\cref{fi:nice}. Thus, $(u,v)$ has $k+1$ crossings in $\Gamma$ -- a contradiction to the fact that $\Gamma$ is $k$-planar.

\medskip\noindent{\bf Case 2:} There is no edge connecting a vertex of level $i-1$ to a vertex of level $i$. As observed above, this can happen only if level $i$ is the first level of a connected component and therefore, it contains only one vertex -- a contradiction to the existence of $(u,v)$.
\end{proof}

Heath and Rosenberg~\cite{DBLP:journals/siamcomp/HeathR92} characterized the graphs that have queue number $1$ and called them \emph{arched level planar} graphs. We remark that, although our definition of nice arched level $0$-planar layouts slightly differs from the definition of arched level planar layouts described in~\cite{DBLP:journals/siamcomp/HeathR92}, \cref{th:charact-nice} implies that the family of arched level planar graphs coincides with the family of graphs admitting a nice arched level $0$-planar layout.

\medskip
Establishing the complexity of the following problems remains open.
\begin{description}
\item[OP4.] Testing whether a graph admits a $k$-defective $h$-queue layout, for~$k \geq 1$ and $h \geq 1$.
\item[OP5.] Testing whether a graph admits a $k$-defective $h$-stack layout, for $k \geq 1$ and~$h \geq 2$.
\end{description}

\subsection{Further research directions}\label{sse:further-directions}
Our work also opens broader research directions beyond those already mentioned.

\begin{description}
\item[OP6.] Our research can naturally extend to other types of linear layouts (e.g., deques or riques) or even to mixed settings. More in general, one can study defective variants of related problems, such as defective track layouts and their implications for 3D graph~drawings. 

\item[OP7.] We considered defects obtained by imposing restrictions on the degree of the conflict graph. One could study other types of defects, e.g., conflict graphs of bounded diameter.
\end{description}

\bibliography{literature}

\end{document}